\documentclass[11pt,twoside,english]{article}
\usepackage[T1]{fontenc}
\usepackage[latin9]{inputenc}
\usepackage[a4paper]{geometry}
\usepackage{color}
\usepackage{babel}
\usepackage{amsmath}
\usepackage{amsthm}
\usepackage{amssymb}
\usepackage{graphicx}
\usepackage{setspace}
\usepackage[authoryear]{natbib}
\usepackage[unicode=true,pdfusetitle,
 bookmarks=true,bookmarksnumbered=false,bookmarksopen=false,
 breaklinks=false,pdfborder={0 0 0},pdfborderstyle={},backref=false,colorlinks=true]
 {hyperref}
\hypersetup{
 citecolor=blue, linkcolor=blue, urlcolor=blue}

\makeatletter

\providecommand{\tabularnewline}{\\}

\numberwithin{equation}{section}
\numberwithin{figure}{section}
\theoremstyle{remark}
\newtheorem{rem}{\protect\remarkname}
\theoremstyle{definition}
\newtheorem{defn}{\protect\definitionname}
\theoremstyle{definition}
 \newtheorem{example}{\protect\examplename}
\theoremstyle{plain}
\newtheorem{prop}{\protect\propositionname}
\theoremstyle{plain}
\newtheorem{cor}{\protect\corollaryname}
\theoremstyle{remark}
\newtheorem{claim}{\protect\claimname}
\theoremstyle{remark}
\newtheorem*{rem*}{\protect\remarkname}
\theoremstyle{plain}
\newtheorem{lem}{\protect\lemmaname}

\usepackage{lmodern}
\usepackage[T1]{fontenc}
\usepackage{chngcntr}
\counterwithout{figure}{section}
\usepackage{color}
\definecolor{brown}{RGB}{120,60,10}

\makeatother

\providecommand{\claimname}{Claim}
\providecommand{\corollaryname}{Corollary}
\providecommand{\definitionname}{Definition}
\providecommand{\examplename}{Example}
\providecommand{\lemmaname}{Lemma}
\providecommand{\propositionname}{Proposition}
\providecommand{\remarkname}{Remark}

\begin{document}
\title{Biased-Belief Equilibrium\thanks{The authors are very grateful to the anonymous referees for very helpful
comments and suggestions.} }
\author{Yuval Heller\thanks{Department of Economics, Bar Ilan University, Israel. \protect\href{mailto:yuval.heller@biu.ac.il}{yuval.heller@biu.ac.il}.
URL: \protect\href{https://sites.google.com/site/yuval26/}{https://sites.google.com/site/yuval26/}.
The author is grateful to the \emph{European Research Council} for
its financial support (ERC starting grant \#677057).} \and Eyal Winter\thanks{Lancaster University, Management School, and The Hebrew University,
Department of Economics. \protect\href{mailto:mseyal@mscc.huji.ac.il\%20}{mseyal@mscc.huji.ac.il }.
URL: \protect\href{http://www.ma.huji.ac.il/~mseyal/}{http://www.ma.huji.ac.il/$\sim$mseyal/}.
The author is grateful to the German-Israeli Foundation for Scientific
Research and Google for their financial support.}}

\maketitle
\noindent Final pre-print of a paper published in \emph{American Economic
Journal: Microeconomics} 12(2), 1--40, 2020.
\begin{abstract}
We investigate how distorted, yet structured, beliefs can persist
in strategic situations. Specifically, we study two-player games in
which each player is endowed with a biased-belief function that represents
the discrepancy between a player\textquoteright s beliefs about the
opponent's strategy and the actual strategy. Our equilibrium condition
requires that: (1) each player choose a best-response strategy to
his distorted belief about the opponent's strategy, and (2) the distortion
functions form best responses to one another. We obtain sharp predictions
and novel insights into the set of stable outcomes and their supporting
stable biases in various classes of games.


\end{abstract}

\section{Introduction}

Standard models of equilibrium behavior attribute rationality to players
at two different levels: beliefs and actions (see, e.g., \citealp{aumann1995epistemic}).
Players are assumed to behave as if they form correct beliefs about
the opponents' behavior, and they choose actions that maximize their
utility given the beliefs that they hold. Much of the literature in
behavioral and experimental economics that documents violations of
the assumption that players have correct beliefs ascribes these violations
to cognitive limitations. However, in interactive environments where
one person\textquoteright s beliefs affect other persons\textquoteright{}
actions, belief distortions are not arbitrary, and they may arise
to serve some strategic purpose. 

In this paper we investigate how distorted, yet structured, beliefs
can persist in strategic situations. Our basic assumption here is
that distorted beliefs can persist because they offer a strategic
advantage to those who hold them even when these beliefs are wrong.
More specifically, players often hold distorted beliefs as a form
of commitment device that affects the behavior of their counterparts.
The precise cognitive process that is responsible for the formation
of beliefs is complex, and it is beyond the scope of this paper to
outline it. We believe, however, that in addition to analytic assessment
of evidence, preferences in the form of desires, fears, and other
emotions contribute to the process and, to an extent, facilitate belief
biases. If the evidence is unambiguous and decisive, or if the consequence
of belief distortion is detrimental to the player\textquoteright s
welfare, preferences may play less of a role and learning may work
to calibrate beliefs to reality. But when beliefs are biased in ways
that favor their holders by affecting the behavior of their counterparts,
learning can actually reinforce biases rather than diminish them. 

\paragraph{Biased Beliefs}

Standard equilibrium notions in game theory draw a clear line between
preferences and beliefs. The former are exogenous and fixed; the latter
can be amended through Bayesian updating but are not allowed to be
affected by preferences. However, phenomena such as wishful thinking
(see, e.g., \citealp{babad1991wishful}) and overconfidence (see,
e.g., \citealp{forbes2005some,barber2001boys,malmendier2005ceo,heller2014overconfidence}),
where beliefs are tilted toward what their holder desires reality
to be, suggest that in real life, beliefs and preferences can intermingle,
and that biased beliefs may be persistent. Similarly, belief rigidity
and belief polarization (see, e.g., \citealp{lord1979biased,ross1982shortcomings})
refer to situations in which two people with conflicting prior beliefs
each strengthen their beliefs in response to observing the same data.
The parties\textquoteright{} aversion to depart from their original
beliefs can also be regarded as a form of interaction between preferences
and beliefs. 

It is easy to see how the belief biases described above can have strategic
benefits in interactive situations. Wishful thinking and optimism
can facilitate cooperation in interactions that require mutual trust.
Overconfidence can deter competitors, and belief rigidity can allow
an agent to support a credible threat. An important objective of our
analysis is to identify the strategic environments that support biases
such as wishful thinking  as part of equilibrium behavior. It is
worthwhile to note that individuals are not the only ones susceptible
to strategically motivated belief biases. Governments are prone to
be affected by such biases as well. The Bush administration\textquoteright s
unsubstantiated confidence in Saddam Hussein\textquoteright s possession
of \textquotedblleft weapons of mass destruction\textquotedblright{}
prior to the Second Gulf War and the vast discrepancy between Israeli
and US intelligence assessments of Iran\textquoteright s nuclear intentions
prior to the signing of the Iran nuclear deal can be interpreted as
strategically motivated belief distortion.\footnote{There are other possible interpretations of these controversial real-life
examples. In a dynamic real-life setup it is hard to have access to
agents' private information, and therefore it is very difficult to
achieve direct empirical evidence for persistent biased beliefs. There
are a few lab experiments that elicit subjects' beliefs (using monetary
incentives and proper scoring rules) about the expected behavior of
the opponent. \citet{nyarko2002experimental} demonstrate that the
elicited forecasts of subjects about about the opponents' future behavior
substantially differ from the empirical play of opponents in the past.
\citet{palfrey2009eliciting} present evidence that forecasts by players
(about the opponent's behavior in a simple two-player game) are significantly
different from the forecasts of external observers. Moreover, the
players' forecasts are systematically biased, and significantly less
accurate than the forecasts of the external observers.} 

Belief biases in strategic environments are also connected to self-interest
biases regarding moral and ethical standards. \citet{Babcock-Lewonstin}
had participants in a lab experiment negotiate a deal between a plaintiff
and a defendant in a court case. When they asked participants to make
predictions about the outcome of the real court case the authors found
a significant belief divergence depending on the role participants
were assigned to in the negotiations. A similar moral hypocrisy was
revealed by \citet{rustichini2014moral} who showed that subjects'
subjective judgments regarding fairness in bargaining depended on
the bargaining power they were assigned in the experiment. 

A different body of empirical evidence consistent with strategic beliefs
is offered by the psychiatric literature on \textquotedblleft depressive
realism\textquotedblright{} (e.g., \citealp{dobson1989conceptual}).
This literature compares probabilistic assessments conveyed by psychiatrically
healthy people with those suffering from clinical depression. Participants
in both categories were requested to assess the likelihood of experiencing
negative or positive events in both public and private setups. Comparing
subjects\textquoteright{} answers with the objective probabilities
of these events revealed that in a public setup clinically depressed
individuals were more realistic than their healthy counterparts for
both types of events. The apparent belief bias among healthy individuals
can be reasonably attributed to the strategic component of beliefs.
Mood disorders negatively affect strategic reasoning (\citealp{inoue2004deficiency}),
which, to a certain extent, may diminish strategic belief distortion
among clinically depressed individuals relative to their healthy counterparts. 

For biased beliefs to yield a strategic advantage to the agents holding
them, it is essential that (1) agents be committed to follow their
biased beliefs, and (2) agents best-reply to the perceived behavior
induced by their counterparts\textquoteright{} biases (both on and
off the equilibrium path). For the sake of tractability, we shall
avoid formalizing a concrete dynamic model that describes how biased
beliefs are formed, and how agents credibly commit to these biased
beliefs. Instead, we shall adopt a static approach by imposing equilibrium
conditions on the agents\textquoteright{} beliefs and the opponents'
interpretation of their beliefs. (We discuss our modeling approach
and its evolutionary interpretation in Section \ref{subsec:Discussion-of-the},
and we present a formal evolutionary foundation in Appendix \ref{sec:evol-interpetation-of-BBE}.)
This static approach is consistent with a large part of the literature
on endogenous preferences (see, e.g., the literature cited below).
Nevertheless, we mention a few mechanisms that can facilitate these
processes and turn biased beliefs into a credible commitment device.
\begin{enumerate}
\item Refraining from accessing or using biased sources of information,
e.g., subscribing to a newspaper with a specific political orientation,
consulting biased experts, and reading Facebook's personalized news
feeds, which are typically biased due to friends who hold similar
beliefs.
\item Passionately following a religion, a moral principle, or an ideology
that has belief implications on human behavior. 
\item Possessing personality traits that have implications on beliefs (e.g.,
narcissism or naivety).
\end{enumerate}
The mechanisms described above are likely not only to induce belief
biases, but also to generate signals sent to the player\textquoteright s
counterparts about these biases with a certain degree of verifiability.
These mechanisms, the signals they induce, and their interpretation
are the main forces that facilitate biased-belief equilibrium. 

\paragraph{Solution Concept}

Our notion of biased-belief equilibrium (henceforth, BBE) uses a two-stage
paradigm. In the first stage each player is endowed with a biased-belief
function. This function represents the discrepancy between a player\textquoteright s
beliefs about the strategy profile of other players and the actual
profile. In the second stage the players play the biased game induced
by their distortion functions, in which each player chooses a best-reply
strategy to his biased belief about the opponent's strategy (the chosen
strategy profile is referred to as the equilibrium outcome). Finally,
our equilibrium condition requires that the distortion functions not
be arbitrary, but form best replies to one another.

If one of the players deviates to being endowed with a different biased-belief
function, then there might be multiple Nash equilibria in the new
biased game induced by this deviation. Our weak notion (\emph{weak
BBE}) requires the deviator to be outperformed in at least one equilibrium
of the new biased game. Our strong notion (\emph{strong BBE}) requires
(1) each agent to have a \emph{monotone} biased belief, according
to which he assigns a higher probability to his opponent playing a
certain strategy than this probability actually is, and (2) a deviator
to be outperformed in \emph{all} Nash equilibria of the new biased
game. Our main notion, \emph{BBE},\emph{ }lies in between these two
notions, and it requires (1) each agent to have a monotone biased
belief, and (2) the deviator to be outperformed in at least one plausible
Nash equilibrium of the new biased game, where we rule out implausible
Nash equilibria in which the non-deviator behaves differently even
though he does not observe any change in the deviator's perceived
strategy.

In Section 2.5 we present our main evolutionary interpretation of
our solution concept, according to which the endowed biased beliefs
are the result of an evolutionary process of social learning (the
interpretation is formalized in Appendix \ref{sec:evol-interpetation-of-BBE}).
In addition, we present an alternative, delegation interpretation
of the model (which is formalized in Appendix \ref{sec:Principal-Agent}).

\paragraph{Nash Equilibrium and BBE}

We begin our analysis by studying the relations between BBE outcomes
and Nash equilibria. We show that any Nash equilibrium can be implemented
as the outcome of a BBE, though in some cases this requires that the
players have biased beliefs that are accurate on the equilibrium path,
but that they be blind to some deviations of the opponent off the
equilibrium path. This, in particular, implies that every game admits
a BBE. Next, we show that introducing biased beliefs does not change
the set of equilibrium outcomes in games in which at least one of
the players has a dominant action. By contrast, BBE admits non-Nash
behavior in most other games.

\paragraph{Main Results}

Our main results show that the notion of BBE induces substantial predictive
power in various classes of interval games. In these classes of games
the strategy of each player is a number in a bounded interval, where
a higher strategy (interpreted as a higher investment) induces a higher
payoff for the opponent. We begin by characterizing the set of BBE
in games with strategic complements (\citealp{bulow1985multimarket}),
such as price competition with differentiated goods (Example \ref{exam-Price-competition-with-differianted}),
input games (Example \ref{exam-partnership-game} in Appendix \ref{subsec:Examples-of-Games-complements}),
and stag hunt games (Example \ref{exa-stag-hunt-analysis} in Appendix
\ref{subsec:Examples-of-Games-complements}). We show three key properties
of any BBE: (1) \emph{overinvestment}: the strategy of each agent
is (weakly) higher than the best reply to the opponent's (real) strategy,
(2) \emph{ruling out bad outcomes}: both players invest more than
their investments in the worst Nash equilibrium of the underlying
game, and (3) \emph{wishful thinking}: each agent perceives his opponent
as investing (weakly) more than the opponent's real investment.

Next, we characterize the set of BBE in games with strategic substitutes,
such as Cournot competitions (Example \ref{exam-Cournot}) and hawk-dove
games (Example \ref{Exam-Cornot-Stackelberg} in Appendix \ref{subsec:Hawk-dove-game}).
We show three key properties of any BBE: (1) \emph{underinvestment}:
the strategy of each agent is (weakly) higher than the best reply
to the opponent's (real) strategy, (2) \emph{ruling out excellent
outcomes}: at least one of the players invests less than his investments
in one of the Nash equilibria of the underlying game, and (3) \emph{wishful
thinking}: each agent perceives his opponent as investing (weakly)
more than the opponent's real investment.

Finally, we characterize the set of BBE in a class of games (which
are less common in economic interactions), in which the strategy of
player 1 is a complement of player 2's strategy, while the strategy
of player 2 is a substitute of player 1's strategy (e.g., duopolistic
competition in which one firm chooses its quantity while the opposing
firm chooses its price (\citealp{singh1984price}), and various classes
of asymmetric contests (\citealp{dixit1987strategic})). We show that
in this class of games agents present \emph{pessimism} in any BBE\emph{:}
each agent perceives his opponent as investing (weakly) less than
the opponent's real investment.

\paragraph{Additional Results}

Our next result shows an interesting class of BBE that exist in all
games. We say that a strategy is undominated Stackelberg if it maximizes
a player\textquoteright s payoff in a setup in which the player can
commit to an undominated strategy, and his opponent reacts by best-replying
to this strategy. We show that every game admits a BBE in which one
of the players is ``strategically stubborn'' in the sense of having
a constant belief about the opponent's strategy, and always playing
his undominated Stackelberg strategy, while the opponent is ``rational''
in the sense of having undistorted beliefs and best-replying to the
player's true strategy.

Section \ref{sec:folk-theorem-results} shows that unless one imposes
both requirements on the definition of a BBE, namely, monotonicity
and ruling out implausible equilibria, then the set of BBE outcomes
is very large in various classes of games. Specifically, Proposition
\ref{pro-monotone-BBE-outcomes-finite-games} shows that for a large
class of finite games, a strategy profile is a monotone weak BBE iff
(1) no player uses a strictly dominated strategy, and (2) the payoff
of each player is above the minmax payoff of the player in a setup
in which both players are restricted to choose only undominated strategies
(i.e., strategies that are not strictly dominated). Proposition \ref{pro-interval-strong-continous-folk}
shows a similar folk theorem result for non-monotone strong BBE in
a large class of interval games.

\paragraph{Empirical Predictions}

Our main results imply two empirical predictions. First, they suggest
that efficient (non-Nash equilibrium) outcomes are easier to support
in games with strategic complements, relative to games with strategic
substitutes. This prediction is consistent with the experimental findings
of \citet{potters2009cooperation}, which show that there is significantly
more cooperation in games with strategic complements than in the case
of strategic substitutes.

Our second empirical prediction is that wishful thinking is strategically
stable in many common environments, though some (less common) strategic
interactions may induce pessimism. This empirical prediction is consistent
with the experimental evidence that people tend to present wishful
thinking, while the presented level of wishful thinking may substantially
differ between various environments; see, e.g., \citet{babad1991wishful,budescu1995relationship,doi:10.1080/13546789508256906}
and \citet{mayraz2013wishful}.

\paragraph{Structure}

The structure of this paper is as follows. We discuss the related
literature in Section \ref{sec:Related-Literature-and}. Section \ref{sec:Model}
describes the model. In Section \ref{sec:Relations-with-Nash} we
analyze the relations between BBE and Nash equilibria. Section \ref{sec:Monotone-Games-and-wishful}
defines games with strategic complements/substitutes and wishful thinking.
We analyze these games and present our main results in Section \ref{sec:Main-Results}.
In Section \ref{sec:Additional-Results} we present additional results:
(1) the relation between BBE and strategies played by a Stackelberg
leader, and (2) folk theorem results when relaxing the definition
of BBE. We conclude in Section \ref{sec:Discussion}. All the appendices
of the paper appear in the online supplementary material. Appendix
\ref{sec:Additional-Examples} presents various interesting examples.
We formally present the evolutionary interpretation of our solution
concept in Appendix \ref{sec:evol-interpetation-of-BBE}, and the
delegation interpretation in Appendix \ref{sec:Principal-Agent}.
Appendix \ref{sec:Discontinuous-Biased-Beliefs} relaxes the assumption
that biased beliefs have to be continuous. Appendix \ref{sec:Partial-Observability}
shows how to extend our results to a setup with partial observability.
Appendix \ref{sec:Proofs} presents our formal proofs.

\section{Related Literature and Contributions\label{sec:Related-Literature-and}}

Our paper aims at making a contribution to the behavioral game theory
literature. Much of this literature concerns behavioral equilibrium
concepts that depart from the framework of Nash equilibrium by introducing
weaker rationality conditions. This has been done primarily at the
level of preferences (e.g., \citealp{GuethYaari1992Explaining,fehr1999theory,bolton2000erc,acemoglu2001evolution,heifetz2004evolutionary,DekelElyEtAl2007Evolution,heifetz2007dynamic,friedman2009equilibrium,Herold_Kuzmics_2009,heller2016rule,winter2017mental}).
But it has also been done at the level of beliefs (e.g., \citealp{geanakoplos1989psychological,rabin1993incorporating,battigalli2007guilt,attanasi2008survey,battigalli2009dynamic,battigalli2017frustration,gannon2017evolutionary}).
This latter literature deals with belief-dependent preferences, and
focuses primarily on the way players\textquoteright{} beliefs about
the intentions of others affect their preferences and behavior. 

Our equilibrium concept also operates on beliefs rather than preferences
but is based on an inherently different approach. Preferences in our
model are not affected by beliefs but beliefs are biased in a way
that serves players\textquoteright{} strategic purposes. Our analysis
of biased belief goes beyond characterizing equilibrium outcomes.
An additional important objective is to identify the belief biases
that support these equilibrium outcomes in different strategic environments.
Central to our analysis are belief-distortion properties, such as
wishful thinking and pessimism, that sustain BBE in different strategic
environments.

The existing literature has presented various prominent solution concepts
that assume that players have distorted beliefs. Some examples include
models of level-\emph{k} and cognitive hierarchy (see, e.g., \citealp{stahl1994experimental,nagel1995unraveling,costa2001cognition,camerer2004cognitive}),
analogy-based expectation equilibrium (\citealp{jehiel2005analogy}),
cursed equilibrium (\citealp{eyster2005cursed}), and Berk-Nash equilibrium
(\citealp{esponda2016berk}). These equilibrium notions have been
helpful in understanding strategic behavior in various setups, and
yet these notions pose a conceptual challenge to our understanding
of the persistence of distorted beliefs, even in view of the empirical
evidence for such persistence. If players can infer the truth ex post
why don\textquoteright t they calibrate their beliefs toward reality?
Much of the literature presenting such models points to cognitive
limitations as the source of this rigidity. Our model and analysis
offer an additional perspective to this issue by suggesting that belief
biases that yield a strategic advantage in the long run are likely
to emerge in equilibrium. In this sense our approach can be viewed
as providing a tool to explain why some cognitive limitations persist
while others do not (see Example \ref{exam-traveler-dilemma-1} in
Appendix \ref{sec:Additional-Examples}, in which we show how level-1
behavior can be supported as part of a BBE outcome in the traveler's
dilemma).

Our notion of BBE is related to the notion of conjectural equilibrium
(\citealp[originally written in 1988]{battigalli1997conjectural})
insofar as both solution concepts relax the Nash equilibrium's requirement
that beliefs need to be consistent with actual play (while still requiring
that an agent's action has to be optimal given the agent's belief).
A conjectural equilibrium is defined in an environment in which players
do not observe each other's actions but rather observe signals of
each other's actions, according to an exogenous feedback correspondence\emph{.
}In a conjectural equilibrium each player best replies to his belief
about the opponent's action, and this belief is required to be consistent
with the signal observed by the player. There are two key structural
differences between a BBE and a conjectural equilibrium. First, a
BBE is defined in an environment in which there is no exogenous feedback
correspondence; rather, the feedback correspondence is implicitly
defined as part of the solution concept by the agents' biased-belief
functions. These biased-belief functions are not restricted by a consistency
requirement with respect to an exogenous feedback mechanism, but rather
they are are restricted by the requirement that each biased-belief
function has to be a best reply against the opponent's biased belief.
The second structural difference is that while a BBE describes what
would be the agent's belief for \emph{any} feasible action of the
opponent, a conjectural equilibrium describes only the agent's belief
about the equilibrium action of the opponent.

Despite these structural differences, it is interesting to discuss
relations between the equilibrium behavior induced by each solution
concept, i.e., the relations between a BBE outcome and a conjectural
equilibrium outcome. Without restricting the feedback correspondence,
the notion of conjectural equilibrium is rather broad (it rules out
only strictly dominated strategies), and, accordingly, any BBE outcome
is a conjectural equilibrium outcome. \citeauthor{fudenberg1993self}'s
\citeyearpar{fudenberg1993self} notion of self-confirming equilibrium
deals with extensive-form games, and refines conjectural equilibrium
by requiring that the feedback correspondence is the one in which
each player observes the opponent's realized actions (but does not
observe the opponent's behavior off the equilibrium path). In the
setup of two-player one-shot games, which is the focus of the present
paper, the set of self-confirming equilibria coincides with the set
of Nash equilibria (whereas the set of BBE outcomes is broader and
includes non-Nash outcomes). Another refinement of conjectural equilibrium
is the rationalizable conjectural equilibrium (\citealp{rubinstein1994rationalizable};
the notion has been generalized to games with structural uncertainty
in \citealp{esponda2013rationalizable}). This concept requires that
the agents' beliefs be consistent with the common knowledge that all
players maximize utility given their signals. There is no inclusion
relation between the set of BBE outcomes and the set of rationalizable
conjectural equilibrium outcomes. Specifically, in games with a unique
rationalizable action profile, such as price competitions with differentiated
goods and Cournot competitions, the unique rationalizable conjectural
equilibrium outcome is the Nash equilibrium (for any feedback correspondence),
while the set of BBE outcomes is substantially larger (see Examples
\ref{exam-Price-competition-with-differianted} and \ref{exam-Cournot}).
By contrast, in games such as stag hunt and hawk\textendash dove,
when the feedback correspondence is non-informative any action profile
is a conjectural equilibrium outcome, while the set of BBE outcomes
is much more restricted (see Examples \ref{exa-stag-hunt-analysis}
and \ref{exam-The-hawk-dove-r} in Appendix \ref{sec:Additional-Examples}).

\section{Model\label{sec:Model} }

\subsection{Underlying Game}

Let $i\in\left\{ 1,2\right\} $ be an index used to refer to one of
the players in a two-player game, and let $j$ be an index referring
to the opponent. Let $G=\left(S,\pi\right)$ be a normal-form two-player
game (henceforth, \emph{game}), where $S=\left(S_{1},S_{2}\right)$
and each $S_{i}$ is a convex compact set of strategies. Specifically,
we focus on two cases:
\begin{enumerate}
\item \emph{Finite games}: Each $S_{i}$ is a simplex over a finite set
of pure actions, where each strategy corresponds to a mixed action
(i.e., $A_{i}$ is a finite set of pure actions, and $S_{i}=\Delta\left(A_{i}\right)$),
and the von Neumann\textendash Morgenstern payoff function is linear
with respect to the mixing probability.
\item \emph{Interval games}: Each $S_{i}$ is a bounded interval in $\mathbb{R}$
(e.g., each player chooses a real number representing quantity, price,
or effort). 
\end{enumerate}
We denote by $\pi=\left(\pi_{1},\pi_{2}\right)$ players' payoff functions;
i.e., $\pi_{i}:S\rightarrow\mathbb{R}$ is a function assigning each
player a payoff for each strategy profile. We use $s_{i}$ to refer
to a typical strategy of player $i$. We assume each payoff function
$\pi_{i}\left(s_{i},s_{j}\right)$ to be continuously twice differentiable
in both parameters and weakly concave in the first parameter ($s_{i}$). 

Let $BR$ (resp., $BR^{-1}$) denote the (inverse) best-reply correspondence;
i.e., 
\[
BR\left(s_{i}\right)=argmax_{s_{j}\in S_{j}}\left(\pi_{j}\left(s_{i},s{}_{j}\right)\right)
\]
 is the set of best replies against strategy $s_{i}\in S_{i}$, and
\[
BR^{-1}\left(s_{i}\right)=\left\{ s_{j}\in S_{j}|s_{i}\in BR\left(s_{j}\right)\right\} 
\]
 is the set of strategies for which $s_{i}$ is a best reply against
them.

In a finite game, we use $a_{i}\in A_{i}$ to denote also the degenerate
mixed action that assigns mass one to $a_{i}$. When the set of actions
of a player is given as an ordered set $A_{i}=\left(a_{i}^{1},a_{i}^{2},...,a_{i}^{n}\right)$,
we identify a mixed action with a vector $s_{i}=\left(\alpha_{1},\alpha_{2},...,\alpha_{n}\right)$,
where $0\leq\alpha_{k}=s_{i}\left(a_{i}^{k}\right)$ for each $1\leq k\leq n$,
and $\sum_{k}\alpha_{k}=1$. Given two strategies $s_{i},s'_{i}\in S_{j}$
and $\alpha\in\left[0,1\right]$, let $\alpha\cdot s_{i}+\left(1-\alpha\right)\cdot s'_{i}$
be the mixture of the two strategies: $\left(\alpha\cdot s_{i}+\left(1-\alpha\right)\cdot s'_{i}\right)\left(a_{i}\right)=\alpha\cdot s_{i}\left(a_{i}\right)+\left(1-\alpha\right)\cdot s'_{i}\left(a_{i}\right)$.

When there are two (ordered) actions for each player (say, $A_{i}=\left\{ c_{i},d_{i}\right\} $),
we identify a mixed action $s_{i}$ with the probability it assigns
to the first pure action $s_{i}\left(c_{i}\right)$, and we identify
the set of strategies $S_{i}$ with the interval $\left[0,1\right]$.
Thus, a game with two actions for each player can be captured both
as a finite game and as an interval game.

\subsection{Biased-Belief Function}

We start here with the definition of biased-belief functions that
describe how players' beliefs are distorted. A\emph{ biased belief}
$\psi_{i}:S_{j}\rightarrow S_{j}$ is a \emph{continuous} function
that assigns to each strategy of the opponent, a (possibly distorted)
belief about the opponent's play. That is, if the opponent plays $s_{j}$,
then player $i$ believes that the opponent plays $\psi_{i}\left(s_{j}\right)$.
We call $s_{j}$ the opponent's real strategy, and we call $\psi_{i}\left(s_{j}\right)$
the opponent's perceived (or biased) strategy. Formally, the continuity
requirement is that if $\left(s_{j,n}\right)_{n}\underset{n\rightarrow\infty}{\rightarrow}s_{j}$,
then $\left(\psi_{i}\left(s_{j,n}\right)\right)_{n}\underset{n\rightarrow\infty}{\rightarrow}\psi_{i}\left(s_{j}\right)$
(where in a finite game, we say that $\left(s_{j,n}\right)_{n}\underset{n\rightarrow\infty}{\rightarrow}s_{j}$
iff $\left(s_{j,n}\left(a\right)\right)_{n}\underset{n\rightarrow\infty}{\rightarrow}s_{j}\left(a\right)$
for each action $a$).
\begin{rem}
Two reasons motivate us to require that a biased belief be continuous:
(1) continuity implies that each biased game (defined below) admits
a Nash equilibrium, which allows us to simplify the definition of
BBE, and (2) continuity reflects a plausible restriction that a small
change in the opponent's strategy should induce a small change in
the perceived strategy. In Appendix \ref{sec:Discontinuous-Biased-Beliefs}
we present an alternative (and somewhat more complicated) definition
of a BBE that relaxes the assumption that biased beliefs must be continuous,
and we show that all the BBE characterized in the results of the paper
remain BBE when we allow deviators to use discontinuous biased beliefs. 
\end{rem}
We say that a biased belief $\psi_{i}:S_{j}\rightarrow S_{j}$ is\emph{
monotone} if:
\begin{enumerate}
\item In interval games: $s_{j}\geq s'_{j}$ implies $\psi_{i}\left(s_{j}\right)\geq\psi_{i}\left(s'_{j}\right)$
for each strategy $s_{j}\in S_{j}$.
\item In finite games: If the opponent plays $a_{j}$ more often, while
keeping the same proportion of playing the remaining actions, then
the perceived probability that the opponent plays any other action
weakly decreases (which implies, in particular, that the perceived
probability that the opponent plays $a_{j}$ weakly increases); that
is, 
\[
\left(\psi_{i}\left(\left(1-\alpha\right)\cdot s_{j}+\alpha\cdot a_{j}\right)\right)\left(a'_{j}\right)\leq\left(\psi_{i}\left(s_{j}\right)\right)\left(a'_{j}\right)
\]
 for each $\alpha\in\left[0,1\right]$, each action $a_{j}\in A_{j}$,
each action $a'_{j}\neq a_{j}$, and each strategy $s_{j}\in\Delta\left(A_{j}\right)$.
In particular, when the game has two actions for each player, a biased
belief $\psi_{i}$ is\emph{ }monotone iff $\psi_{i}$ is weakly increasing
in $\alpha_{j}$; i.e., $\alpha_{j}\geq\alpha'_{j}$ implies that
$\psi_{i}\left(\alpha_{j}\right)\geq\psi_{i}\left(\alpha'_{j}\right)$.
\end{enumerate}
Monotone biased beliefs reflect a plausible restriction on the distortion
of agents, namely, that if the opponent changes his real strategy
in some direction, the agent captures the direction of the change
correctly, but may have the wrong perception about the magnitude of
the change. 

Let $I_{d}$ be the undistorted (identity) function, i.e., $I_{d}\left(s\right)=s$
for each strategy $s$. A biased belief $\psi$ is \emph{blind }if
the perceived opponent's strategy is independent of the opponent's
real strategy, i.e., if\emph{ $\psi\left(s_{j}\right)=\psi\left(s_{j}'\right)$
}for each \emph{$s_{j},s_{j}'\in S_{j}$. }With a slight abuse of
notation we use $s_{i}$ to denote also the blind biased belief $\psi_{j}$
that is always equal to $s_{i}$.

\subsection{Biased Game}

An underlying game and a profile of biased beliefs jointly induce
a biased game in which the (biased) payoff of each player is determined
by the perceived strategy of the opponent. Formally:
\begin{defn}
Given an underlying game $G=\left(S,\pi\right)$ and a profile of
biased beliefs $\left(\psi_{i},\psi_{j}\right)$, let the \emph{biased
game} \emph{$G_{\psi}=\left(S,\psi\circ\pi\right)$ }be defined as
the game with the following payoff function $\left(\psi\circ\pi\right)_{i}:S_{i}\times S_{j}\rightarrow\mathbb{R}$
for each player $i$: 
\[
\left(\psi\circ\pi\right)_{i}\left(s_{i},s_{j}\right)=\pi_{i}\left(s_{i},\psi_{i}\left(s_{j}\right)\right).
\]
\end{defn}
A Nash equilibrium of a biased game is defined in the standard way.
Formally, a pair of strategies $s^{*}=\left(s_{1}^{*},s_{2}^{*}\right)$
is a Nash equilibrium of a biased game $G_{\psi}=\left(S,\psi\circ\pi\right)$,
if each $s_{i}^{*}$\emph{ }is a best reply against the perceived
strategy of the opponent, i.e., 
\[
s_{i}^{*}=argmax_{s_{i}\in S_{i}}\left(\pi_{i}\left(s_{i},\psi_{i}\left(s_{j}^{*}\right)\right)\right).
\]

Let\emph{ $NE\left(G_{\psi}\right)\subseteq S_{1}\times S_{2}$ }denote
the set of all Nash equilibria of the biased game \emph{$G_{\psi}$. }

Observe that the set of strategies of a biased game is convex and
compact, and the payoff function $\left(\psi\circ\pi\right)_{i}:S_{i}\times S_{j}\rightarrow\mathbb{R}$
is weakly concave in the first parameter and continuous in both parameters.
This implies (due to a standard application of Kakutani's fixed-point
theorem) that each biased game \emph{$G_{\psi}$ }admits a Nash equilibrium
(i.e., $NE\left(G_{\left(\psi'_{i},\psi_{j}^{*}\right)}\right)\neq\emptyset$.) 

\subsection{Weak and Strong BBE}

We are now ready to define our equilibrium concept. A weak biased-belief
equilibrium (abbr. weak BBE) is a pair consisting of a profile of
biased beliefs and a profile of strategies, such that: (1) each strategy
is a best reply to the perceived strategy of the opponent, and (2)
each biased belief is a best reply to the opponent's biased belief,
in the sense that any agent who chooses a different biased-belief
function is outperformed in at least one equilibrium in the new biased
game (relative to the agent's payoff in the original equilibrium).
Formally:
\begin{defn}
\label{def-weak-BBE}A\emph{ weak BBE} is a pair $\left(\psi^{*},s^{*}\right)$,
where $\psi^{*}=\left(\psi_{1}^{*},\psi_{2}^{*}\right)$ is a profile
of biased beliefs and $s^{*}=\left(s_{1}^{*},s_{2}^{*}\right)$ is
a profile of strategies satisfying: (1) $\left(s_{i}^{*},s_{j}^{*}\right)\in NE\left(G_{\psi^{*}}\right)$,
and (2) for each player $i$ and each biased belief $\psi_{i}'$,
there exists a strategy profile $\left(s'_{i},s_{j}'\right)\in NE\left(G_{\left(\psi'_{i},\psi_{j}^{*}\right)}\right)$,
such that the following inequality holds: $\pi_{i}\left(s'_{i},s_{j}'\right)\leq\pi_{i}\left(s_{i}^{*},s_{j}^{*}\right)$. 
\end{defn}
The notion of weak BBE is arguably too permissive because it allows
incumbents: (1) to have implausible non-monotone beliefs, and (2)
to outperform the deviators in a single Nash equilibrium of the biased
game (while, possibly, the incumbents are outperformed by the deviators
in many other equilibria). Proposition \ref{pro-monotone-BBE-outcomes-finite-games}
(in Section \ref{subsec:Folk-Theorem-Result:-Monotone}) demonstrates
that this single Nash equilibrium, in which the deviators are outperformed,
may be implausible due to allowing the incumbents to ``discriminate''
against the deviators, even though the  deviators exhibit exactly
the same perceived behavior as the rest of the population.

The more restrictive refinement of strong BBE requires that (1) incumbents
have monotone beliefs, and (2) deviators who choose a different biased-belief
function be outperformed in \emph{all} equilibria of the induced biased
game. Formally:
\begin{defn}
\label{def-strong-BBE}A\emph{ weak BBE} $\left(\psi^{*},s^{*}\right)$
is a\emph{ strong} \emph{BBE} if (1) each biased function $\psi_{i}^{*}$
is monotone, and (2) the inequality $\pi_{i}\left(s'_{i},s_{j}'\right)\leq\pi_{i}\left(s_{i}^{*},s_{j}^{*}\right)$
holds for every player $i$, every biased belief $\psi'_{i}$, and
every strategy profile $\left(s'_{i},s_{j}'\right)\in NE\left(G_{\left(\psi'_{i},\psi_{j}^{*}\right)}\right)$.
\end{defn}

\subsection{BBE}

Finite games typically induce multiple Nash equilibria. This is often
the case also with respect to biased games. This suggests that the
refinement of strong BBE may be too restrictive, as there are are
potentially many Nash equilibria of many biased games, and the requirement
of the deviators being outperformed in all these equilibria might
be too demanding. Our main solution concept, BBE, lies in between
weak BBE and strong BBE. 

In a BBE, the deviator is required to be outperformed in at least
one \emph{plausible} equilibrium of the new biased game. Roughly speaking,
in a plausible equilibrium of the new biased game induced by a deviation
of player \emph{i} to a different biased belief, player $j$ is allowed
to choose a new strategy only if he distinguishes between \emph{i}'s
original strategy and \emph{i}'s new strategy. More precisely, implausible
equilibria are defined as follows. We say that a Nash equilibrium
of a biased game induced by a deviation of player $i$ is implausible
if (1) player $i$'s strategy is perceived by the non-deviating player
$j$ as coinciding with player $i$'s original strategy, (2) player
$j$ plays differently relative to his original strategy, and (3)
player $j$ playing his original strategy induces an equilibrium of
the biased game. That is, implausible equilibria are those in which
the non-deviating player $j$ plays differently against a deviator
even though player $j$ has no reason to do so: player $j$ does not
observe any change in player $i$'s behavior, and player $j$'s original
behavior remains an equilibrium of the biased game. Formally:
\begin{defn}
\label{def-implausible}Given weak BBE $\left(\psi^{*},s^{*}\right)$,
deviating player $i$, and biased belief $\psi'_{i}$, we say that
a Nash equilibrium of the biased game $\left(s'_{i},s'_{j}\right)\in NE\left(G_{\left(\psi'_{i},\psi_{j}^{*}\right)}\right)$
is \emph{implausible} if: (1) $\psi_{j}^{*}\left(s'_{i}\right)=\psi_{j}^{*}\left(s^{*}{}_{i}\right)$,
(2) $s_{j}^{*}\neq s'_{j}$, and (3) $\left(s'_{i},s_{j}^{*}\right)\in NE\left(G_{\left(\psi'_{i},\psi_{j}^{*}\right)}\right)$.
An equilibrium is \emph{plausible} if it is not implausible. Let $PNE\left(G_{\left(\psi'_{i},\psi_{j}^{*}\right)}\right)$
be the set of all plausible equilibria of the biased game $G_{\left(\psi'_{i},\psi_{j}^{*}\right)}$.
\end{defn}
Note that it is immediate from Definition \ref{def-implausible} and
the nonemptiness of $NE\left(G_{\left(\psi'_{i},\psi_{j}^{*}\right)}\right)$
that $PNE\left(G_{\left(\psi'_{i},\psi_{j}^{*}\right)}\right)$ is
nonempty.
\begin{defn}
\label{Def-A-biased-belief-equilibrium}Weak BBE $\left(\psi^{*},s^{*}\right)$
is \emph{a BBE} if (1) each biased function $\psi_{i}^{*}$ is monotone,
and (2) for each player $i$ and each biased belief $\psi_{i}'$,
there exists a plausible Nash equilibrium $\left(s'_{i},s_{j}'\right)\in PNE\left(G_{\left(\psi'_{i},\psi_{j}^{*}\right)}\right)$,
such that $\pi_{i}\left(s'_{i},s_{j}'\right)\leq\pi_{i}\left(s_{i}^{*},s_{j}^{*}\right)$.
\end{defn}
A strategy profile $s^{*}=\left(s_{1}^{*},s_{2}^{*}\right)$ is a\emph{
(resp., strong, weak) BBE outcome }if there exists a profile of biased
beliefs $\psi^{*}=\left(\psi_{1}^{*},\psi_{2}^{*}\right)$ such that
$\left(\psi^{*},s^{*}\right)$ is a (resp., strong, weak) BBE. In
this case we say that the biased belief $\psi^{*}$ supports (or implements)
the outcome $s^{*}$. 

\subsection{Discussion of the Model\label{subsec:Discussion-of-the}}

\paragraph{Evolutionary/Learning Interpretation}

Biases can emerge in a learning process that reinforces biases that
yield a strategic advantage to their holders. Specifically, we interpret
a BBE to be a reduced-form solution concept capturing the essential
features of an evolutionary process of cultural or social learning.
Our methodology follows the extensive literature that studies the
stability of endogenous preferences using the ``indirect evolutionary
approach'' (see, e.g., \citealp{GuethYaari1992Explaining,guth1995evolutionary,Fershtman_Weiss_1998,Dufwenberg_Guth_1999_EvoandDel,Kockesen2000,guttman2003repeated,Guth_Napel2006,heifetz2007maximize,friedman2009equilibrium,Herold_Kuzmics_2009,Alger_Weibull_HomoMoralis,heller2015coevolution}).
We apply this modeling approach to the study of endogenous biased
beliefs in a setup in which biased beliefs induce behavior, behavior
determines \textquotedblleft success,\textquotedblright{} and success
regulates the evolution of biased beliefs.

In Appendix \ref{sec:evol-interpetation-of-BBE} we formally adapt
the definition of a stable population state from \citet{DekelElyEtAl2007Evolution}
to our setup, and show that the adapted definition is equivalent to
a strong BBE. In what follows we briefly and informally present our
evolutionary interpretation. Consider two large populations of agents:
agents who play the role of player 1, and agents who play the role
of player 2. In each round agents from each population are randomly
matched to play a two-person game against opponents from the other
population. Each agent in each population is endowed with a biased-belief
function. For simplicity, we focus on ``homogeneous'' populations,
in which all agents in the population have the same monotone biased-belief
function. Agents distort their perception about the behavior of the
agents in the other population according to their endowed biased-belief
functions, and they play a Nash equilibrium of the biased game. 

With small probability a few agents (``mutants'') in one of the
populations (say, population 1) may be endowed with a different biased-belief
function due to a random error or experimentation. We assume that
agents of population 2 observe whether their opponents are mutants
or not, and that the agents of population 2 and the mutants of population
1 gradually adapt their play against each other into an equilibrium
of the new biased game. Note that a dynamic adaptation into playing
a Nash equilibrium of the biased game requires agents of population
2 to know the perceived strategy currently being played by the mutants
of population 1, but the agents do not need to know the biased beliefs
of the mutants of population 1.

Finally, we assume that the total ``success'' (fitness) of agents
is monotonically influenced by their (unbiased) payoff in the underlying
game, and that there is a slow process in which the composition of
the population evolves. This slow process might be the result of a
slow flow of new agents who join the population. Each new agent randomly
chooses one of the incumbents in his own population as a ``mentor''
(and mimics the mentor's biased belief), where the probabilities are
such that agents with higher fitness are more likely to be chosen
as mentors. If the original population state is not a BBE, it implies
that there are mutants who outperform the remaining incumbents in
their own population, which in turn implies that the original population
state is not stable, as new agents are likely to mimic more successful
mutants. By contrast, if the original population state is a BBE, it
implies that for any mutant there is a new equilibrium in which the
mutants are weakly outperformed relative to the incumbents of their
own population, and this can allow the BBE to remain a stable state
(as illustrated in the detailed example in Appendix \ref{subsec:Illustration-of-the}).

\paragraph{Variants of the Solution Concept }

\emph{The main solution concept we use in the paper is BBE}. In Section
\ref{sec:folk-theorem-results} we demonstrate that unless one applies
both requirements of Definition \ref{Def-A-biased-belief-equilibrium},
namely, monotonicity and ruling out implausible equilibria, then the
set of BBE is very large (folk theorem results), and some of the biased
beliefs that support some of these equilibria seem implausible. The
intuition for the monotonicity requirement is quite straightforward
(ruling out peculiar biased beliefs in which an opponent who deviates
to play a higher strategy is perceived as deviating to play a lower
strategy). The second requirement rules out implausible equilibria
in which a player responds to his opponent's deviation in spite of
not being able to perceive it

In what follows we sketch a dynamic justification for the second requirement
of ruling out implausible equilibria (following the evolutionary interpretation
described above). Consider a BBE $\left(\left(\psi_{1}^{*},\psi_{2}^{*}\right),\left(s_{1}^{*},s_{2}^{*}\right)\right)$.
Assume that both $\left(s'_{1},s_{2}'\right)$ and $\left(s'_{1},s_{2}^{*}\right)$
are Nash equilibria of the biased game $G_{\left(\psi'_{1},\psi_{2}^{*}\right)}$.
In what follows, we briefly, and informally, explain why $\left(s'_{1},s_{2}'\right)$
is not a plausible equilibrium of the new biased game (and, thus,
why it is ruled out in the definition of BBE). Consider a deviation
of some agents in the population playing in the role of player $1$
to having the biased belief $\psi'_{1}$. Following this deviation,
strategy $s_{1}^{*}$ might not be a best reply to the perceived strategy
of player 2 (i.e., $s_{1}^{*}\not\in BR\left(\psi'_{1}\left(s_{2}^{*}\right)\right)$)
and, as a result, the deviating agents might change their strategy
to $s'_{1}$, which is a best reply to the perceived strategy of player
2 (i.e., $s'_{1}\in BR\left(\psi'_{1}\left(s_{2}^{*}\right)\right)$).
The current strategy profile $\left(s'_{1},s_{2}^{*}\right)$ is a
Nash equilibrium of the biased game (i.e., $\left(s'_{1},s_{2}^{*}\right)\in NE\left(G_{\left(\psi'_{1},\psi_{2}^{*}\right)}\right)$).
In order to move from this equilibrium to $\left(s'_{1},s_{2}'\right)$,
agents of population 2, who are matched against the deviators, have
to change their behavior from $s_{2}^{*}$ to $s_{2}'$, but there
is no reason for them to do so, as their current behavior (namely,
$s_{2}^{*}$) is already a best reply to the perceived strategy of
the deviators (i.e., $s_{2}^{*}\in BR\left(\psi{}_{1}^{*}\left(s'_{1}\right)\right)$),
as well as being how they are used to playing against non-deviators.

\paragraph{Delegation Interpretation}

A different interpretation of our solution concept relies on strategic
delegation. The literature on strategic delegation (see, e.g., \citealp{FershtmanJuddEtAl1991Observable,Dufwenberg_Guth_1999_EvoandDel,fershtman2001strategic})
deals with players who strategically use other agents to play on their
behalf, where the agents so used may have different preferences than
the players using them. We adapt this approach to our setup in which
agents differ in their biased beliefs (rather than in their preferences).
Specifically, in Appendix \ref{sec:Principal-Agent} we show that
the notion of weak BBE is equivalent to a subgame-perfect equilibrium
of a two-stage game in which in stage one each unbiased player strategically
chooses the biased belief of his agent, and in the second stage the
biased agents play on behalf of the players (and each agent can observe
the opposing agent's biased beliefs).

\paragraph{Partial Observability}

The requirement that an agent be able to observe that his opponent
belongs to a group of ``mutant'' agents who have different biased
beliefs than the rest of the population can be explained by pre-play
social cues and messages that facilitate this observation. In Appendix
\ref{sec:Partial-Observability} we show that this observability need
not be perfect. We generalize the model to partial observability by
studying a setup in which, when an agent is matched with a mutant
opponent, the agent privately observes the opponent to be a mutant
with probability $0<p\leq1$. We show that all our results hold in
this extended setup for $p$ sufficiently close to one (and some of
the results hold also for low levels of $p$).

\section{Nash Equilibria and BBE Outcomes\label{sec:Relations-with-Nash}}

In this section we study the relations between Nash equilibria and
BBE outcomes. 

\subsection{Nash Equilibria and Biased Beliefs}

We begin with a simple observation that shows that in any weak BBE
in which the outcome is not a Nash equilibrium, at least one of the
players must distort the opponent's perceived strategy. The reason
for this observation is that if both players have undistorted beliefs,
then it must be that each agent best-replies to the opponent's strategy,
which implies that the outcome is a Nash equilibrium of the underlying
game.

The following example demonstrates that even Nash equilibria may require
distorted beliefs to be supported as BBE outcomes. Specifically, Example
\ref{exam-Cournot-first} shows that this is the case for Nash equilibrium
in a Cournot competition. The intuition behind Example \ref{exam-Cournot-first}
is straightforward. The Cournot equilibrium cannot be supported by
undistorted beliefs because such pairs of beliefs will induce one
of the players to adopt a distorted belief by which he expects his
opponent not to produce at all, and to best-reply to this distorted
belief by producing the monopoly quantity. This in turn will force
the opponent to reduce his production substantially below the Cournot
level, making the deviator better off.
\begin{example}[\emph{Cournot equilibrium cannot be supported by undistorted beliefs,
yet it can be supported by blind beliefs}]
\label{exam-Cournot-first} \label{ex-Cournot-Nash-cannot-be-supported-by-identity}Consider
the following symmetric Cournot game $G=\left(S,\pi\right)$: $S_{i}=\left[0,1\right]$
and $\pi_{i}\left(s_{i},s_{j}\right)=s_{i}\cdot\left(1-s_{i}-s_{j}\right)$
for each player $i$. The interpretation of the game is as follows.
Each $s_{i}$ is interpreted as the quantity chosen by firm $i$,
the price of both goods is determined by the linear inverse demand
function $p=1-s_{i}-s_{j}$, and the marginal cost of each firm is
normalized to be zero. The unique Nash equilibrium of the game is
$s_{i}^{*}=s_{j}^{*}=\frac{1}{3}$, which yields a payoff of $\frac{1}{9}$
to both players. Assume to the contrary that this outcome can be supported
as a weak BBE by the undistorted beliefs $\psi_{i}^{*}=\psi_{j}^{*}=I_{d}$.
Consider a deviation of player $1$ to the blind belief $\psi'_{1}\equiv0$.
The unique equilibrium of the biased game $G_{\left(0,I_{d}\right)}$
is $s'_{1}=\frac{1}{2}$, $s'_{2}=\frac{1}{4}$, which yields a payoff
of $\frac{1}{8}>\frac{1}{9}$ to the deviator. The unique Nash equilibrium
$s_{i}^{*}=s_{j}^{*}=\frac{1}{3}$ can be supported as the outcome
of the strong BBE $\left(\left(\frac{1}{3},\frac{1}{3}\right),\left(\frac{1}{3},\frac{1}{3}\right)\right)$
with blind beliefs, in which each agent believes the opponent is playing
$\frac{1}{3}$ regardless of the opponent's actual play, and the agent
plays the unique best reply to this belief, which is the strategy
$\frac{1}{3}$.
\end{example}
\begin{rem}[Interpretation of Nash equilibria supported by blind beliefs.]
 We interpret an undistorted belief as describing an agent who has
an accurate belief about the opponent's behavior on the equilibrium
path, and, in addition, the agent keeps looking for cues that his
opponent might have a different type, and if the agent observes such
a cue, the agent evaluates the opponent's likely behavior, and best-replies
to this assessment. Example \ref{exam-Cournot-first} shows that the
Cournot equilibrium cannot be supported by a population in which each
agent keeps looking for cues for his opponent's type. In such a population,
deviators would strictly earn by having a blind biased belief that
induces the deviator to play the Stackelberg strategy. The incumbents
will identify the mutants' type, and they will respond by playing
the Stackelberg follower action, which will benefit the deviators.

By contrast, the second part of Example \ref{exam-Cournot-first}
(and its generalization in Proposition \ref{Prop--Nash-is-BBE} below)
shows that any Nash equilibrium can be supported by a blind belief,
which is accurate on the equilibrium path. We interpret such a belief
as describing an agent who understands correctly the equilibrium behavior
of the opposing player, and ignores signals that suggest that his
opponent is about to do something else. Our observation that it is
rather equilibrium that supports belief rigidity, a prevalent behavioral
phenomenon, and not disequilibrium is, we believe, quite interesting.
\end{rem}

\subsection{Any Nash Equilibrium is a BBE Outcome}

The following result generalizes the second part of Example \ref{ex-Cournot-Nash-cannot-be-supported-by-identity},
and shows that any (strict) Nash equilibrium is an outcome of a (strong)
BBE in which both players have blind beliefs that are accurate on
the equilibrium path.
\begin{prop}
\label{Prop--Nash-is-BBE}Let $\left(s_{1}^{*},s_{2}^{*}\right)$
be a (strict) Nash equilibrium of the game $G=\left(S,\pi\right)$.
Let $\psi_{1}^{*}\equiv s_{2}^{*}$ and $\psi_{2}^{*}\equiv s_{1}^{*}$.
Then $\left(\left(\psi_{1}^{*},\psi_{2}^{*}\right),\left(s_{1}^{*},s_{2}^{*}\right)\right)$
is a (strong) BBE. 
\end{prop}
\begin{proof}
The fact that $\left(s_{1}^{*},s_{2}^{*}\right)$ is a Nash equilibrium
of the underlying game implies that $\left(s_{1}^{*},s_{2}^{*}\right)$
is an equilibrium of the biased game $G_{\left(\psi_{1}^{*},\psi_{2}^{*}\right)}$.
The fact that the beliefs are blind implies that for any biased belief
$\psi'_{i}$, there is an equilibrium in the biased game $G_{\left(\psi'_{i},\psi_{j}^{*}\right)}$
in which player $j$ plays $s_{j}^{*}$ and player $i$ gains at most
$\pi_{i}\left(s_{i}^{*},s_{j}^{*}\right)$, which implies that $\left(\left(\psi_{1}^{*},\psi_{2}^{*}\right),\left(s_{1}^{*},s_{2}^{*}\right)\right)$
is a BBE. Moreover, if $\left(s_{1}^{*},s_{2}^{*}\right)$ is a strict
equilibrium, then in any equilibrium of any biased game $G_{\left(\psi'_{i},\psi_{j}^{*}\right)}$,
player $j$ plays $s_{j}^{*}$ and player $i$ gains at most $\pi_{i}\left(s_{i}^{*},s_{j}^{*}\right)$,
which implies that $\left(\left(\psi_{1}^{*},\psi_{2}^{*}\right),\left(s_{1}^{*},s_{2}^{*}\right)\right)$
is a strong BBE.
\end{proof}
An immediate corollary of Proposition \ref{Prop--Nash-is-BBE} is
that every game admits a  BBE.
\begin{cor}
Every game admits a BBE.
\end{cor}

\subsection{Zero-Sum Games }

Recall that a game is \emph{zero sum} if there exists $c\in R^{+}$
such that $\pi_{i}(s_{i},s_{j})+\pi_{j}(s_{i},s_{j})=c$ for each
strategy profile $\left(s_{i},s_{j}\right)\in S$. 

The following simple result shows that the unique Nash equilibrium
payoff of a zero-sum game is also the unique payoff in any weak BBE.
\begin{claim}
\label{claim-The-unique-Nash-zero-sum}The unique Nash equilibrium
payoff of a zero-sum game is also the unique payoff in any weak BBE.
\end{claim}
\begin{proof}
Let $v_{i}$ be the unique Nash equilibrium payoff of player $i$
in the underlying zero-sum game. Assume to the contrary that there
exists a weak BBE $\left(\psi^{*},s^{*}\right)$ in which the payoff
of player $i$ is strictly lower than $v_{i}$. Consider a deviation
of player $i$ into the undistorted bias function $\psi'_{i}=I_{d}$.
The assumption that $\left(\psi^{*},s^{*}\right)$ is a weak BBE implies
that the deviator gets strictly less than $v_{i}$ in a Nash equilibrium
$\left(s'_{i},s'_{j}\right)\in NE\left(G_{\left(\psi'_{i},\psi_{j}^{*}\right)}\right)$,
but this is impossible as the definition of $v_{i}$ implies that
there exists $\hat{s}_{i}$ satisfying $\pi_{i}\left(\hat{s}_{i},s'_{j}\right)\geq v_{i}>\pi_{i}\left(s'_{i},s'_{j}\right)$.
\end{proof}
Example \ref{exa-RPS} in Appendix \ref{subsec:Prisoner's-Dilemma-with}
shows that even though the weak BBE payoff must be the Nash equilibrium
payoff in a zero-sum game, the strategy profile sustaining it need
not be a Nash equilibrium.

\subsection{Games with a Dominant Strategy}

Next we show that if at least one of the players has a dominant strategy,
then any weak BBE outcome must be a Nash equilibrium. Formally:

\begin{prop}
\label{prop-dominant-action}If a game admits a strictly dominant
strategy $s_{i}^{*}$ for player $i$, then any  weak BBE outcome
is a Nash equilibrium of the underlying game.
\end{prop}
\begin{proof}
Observe that $s_{i}^{*}$ is the unique best reply of player $i$
to any perceived strategy of player $j$, and, as a result, player
$i$ plays the dominant action $s_{i}^{*}$ in any  weak BBE. Assume
to the contrary that there is a weak BBE in which player $j$ does
not best-reply against $s_{i}^{*}$. Consider a deviation of player
$j$ to choosing the undistorted belief $I_{d}$. Observe that player
$i$ still plays his dominant action $s_{i}^{*}$, and that player
$j$ best-replies to $s_{i}^{*}$ in any Nash equilibrium of the induced
biased game, and, as a result, player $j$ achieves a strictly higher
payoff, and we get a contradiction.
\end{proof}
Proposition \ref{prop-dominant-action} implies, in particular, that
defection is the unique weak BBE outcome in the prisoner's dilemma
game. Example \ref{exa-PD-withdrawl-1} in Appendix \ref{subsec:Monotone-Strong-BBE-zero-sum-game-not-NAsh}
demonstrates that a relatively small change to the prisoner's dilemma
game, namely, adding a third weakly dominated ``withdrawal'' strategy
that transforms ``cooperation'' into a weakly dominated strategy,
allows us to sustain cooperation as a strong BBE outcome. 

\section{Monotone Games and Wishful Thinking\label{sec:Monotone-Games-and-wishful}}

In this section we present a large class of games with monotone externalities
and monotone differences, and define the notions of wishful thinking
and pessimism, which will be analyzed in Section \ref{sec:Main-Results}. 

\subsection{Monotone Games}

We say that an interval game is monotone if it satisfies two conditions: 
\begin{enumerate}
\item \emph{Monotone externalities}: the payoff function of each player
is strictly monotone in the opponent's strategy. Without loss of generality,
we assume that the \emph{externalities} are \emph{positive}, i.e.,
the payoff of each player is increasing in the opponent's strategy,
i.e., that $\frac{\partial\pi_{i}\left(s_{i},s_{j}\right)}{\partial s_{j}}>0$
for each player $i$ and each pair of strategies $s_{i},s_{j}$. The
assumption of positive externalities (given monotone externalities)
is indeed without loss of generality because if originally the externalities
with respect to player $j$ are negative, then we can redefine player
$j$'s strategy to be its inverse, and obtain positive externalities;
for example, defining the difference between maximal capacity and
quantity to be the strategy of each player in a Cournot competition
yields a game with positive externalities.\\
In a game with positive externalities we refer to a player's strategy
as his \emph{investment}, and when $s_{i}$ increases we refer to
this increase a larger investment by as player $i$.
\item \emph{Monotone differences}: For each player $i$, the derivative
of the player's payoff with respect to his own strategy (i.e., $\frac{\partial\pi_{i}\left(s_{i},s_{j}\right)}{\partial s_{i}}$)
is strictly monotone in the opponent's strategy. Specifically, we
divide the set of monotone games into three disjoint and exhaustive
subsets:
\begin{enumerate}
\item \emph{Strategic complements (increasing differences, supermodular
games)}: $\frac{\partial\pi_{i}\left(s_{i},s_{j}\right)}{\partial s_{i}}$
is strictly increasing in $s_{j}$ for each player $i$ and each strategy
$s_{i}$ (or, equivalently, $\frac{\partial^{2}\pi_{i}\left(s_{i},s_{j}\right)}{\partial s_{i}\cdot\partial s_{j}}>0$
for each $s_{i},s_{j}$). Games with strategic complements are common
in the economics literature, and include, in particular, price competitions
with differentiated goods (Example \ref{exam-Price-competition-with-differianted}),
input games (Example \ref{exam-partnership-game} in Appendix \ref{subsec:Examples-of-Games-complements}),
and stag-hunt games (Example \ref{exa-stag-hunt-analysis} in Appendix
\ref{subsec:Examples-of-Games-complements}). Finite games with a
payoff structure that resembles a discrete variant of strategic complements
include the traveler's dilemma (Example \ref{exam-traveler-dilemma-1}
in Appendix \ref{subsec:Examples-of-Games-complements}).
\item \emph{Strategic substitutes (decreasing differences,} \emph{submodular
games}): $\frac{\partial\pi_{i}\left(s_{i},s_{j}\right)}{\partial s_{i}}$
is strictly decreasing in $s_{j}$ for each player $i$ and each strategy
$s_{i}$ (or, equivalently, $\frac{\partial^{2}\pi_{i}\left(s_{i},s_{j}\right)}{\partial s_{i}\cdot\partial s_{j}}<0$
for each $s_{i},s_{j}$). Games with strategic \emph{substitutes}
are common in the economics literature, and include, in particular,
Cournot (quantity) competitions (Example \ref{exam-Cournot} below)
and hawk-dove games (see Example \ref{exam-The-hawk-dove-r} in Appendix
\ref{subsec:Hawk-dove-game}). 
\item \emph{Opposing differences}: $\frac{\partial\pi_{i}\left(s_{i},s_{j}\right)}{\partial s_{i}}$
is decreasing in $s_{j}$ (for each strategy $s_{i}$), while $\frac{\partial\pi_{j}\left(s_{i},s_{j}\right)}{\partial s_{j}}$
is increasing in $s_{i}$ (for each strategy $s_{j}$). Games with
opposing differences are less common in the economics literature.
Examples of these games include (1) duopolies in which one firm chooses
its quantity, while the other firm chooses its price (see, e.g., Singh
and Vives, 1984), and (2) asymmetric contests, in which it is often
the case that a commitment of the favorite (underdog) player to exert
more (less) effort induces the opponent to exert less effort (see,
e.g., Dixit, 1987).
\end{enumerate}
\end{enumerate}

\subsection{Wishful Thinking}

We say that player $i$ exhibits wishful thinking if the perceived
opponent's strategy yields a higher payoff to the player relative
to the real strategy the opponent plays. Formally:
\begin{defn}
Player $i$ \emph{exhibits wishful thinking} in weak BBE $\left(\left(\psi_{1}^{*},\psi_{2}^{*}\right),\left(s_{1}^{*},s_{2}^{*}\right)\right)$
if $\pi_{i}\left(s_{i},\psi_{i}^{*}\left(s_{j}\right)\right)\geq\pi_{i}\left(s_{i},s_{j}^{*}\right)$
for each $s_{i}\in S_{i}$. 
\end{defn}
\begin{rem}
\label{Remark-wihsful-thinking}Note that in a game with positive
externalities player $i$\emph{ }exhibits wishful thinking in weak
BBE $\left(\left(\psi_{1}^{*},\psi_{2}^{*}\right),\left(s_{1}^{*},s_{2}^{*}\right)\right)$
iff $\psi_{2}^{*}\left(s_{1}^{*}\right)\geq s_{1}^{*}$ and $\psi_{1}^{*}\left(s_{2}^{*}\right)\geq s_{2}^{*}$.
\end{rem}
Similarly, we define the opposite notion, that of exhibiting pessimism.
We say that a BBE exhibits pessimism if the perceived opponent's strategy
yields a lower payoff to the player relative to the real opponent's
strategy for all strategy profiles. It exhibits pessimism in equilibrium
if it satisfies this property with respect to the strategy the opponent
plays on the equilibrium path. Formally:
\begin{defn}
A weak BBE $\left(\left(\psi_{1}^{*},\psi_{2}^{*}\right),\left(s_{1}^{*},s_{2}^{*}\right)\right)$
\emph{exhibits} \emph{pessimism} if $\pi_{i}\left(s_{i},\psi_{i}^{*}\left(s_{j}\right)\right)\leq\pi_{i}\left(s_{i},s_{j}^{*}\right)$
for all $s_{i}\in S_{i}$. 
\end{defn}

\subsection{Additional Definitions}

In what follows we present two definitions that will be used in the
analysis in the following sections: undominated Pareto optimality,
and biased-belief minmax payoff.

We say that a strategy profile is undominated Pareto optimal if it
is (1) undominated, and (2) Pareto optimal among all undominated strategy
profiles. Formally:
\begin{defn}
\label{def-undominated-Pareto}Strategy profile $\left(s_{1}^{*},s_{2}^{*}\right)$
is \emph{undominated Pareto optimal if (1) $s_{i}^{*}\in S_{i}^{U}$
for each player $i$, and (2) there does not exist $\left(s'_{1},s'_{2}\right)\in S_{1}^{U}\times S_{2}^{U}$
with a payoff that Pareto dominates $\left(s_{1}^{*},s_{2}^{*}\right)$,
i.e., $\pi_{1}\left(s_{1}^{*},s_{2}^{*}\right)\leq\pi_{1}\left(s'_{1},s'_{2}\right)$
and $\pi_{2}\left(s_{1}^{*},s_{2}^{*}\right)\leq\pi_{2}\left(s'_{1},s'_{2}\right)$
where at least one of these inequalities is strict.}
\end{defn}
A biased-belief minmax payoff for player $i$ (denoted by $\tilde{M}_{i}^{U}$)
is the maximal payoff player $i$ can guarantee to himself in the
following process: (1) player $j$ chooses an arbitrary perceived
strategy of player $i$, and (2) player $i$ chooses a strategy profile,
under the constraint that player $j$'s strategy is a best reply to
the perceived strategy chosen above. That is, $\tilde{M}_{i}^{U}$
is the payoff player \emph{i} can guarantee himself no matter how
his opponent (player \emph{j}, she) perceives player i's action, assuming
that player $j$ best-replies to what he believes player \emph{i}
is doing (and if there are multiple best replies, then we assume that
player $j$ chooses the best reply that is optimal for player $i$).
Formally:
\begin{defn}
\label{def-BB-minamax}Given game $G=\left(A,u\right)$, let $\tilde{M}_{i}^{U}$
, the\emph{ biased-belief minmax payoff} of player $i$, be defined
as follows:
\[
\tilde{M}_{i}^{U}=\min_{s'_{i}\in S_{i}^{U}}\left(\max_{\left(s_{i},s_{j}\right)\in S_{i}\times BR\left(s'_{i}\right)}\pi_{i}\left(s_{i},s_{j}\right)\right).
\]

Observe that the biased-belief minmax is weakly larger than the undominated
maxmin (Definition \ref{def-undominated-minmax}), i.e., $\tilde{M}_{i}^{U}\geq M_{i}^{U}$
with an equality if the strategy of player $j$ that guarantees that
player $i$'s payoff is at most $M_{i}^{U}$ is a unique best reply
against some strategy of player $i$ (which is the case, in particular,
if the payoff function is strictly concave).
\end{defn}

\section{Main Results\label{sec:Main-Results}}

Our main results characterize the set of BBE and BBE outcomes in three
classes of games: games with strategic complements, games with strategic
substitutes, and games with strategic opposites.

\subsection{Preliminary Result: Necessary Conditions for a Weak BBE Outcome\label{subsec:Preliminary-Result:-Necessary}}

We begin by defining undominated strategies and the undominated minmax
payoff, which will be used to characterize necessary conditions for
a strategy profile to be a weak BBE outcome.

Strategy $s_{i}$ of player $i$ is \emph{undominated} if it is a
best reply of some strategy of the opponent, i.e., if there exists
strategy $s_{j}\in S_{j}$, such that $s_{i}\in BR\left(s_{j}\right)$.
We say that a strategy profile is \emph{undominated} if both strategies
in the profile are undominated. Recall that in a finite game, due
to the minmax theorem, a strategy is undominated iff it is not strictly
dominated by another strategy.

Let $S_{i}^{U}\in S_{i}$ denote the\emph{ set of undominated strategies}
of player $i$. Observe that $S_{i}^{U}$ is not necessarily a convex
set. 

An undominated minmax payoff for player $i$ is the maximal payoff
player $i$ can guarantee to himself in the following process: (1)
player $j$ chooses an arbitrary undominated strategy, and (2) player
$i$ chooses a strategy (after observing player $j$'s strategy).
Formally:
\begin{defn}
\label{def-undominated-minmax}Given game $G=\left(S,u\right)$, let
$M_{i}^{U}$ , the\emph{ undominated minmax payoff} of player $i$,
be defined as follows:
\[
M_{i}^{U}=\min_{s_{j}\in S_{j}^{U}}\left(\max_{s_{i}\in S_{i}}\pi_{i}\left(s_{i},s_{j}\right)\right).
\]

Observe that the undominated minmax is weakly larger than the standard
maxmin, i.e., $M_{i}^{U}\geq\min_{s_{j}\in S_{j}}\left(\max_{s_{i}\in S_{i}}\pi_{i}\left(s_{i},s_{j}\right)\right)$
with an equality if player $j$ does not have any strictly dominated
strategy\footnote{The undominated minmax payoff might be strictly higher than the undominated
\emph{maxmin} payoff due to the non-convexity of $S_{U}^{j}$; i.e.,
player $i$ might be able to guarantee only a lower payoff in a setup
in which player $j$ is allowed to choose his undominated strategy
after observing player $i$'s chosen strategy.} (i.e., if $S_{j}^{U}=S_{j}$).
\end{defn}
The following simple result (which will be helpful in deriving the
main results in the following subsections) shows that any weak BBE
outcome is an undominated strategy profile that yields a payoff above
the player's undominated minmax payoff to each player. 
\begin{prop}
\label{prop-neccesary-conditions}If a strategy profile $s^{*}=\left(s_{1}^{*},s_{2}^{*}\right)$
is a weak BBE outcome, then (1) the profile $s^{*}$ is undominated
and (2) $\pi_{i}\left(s^{*}\right)\geq M_{i}^{U}$.
\end{prop}
\begin{proof}
Assume that $s^{*}=\left(s_{1}^{*},s_{2}^{*}\right)$ is a biased-belief
equilibrium outcome. This implies that each $s_{i}^{*}$ is a best
reply to the player's distorted belief, which implies that each $s_{i}^{*}$
is undominated. Assume to the contrary, that $,\pi_{i}\left(s^{*}\right)<M_{i}^{U}.$
Then, by deviating to the undistorted function $I_{d}$, player $i$
can guarantee a fitness of at least $M_{i}^{U}$ in any distorted
equilibrium.
\end{proof}

\subsection{Games with Strategic Complements\label{sec:Wishful-thinking-and-complements}}

Our first main result characterizes the set of BBE outcomes in games
with strategic complements. It shows that a strategy profile is a
BBE outcome essentially iff (I) it is undominated, (II) it yields
a payoff above the undominated/biased-belief minmax payoff to both
players, and (III) both players overinvest (i.e., use a weakly higher
strategy than the best reply to the opponent). Formally:
\begin{prop}
\label{prop-supermodular-monotone}Let $G$ be a game with \emph{strategic
complements and positive externalities}.\emph{ }
\begin{enumerate}
\item Let $\left(s_{1}^{*},s_{2}^{*}\right)$ be a BBE outcome. Then $\left(s_{1}^{*},s_{2}^{*}\right)$
has the following properties: (I) it is undominated, and it satisfies
for each player $i$: (II) $\pi_{i}\left(s_{i}^{*},s_{j}^{*}\right)\geq M_{i}^{U}$,
and (III) overinvestment: $s_{i}^{*}\geq\min\left(BR\left(s_{j}^{*}\right)\right)$.
\item Let $\left(s_{1}^{*},s_{2}^{*}\right)$ be an undominated profile
that satisfies, for each player $i$: (II) $\pi_{i}\left(s_{i}^{*},s_{j}^{*}\right)>\tilde{M}_{i}^{U}$,
and (III) $s_{i}^{*}\geq\min\left(BR\left(s_{j}^{*}\right)\right)$.
Then, $\left(s_{1}^{*},s_{2}^{*}\right)$ is a BBE outcome. \\
Moreover, if $\pi_{i}\left(s_{i},s_{j}\right)$ is strictly concave
in $s_{i}$ (i.e., $\frac{\partial\pi_{i}^{2}\left(s_{i},s_{j}\right)}{\partial s_{i}^{2}}>0$)
then $\left(s_{1}^{*},s_{2}^{*}\right)$ is a strong BBE outcome.
\end{enumerate}
\end{prop}
\begin{proof}[Sketch of Proof (formal proof in Appendix \ref{subsec:Proof-of-Proposition-supermodular})]
\textbf{}\\
\textbf{Part 1: }Proposition \ref{prop-neccesary-conditions} implies
(I) and (II). To prove (III, overinvestment), assume to the contrary
that $s_{i}^{*}<\min\left(BR\left(s_{j}^{*}\right)\right)$. Consider
a deviation of player $i$ that induces him to invest slightly more
than $s_{i}^{*}$. The fact that $s_{i}^{*}<\min\left(BR\left(s_{j}^{*}\right)\right)$
implies that player $i$ strictly earns from his own deviation. The
assumption that the biased belief of the opponent is monotone implies
that the agent's deviation induces the opponent to invest more and,
thereby to further improve the agent's payoff. Thus, the agent gains
from the deviation, and $\left(s_{1}^{*},s_{2}^{*}\right)$ cannot
be a BBE outcome.

\textbf{Part 2:} The strategy profile $\left(s_{1}^{*},s_{2}^{*}\right)$
is supported as a BBE outcome by a profile of biased beliefs $\left(\psi_{1}^{*},\psi_{2}^{*}\right)$
in which each biased belief $\psi_{j}^{*}$ satisfies: (1) blindness
to good news: $\psi_{j}^{*}$ distorts any $s'_{i}\geq s_{i}^{*}$
into $BR^{-1}\left(s_{j}^{*}\right)$, and (2) overreaction to bad
news: $\psi_{j}^{*}$ distorts any $s'_{i}<s_{i}^{*}$ to a sufficiently
low strategy $\psi_{j}\left(s'_{i}\right)$, such that player $i$
loses in any strategy profile $\left(s'_{i},s'_{j}\right)$ in which
player $j$ best-replies to the perceived strategy of player $i$
(i.e., $s'_{j}\in BR\left(\psi_{j}\left(s'_{i}\right)\right)$). These
properties imply that $\left(\left(\psi_{1}^{*},\psi_{2}^{*}\right),\left(s_{1}^{*},s_{2}^{*}\right)\right)$
is a BBE (and a strong BBE if the payoff function is strictly concave).
\end{proof}
Recall that a game with strategic complements admits a \emph{lowest
Nash equilibrium} $\left(\underline{s}_{1},\underline{s}_{2}\right)$
in which both players invest less than in any other Nash equilibrium,
i.e., $s'_{i}\geq\underline{s}_{i}$ for each player $i$ and each
strategy $s'_{i}$ that is played in a Nash equilibrium (see, e.g.,
\citealp{milgrom1990rationalizability}). 

An immediate corollary of Prop. \ref{prop-supermodular-monotone}
is that in each BBE outcome, both players invest more than in any
Nash equilibrium. Formally:
\begin{cor}
\label{cor-no-bad-BBE-outcomes-strategic-complements}Let $G$ be
a game with strategic complements and positive externalities with
a lowest Nash equilibrium $\left(\underline{s}_{1},\underline{s}_{2}\right)$
that satisfies $\underline{s}_{1}<\max\left(S_{i}\right)$ for each
player $i$. Let $\left(s_{1}^{*},s_{2}^{*}\right)$ be a BBE outcome.
Then $\underline{s}_{i}\leq s_{i}^{*}$ for each player $i$.
\end{cor}
\begin{proof}
The result is immediate from part (1.III) of Proposition \ref{prop-supermodular-monotone}
(namely, that both agents weakly overinvest in any BBE outcome), and
the observation (which is formally proved in Lemma \ref{lemma-no-worse-than-worset-Nash}
in Appendix \ref{subsec-proof-og-core-no-worse-than-worst}) that
$s_{i}^{*}<\underline{s}_{i}$ implies that at least one of the players
strictly underinvests.
\end{proof}
Corollary \ref{cor-no-bad-BBE-outcomes-strategic-complements} shows
that the notion of BBE rules out socially bad outcomes in which one
(or both) of the players invests less effort than the lowest Nash
equilibrium. In particular, in a price competition with differentiated
goods (see Example \ref{exam-Price-competition-with-differianted}
below), the corollary implies that the price chosen by any player
in any  BBE is at least the player's price in the unique Nash equilibrium
of the game.

The final corollary shows the close relation between BBE and wishful
thinking. Specifically, it shows that any biased belief in any BBE
(with a non-extreme outcome) of a game with strategic complements
exhibits wishful thinking. The intuition is that wishful thinking
causes an agent to believe that the opponent is playing a higher action,
which induces the agent to respond with a higher action, which, in
turn, causes the opponent to respond by playing a higher action, which
benefits the agent.\footnote{Corollary \ref{cor-wishful-tiniking-with-complements} allows for
pessimism of player $i$ in a BBE only if player $i$ plays an extreme
strategy (either, the minimal feasible strategy or the maximal feasible
strategy) and his pessimism does not affect his play; i.e., the best
reply against the real opponent's strategy and the best reply against
the perceived opponent's strategy coincide in being the same extreme
strategy. For example, this is the case in the biased beliefs that
support the action profile $\left(s_{i},s_{j}\right)$ in the stag
hunt game analyzed below.}
\begin{cor}
\label{cor-wishful-tiniking-with-complements}Let $G$ be a game with
positive externalities and strategic complements. Let\\
 $\left(\left(\psi_{1}^{*},\psi_{2}^{*}\right),\left(s_{1}^{*},s_{2}^{*}\right)\right)$
be a BBE. If $s_{i}^{*}\notin\left\{ \min\left(S_{i}\right),\max\left(S_{i}\right)\right\} $,
\emph{then player $i$ exhibits wishful thinking (i.e., $\psi_{i}^{*}\left(s_{j}^{*}\right)\geq s_{j}^{*}$).
}
\end{cor}
\begin{proof}
Assume to the contrary that $\psi_{i}^{*}\left(s_{j}^{*}\right)<s_{j}^{*}$.
The strategic complementarity implies that $\max\left(BR\left(\psi_{i}^{*}\left(s_{j}^{*}\right)\right)\right)\leq\min\left(BR\left(s_{j}^{*}\right)\right)$
with an equality only if 
\[
\max\left(BR\left(\psi_{i}^{*}\left(s_{j}^{*}\right)\right)\right)\in\left\{ \min\left(S_{i}\right),\max\left(S_{i}\right)\right\} 
\]
 (see Lemma \ref{lemma-wishful-somplements} in Appendix \ref{subsec:Proof-of-cor-wishful-complements}
for a formal proof of this claim). Part 1 of Proposition \ref{prop-supermodular-monotone}
and the definition of a BBE imply that 
\[
\max\left(BR\left(\psi_{i}^{*}\left(s_{j}^{*}\right)\right)\right)\geq s_{i}^{*}\geq\min\left(BR\left(s_{j}^{*}\right)\right).
\]
 The previous inequalities jointly imply that 
\[
\max\left(BR\left(\psi_{i}^{*}\left(s_{j}^{*}\right)\right)\right)=s_{i}^{*}=\min\left(BR\left(s_{j}^{*}\right)\right)\in\left\{ \min\left(S_{i}\right),\max\left(S_{i}\right)\right\} ,
\]
 which contradicts the assumption that $s_{i}^{*}\notin\left\{ \min\left(S_{i}\right),\max\left(S_{i}\right)\right\} $.
\end{proof}
Next, we apply our analysis of games with strategic complements to
price competition with differentiated goods (the linear city model
$\grave{\textrm{a}}$ la Hotelling). Specifically, we show that (1)
players choose prices above the unique Nash equilibrium price in all
BBE, and (2) any undominated symmetric price profile above the Nash
equilibrium price can be supported as a strong BBE. In Appendix \ref{subsec:Examples-of-Games-complements}
we present three additional examples: input games, stag hunt games,
and the traveler's dilemma.
\begin{example}[{\emph{Price competition with differentiated goods; see a textbook
analysis in \citealp[Section 12.C]{mas1995microeconomic}}}]
\emph{\label{exam-Price-competition-with-differianted} }Consider
a mass one of consumers equally distributed in the interval $\left[0,1\right]$.
Consider two firms that produce widgets, located at the two extreme
locations: 0 and 1. Every consumer wants at most one widget. Producing
a widget has a constant marginal cost, which we normalize to be zero.
Each firm $i$ chooses price $s_{i}\in\left[0,M\right]$ for its widgets.
The total cost of buying a widget from firm $i$ is equal to its price
$s_{i}$ plus $t$ times the consumer's distance from the firm, where
$t\in\left[0,M\right]$). Each buyer buys a widget from the firm with
the lower total buying cost. This implies that the total demand for
good $i$ is given by function $q_{i}\left(s_{i},s_{j}\right)$: 
\[
q_{i}\left(s_{i},s_{j}\right)=\begin{cases}
\begin{array}{cc}
0 & \frac{s_{j}-s_{i}+t}{2\cdot t}<0\\
\frac{s_{j}-s_{i}+t}{2\cdot t} & 0<\frac{s_{j}-s_{i}+t}{2\cdot t}<1\\
1 & \frac{s_{j}-s_{i}+t}{2\cdot t}>1,
\end{array}\end{cases}
\]
The payoff (profit) of firm $i$ is given by $\pi_{i}\left(s_{i},s_{j}\right)=s_{i}\cdot q_{i}\left(s_{i},s_{j}\right)$.
Observe that the payoff function is strictly concave in $s_{i}$ for
any non-extreme $s_{j}$ (and it is weakly concave for the extreme
values of $s_{j}$). One can show\emph{ }that the game has strategic
complements, and that the best-reply function of each player is:
\[
s_{i}\left(s_{j}\right)=\begin{cases}
\begin{array}{cc}
\frac{s_{j}+t}{2} & s_{j}<3\cdot t\\
s_{j}-t & s_{j}\geq3\cdot t.
\end{array}\end{cases}
\]

It is well known that the unique Nash equilibrium of this example
is given by $s_{i}=s_{j}=t,$ which yields a payoff of $\frac{t}{2}$
to each firm. 

Observe that the set of undominated strategies of each player is the
interval $\left[\frac{t}{2},\frac{M+t}{2}\right]$ (where $\frac{t}{2}$
is the best reply against $0$ and $\frac{M+t}{2}$ is the best reply
against $M$). This implies that the undominated minmax of each player
is equal to $\pi_{i}\left(\frac{3}{4}\cdot t,\frac{t}{2}\right)$=$\frac{3}{4}\cdot t\cdot\frac{3}{8}=\frac{9}{32}\cdot t$.
Proposition \ref{prop-supermodular-monotone} implies that a strategy
profile $\left(s_{i},s_{j}\right)$ is a BBE outcome if for each player
$i$: (1) $s_{i}\in\left[\frac{t}{2},\frac{M+t}{2}\right]$ (undominated
strategy), (2) $\pi_{i}\left(s_{i},s_{j}\right)>\frac{9}{32}\cdot t$
(payoff above the undominated minmax payoff),\footnote{One can show that the constraint on $s_{i}$ implied by $\pi_{i}\left(s_{i},s_{j}\right)>\frac{9}{32}\cdot t$
is nonbinding. The constraint is
\[
s_{i}\in\left(\frac{s_{j}+t-\sqrt{\left(s_{j}+t\right)^{2}-2.25\cdot t^{2}}}{2},\frac{s_{j}+t-\sqrt{\left(s_{j}+t\right)^{2}-2.25\cdot t^{2}}}{2}\right).
\]
 } and (3) overinvestment: $s_{i}\geq\frac{s_{j}+t}{2}$. 

Figure \ref{fig:The-Set-of-differiantaed-goods} shows the set of
BBE outcomes (which coincides with the set of strong BBE outcomes,
due to the strict concavity of the payoff function), for $t=1$ and
$M=3$.

\begin{figure}[h]

\begin{centering}
\caption{\label{fig:The-Set-of-differiantaed-goods}The Set of (Strong) BBE
Outcomes in Example \ref{exam-Price-competition-with-differianted}
($t=1$, $M=3$)}
\includegraphics[scale=0.65]{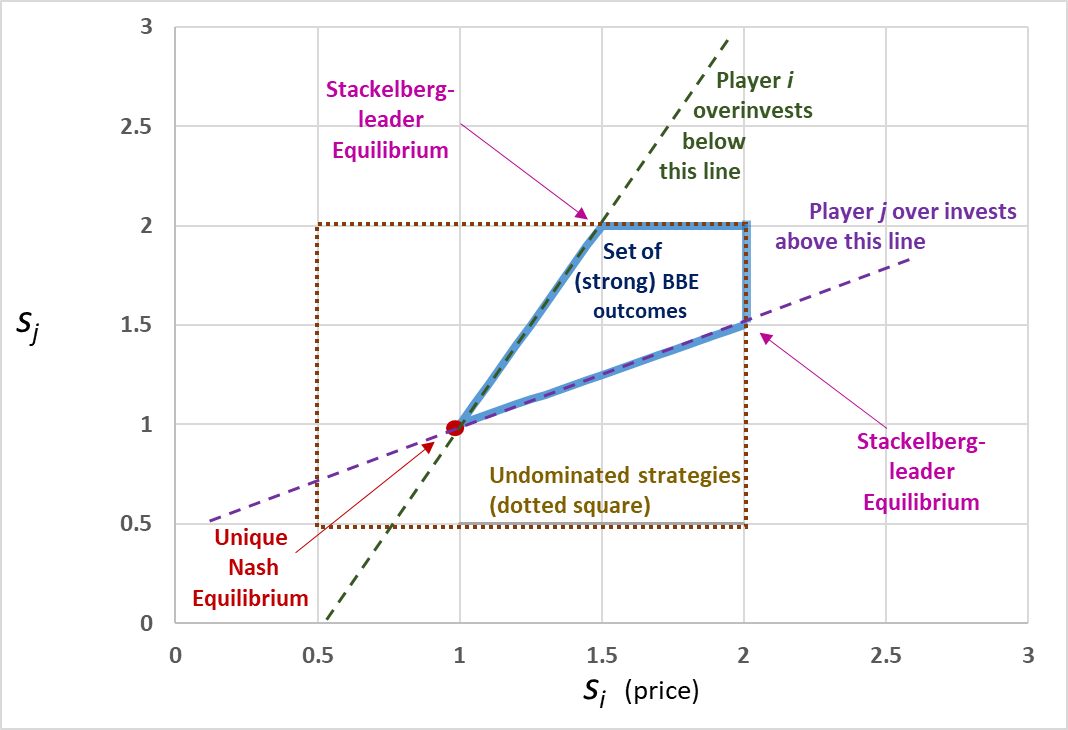}
\par\end{centering}
\end{figure}
Observe that the sum of the payoffs to the two firms, $s_{i}\cdot q_{i}\left(s_{i},s_{j}\right)+s_{j}\cdot q_{j}\left(s_{i},s_{j}\right)$,
is a mixed average of $s_{i}$ and $s_{j}$. The fact that the Nash
equilibrium is in the bottom left corner of the set of BBE outcomes
implies that all BBE outcomes (except the Nash equilibrium itself)
strictly improve social welfare relative to the Nash equilibrium (as
measured by the sum of payoffs of the two firms).

Next, we make two observations regarding the implications of the extent
of wishful thinking on the players' payoffs (both observations hold
also for the input games in Example \ref{exam-partnership-game} in
Appendix \ref{subsec:Examples-of-Games-complements}):
\begin{enumerate}
\item Increasing the wishful thinking of both players improves the players'
payoffs. Specifically, with respect to symmetric BBE outcomes, a higher
level of wishful thinking induces a higher equilibrium price and a
higher payoff to the players: a wishful thinking level of $x^{*}\equiv\psi^{*}\left(s^{*}\right)-s^{*}\in\left[0,1\right]$
induces the symmetric BBE price $x^{*}+t=x^{*}+1$ (which is implied
by the perceived bast-reply equation $s^{*}=\frac{\psi^{*}\left(s^{*}\right)+t}{2}=\frac{s^{*}+x^{*}+t}{2}$),
which yields a payoff of $\frac{x^{*}+1}{2}$ to each player.
\item When the wishful thinking levels of the two players differ, the player
with the higher wishful thinking level has a lower payoff. This is
because the difference between the payoffs of a firm with price $s_{i}$
and an opponent with price $s_{j}<s_{i}$ is equal to: 
\[
\pi_{i}-\pi_{j}=s_{i}\cdot\left(\frac{s_{j}-s_{i}+1}{2}\right)-s_{j}\cdot\left(\frac{s_{i}-s_{j}+1}{2}\right)=0.5\left(s_{j}\left(s_{j}-1\right)-s_{i}\left(s_{i}-1\right)\right)<0.
\]
\end{enumerate}
Intuitively, wishful thinking is like a public good in this setup:
(1) a higher level of wishful thinking is beneficial to social welfare,
and (2) if the two players have different levels of wishful thinking,
the player with the higher level obtains a lower payoff.

We conclude the example by presenting a symmetric biased belief $\psi_{1}^{*}=\psi_{2}^{*}$
that supports the outcome $\left(2,2\right)$ as the BBE $\left(\left(\psi_{1}^{*},\psi_{2}^{*}\right),\left(2,2\right)\right)$
in the game with $M=3$ and $t=1$: 
\[
\psi_{i}^{*}\left(s_{j}\right)=\begin{cases}
3 & s_{j}>2\\
2\cdot s_{j}-1 & s_{j}\in\left[0.5,2\right]\\
0 & s_{j}<0.5.
\end{cases}
\]
Observe that: (1) this BBE yields a payoff of 1 to each player and
(2) the biased belief presents wishful thinking, i.e., $\psi_{i}^{*}\left(2\right)=3>2$.
Further observe that a player with biased belief $\psi_{i}^{*}$ plays
the same strategy as the opponent (regardless of the opponent's biased
belief) in any equilibrium of the biased game in which the opponent
plays any intermediate value of $s_{j}$ (i.e., $s_{j}\in\left[0.5,2\right]$)
:
\[
s_{i}\left(\psi_{i}^{*}\left(s_{j}\right)\right)=\begin{cases}
s_{i}\left(3\right)=2 & s_{j}>2\\
s_{i}\left(2\cdot s_{j}-1\right)=0.5\cdot\left(2\cdot s_{j}-1+1\right)=s_{j} & s_{j}\in\left[0.5,2\right]\\
s_{i}\left(0\right)=0.5 & s_{j}<0.5.
\end{cases}
\]
This implies that the equilibrium payoff of a deviating player $j$
who plays strategy $s_{j}$ is equal to:
\[
\pi_{j}\left(s_{j},s_{i}\left(\psi_{i}^{*}\left(s_{j}\right)\right)\right)=\begin{cases}
s_{j}\cdot0.5\cdot\left(2-s_{j}+1\right)=s_{j}\cdot0.5\cdot\left(3-s_{j}\right)<1 & s_{j}>2\\
s_{j}\cdot0.5\cdot q\left(s_{j},s_{j}\right)=0.5\cdot s_{j} & s_{j}\in\left[0.5,2\right]\\
s_{j}\cdot0.5\cdot\left(0.5-s_{j}+1\right)=s_{j}\cdot0.5\cdot\left(1.5-s_{j}\right)<0.25 & s_{j}<0.5,
\end{cases}
\]
 and it is at most 1, which implies that a deviator cannot gain from
his deviation. 

Finally, note that Figure \ref{fig:The-Set-of-differiantaed-goods}
shows that the two Stackelberg-leader equilibria (the unique subgame-perfect
equilibrium of the sequential games in which one of the players plays
first, and the opponent replies after observing the leader's strategy)
are included in the set of BBE, as is proven in general in Proposition
\ref{Prop-Stackelberg}.
\end{example}

\subsection{Games with Strategic Substitutes\label{sec:Wishful-thinking-and-substitute}}

Our next result characterizes the set of BBE outcomes in games with
strategic substitutes (and positive externalities). It shows that
a strategy profile is a BBE outcome essentially iff (I) it is undominated,
(II) it yields a payoff above the undominated/biased-belief minmax
payoff to both players, and (III) both players underinvest (i.e.,
use a weakly lower strategy than the best reply to the opponent).
Formally:
\begin{prop}
\label{prop-substitues}Let $G$ be a game with \emph{strategic substitutes
and positive externalities}.\emph{ }
\begin{enumerate}
\item Let $\left(s_{1}^{*},s_{2}^{*}\right)$ be a BBE outcome. Then $\left(s_{1}^{*},s_{2}^{*}\right)$
has the following properties: (I) it is undominated, and if satisfies
for each player $i$: (II) $\pi_{i}\left(s_{i}^{*},s_{j}^{*}\right)\geq M_{i}^{U}$
and (III) $s_{i}^{*}\leq\max\left(BR\left(s_{j}^{*}\right)\right)$
(underinvestment).
\item Let $\left(s_{1}^{*},s_{2}^{*}\right)$ be an undominated profile
that satisfies, for each player $i$: (II) $\pi_{i}\left(s_{i}^{*},s_{j}^{*}\right)>\tilde{M}_{i}^{U}$
, and (III) $s_{i}^{*}\leq\max\left(BR\left(s_{j}^{*}\right)\right)$.
Then, $\left(s_{1}^{*},s_{2}^{*}\right)$ is a BBE outcome. \\
Moreover, if $\pi_{i}\left(s_{i},s_{j}\right)$ is strictly concave
then $\left(s_{1}^{*},s_{2}^{*}\right)$ is a strong BBE outcome.
\end{enumerate}
\end{prop}
The proof, which is analogous to the proof of Proposition \textbf{\ref{prop-supermodular-monotone}},
is presented in Appendix \ref{subsec:Proof-of-Proposition-substitutes}.

An immediate corollary of Proposition \ref{prop-supermodular-monotone}
is that in each BBE outcome, at least one of the players invests less
relative to his maximal Nash equilibrium investment. Formally:
\begin{cor}
\label{cor-no-efficient-outcome-substitutes}Let $G$ be a game with
strategic substitutes and positive externalities. Let $\left(s_{1}^{*},s_{2}^{*}\right)$
be a BBE outcome. Then, there exists a Nash equilibrium of the underlying
game $\left(s_{1}^{e},s_{2}^{e}\right)$, and a player $i$ such that
$s_{i}^{e}\geq s_{i}^{*}$.
\end{cor}
\begin{proof}
The result is immediate from part (1-III) of Proposition \ref{prop-supermodular-monotone}
(namely, that both agents weakly underinvest in any  BBE outcome),
and the observation (which is formally proved in Lemma \ref{lemma-no-better-than-all-Nash-both-players}
in Appendix \ref{subsec:Proof-of-Lemma-3}) that if the effort of
each player $s_{i}^{*}$ is strictly below all of his Nash equilibrium
efforts, then at least one of the players strictly underinvests.
\end{proof}
Corollary \ref{cor-no-efficient-outcome-substitutes} shows that the
notion of BBE rules out socially good outcomes in which both players
invest more effort relative to their maximal Nash equilibrium effort.
In particular, in a Cournot competition (see Example \ref{exam-Cournot}
below), the corollary implies that a collusive outcome in which both
players retain more unused capacity relative to the unique Nash equilibrium. 

Combining Corollary \ref{cor-no-bad-BBE-outcomes-strategic-complements}
and Corollary \ref{cor-no-efficient-outcome-substitutes} implies
the following\emph{ empirical prediction of our model and the notion
of  BBE: efficient (non-Nash equilibrium) outcomes are easier to support
in games with strategic complements, relative to games with strategic
substitutes. }This prediction is consistent with the experimental
findings of \citet{potters2009cooperation}, which show that there
is significantly more cooperation in games with strategic complements
than in games with strategic substitutes. 

The following corollary shows that in games with strategic substitutes,
as in games with strategic complements, there is the a close relation
between BBE and wishful thinking. Specifically, it shows that any
biased belief in any BBE (with a non-extreme outcome) of a game with
strategic substitutes exhibits wishful thinking. The intuition is
that wishful thinking causes an agent to believe that the opponent
is playing a higher action, which induces the agent to respond with
a lower action, which, in turn, causes the opponent to respond by
playing a higher action, which benefits the agent. 
\begin{cor}
\label{cor-wishful-tiniking-with-substitutes}Let $G$ be a game with
positive externalities and strategic substitutes. Let $\left(\left(\psi_{1}^{*},\psi_{2}^{*}\right),\left(s_{1}^{*},s_{2}^{*}\right)\right)$
be a BBE. If $s_{i}^{*}\notin\left\{ \min\left(S_{i}\right),\max\left(S_{i}\right)\right\} $,
\emph{then player $i$ exhibits wishful thinking (i.e., $\psi_{i}^{*}\left(s_{j}^{*}\right)\geq s_{j}^{*}$). }
\end{cor}
The proof, which is analogous to the proof of Corollary \ref{cor-wishful-tiniking-with-complements},
is presented in Appendix \ref{subsec:Proof-of-Corollary-wishful-substitutes}.

The following example characterizes the set of BBE outcomes in a Cournot
competition. Appendix \ref{subsec:Hawk-dove-game} presents an analysis
of another game of strategic substitutes: the hawk-dove game.
\begin{example}[\emph{Cournot competition with linear demand}]
\label{exam-Cournot} Consider a symmetric Cournot competition, where
we relabel the set of strategies to describe unused capacity, rather
than quantity, in order to follow the normalization of positive externalities.
Formally, let $G=\left(S,\pi\right)$: $S_{i}=\left[0,1\right]$ and
$\pi_{i}\left(s_{i},s_{j}\right)=\left(1-s_{i}\right)\cdot\left(s_{i}+s_{j}-1\right)$
for each player $i$. Each $s_{i}$ is interpreted as the unused capacity
(= one minus the quantity, i.e., $s_{i}=1-q_{i}$) chosen by firm
$i$, the price of both goods is determined by the linear inverse
demand function $p=1-q_{i}-q_{j}=s_{i}+s_{j}-1$, and the marginal
cost of each firm is normalized to be zero. 

Observe that:
\begin{enumerate}
\item $BR\left(s_{i}\right)=1-\frac{s_{i}}{2}$, and the unique Nash equilibrium
of the game is $s_{1}^{*}=s_{2}^{*}=\frac{2}{3}$, which yields a
payoff of $\frac{1}{9}$ to both players. 
\item The set of undominated strategies of each player is the interval $\left[0.5,1\right]$
(where $1$ is the best reply against 0, and 0.5 is the best reply
against 1).
\item The symmetric Pareto optimal profile (which is also undominated) is
$s_{i}=s_{j}=\frac{3}{4}$, yielding a payoff of $\frac{1}{8}$ to
each player.
\item The undominated minmax payoff $M_{i}^{U}=\frac{1}{16}$, which is
achieved by the opponent playing his lowest undominated strategy $s_{i}=0.5$.
\item The sum of payoffs of both players when they play profile $\left(s_{1},s_{2}\right)$
is $\pi_{1}\left(s_{1},s_{2}\right)+\pi_{2}\left(s_{1},s_{2}\right)=\left(2-\left(s_{i}+s_{j}\right)\right)\cdot\left(\left(s_{i}+s_{j}\right)-1\right)$,
which is an increasing function of $s_{i}+s_{j}$ in the domain of
undominated strategies $s_{i},s_{j}\geq0.5$.
\end{enumerate}
Applying the analysis of the previous subsection to a Cournot competition
shows that strategy profile $\left(s_{1},s_{2}\right)$ is a BBE outcome
iff it satisfies for each player $i$: (1) the strategy is undominated:
$s_{i}\geq0.5$, (2) the payoff is greater than the undominated minmax
payoff: $\left(1-s_{i}\right)\cdot\left(s_{i}+s_{j}-1\right)\geq\frac{1}{16}=M_{i}^{U}$,
and (3) underinvestment relative to the best reply against the opponent:
$s_{i}\leq BR\left(s_{j}\right)=1-\frac{s_{j}}{2}$. Due to having
a strictly concave payoff function, the set of BBE outcomes coincides
with the set of strong BBE outcomes. Figure \ref{fig:The-Set-of-differiantaed-goods}
shows this set of BBE outcomes (the strategy profiles that satisfy
the above three conditions). 
\begin{figure}[h]

\begin{centering}
\caption{\label{fig:The-Set-of-Cornot}The Set of (Strong) BBE Outcomes in
a Cournot Competition}
 \includegraphics[scale=0.8]{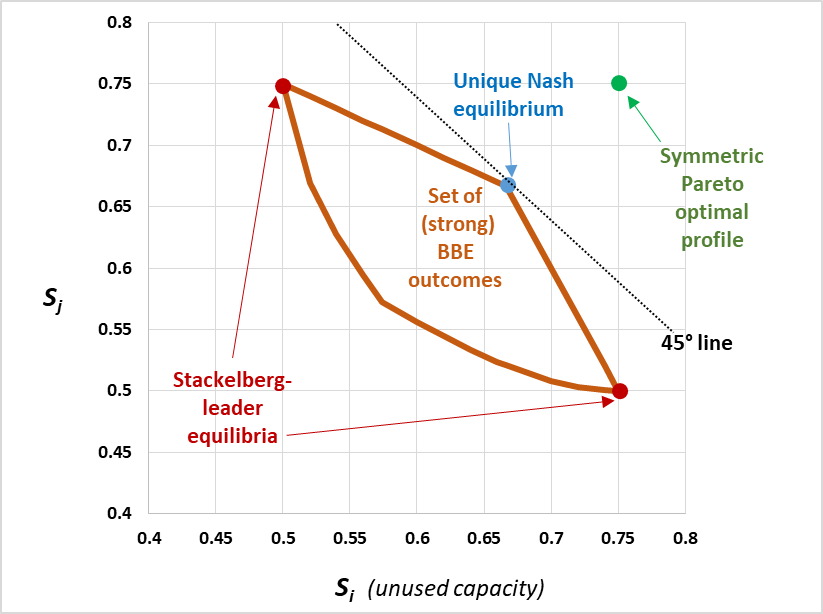}
\par\end{centering}
\end{figure}

Observe that the unique Nash equilibrium $\left(\frac{2}{3},\frac{2}{3}\right)$
is the profile that maximizes the sum $s_{i}+s_{j}$ within the set
of BBE. This implies that all other BBE outcomes yield lower social
welfare (as measured by the sum of payoffs) relative to the Nash equilibrium. 

Next, we make two observations regarding the implications of the level
of wishful thinking on the players' payoffs.
\begin{enumerate}
\item Increasing the wishful thinking of both players decreases the players'
payoffs. Specifically, when focusing on symmetric BBE outcomes, a
higher level of wishful thinking induces a lower level of unused capacity
and a lower payoff to both players; the higher level of production
is induced by the false assessment of each firm that the other firm
is producing less than it actually does.\footnote{A wishful thinking level of $x^{*}\equiv\psi^{*}\left(s^{*}\right)-s^{*}\in\left[0,0.28\right]$
induces a symmetric BBE unused capacity of $s^{*}=\frac{2-x^{*}}{3}$
(which is implied by the perceived best-reply equation $s^{*}=1-\frac{\psi^{*}\left(s^{*}\right)}{2}=1-\frac{s^{*}+x^{*}}{2}$).}
\item When the wishful thinking levels of the two players differ, the player
with the higher wishful thinking has a higher payoff. This is because
the difference between the payoffs of a firm with price $s_{i}$ and
an opponent with price $s_{j}<s_{i}$ is equal to 
\[
\pi_{i}-\pi_{j}=s_{i}\cdot\left(\frac{s_{j}-s_{i}+1}{2}\right)-s_{j}\cdot\left(\frac{s_{i}-s_{j}+1}{2}\right)=0.5\left(s_{j}\left(s_{j}-1\right)-s_{i}\left(s_{i}-1\right)\right)<0.
\]
\end{enumerate}
Thus, a higher level of wishful thinking is beneficial to social welfare,
but harms the player with the higher level (relative to the opponent's
payoff).

Finally, note that Figure \ref{fig:The-Set-of-Cornot} shows that
the two Stackelberg-leader equilibria (the unique subgame-perfect
equilibria of the sequential games in which one of the players plays
first, and the opponent replies after observing the leader's strategy)
are included in the set of BBE, as is proven in general in Proposition
\ref{Prop-Stackelberg}.
\end{example}

\subsection{Pessimism in Games with Opposing Differences\label{sec:Pessimism-in-Games}}

The results of the previous two subsections present a strong tendency
of  BBE to exhibit wishful thinking both in games with strategic complements
and in games with strategic substitutes. This raises the question
of which class of games induces pessimism. In this section we show
that the answer to this question is games with strategic opposites.
Recall that these are games in which the strategy of player 1 is a
complement of player 2's strategy, while the strategy of player 2
is a substitute of player 1's strategy, e.g., duopolistic competitions
in which one firms chooses its quantity while the opposing firm chooses
its price (\citealp{singh1984price}) and various classes of asymmetric
contests (\citealp{dixit1987strategic}). 

Proposition \ref{prop-strategic-opposites} characterizes the set
of BBE outcomes in games with strategic opposites (and positive externalities). 

It shows that a strategy profile is a BBE outcome essentially iff
(I) it is undominated, (II) it yields a payoff above the undominated/biased-belief
minmax payoff to both players, and (III) player 1 (for whom player
2's strategy is a complement) underinvests, while player 2 (for whom
player 1's strategy is a substitute) overinvests. Formally:
\begin{prop}
\label{prop-strategic-opposites}Let $G$ be a game\emph{ }with positive
externalities and strategic opposites\emph{: $\frac{\partial^{2}\pi_{1}\left(s_{1},s_{2}\right)}{\partial s_{1}\partial s_{2}}>0$
and $\frac{\partial^{1}\pi_{2}\left(s_{1},s_{2}\right)}{\partial s_{1}\partial s_{2}}<0$
for each pair of strategies $s_{1},s_{2}$}.\emph{ }
\begin{enumerate}
\item Let $\left(s_{1}^{*},s_{2}^{*}\right)$ be a BBE outcome. Then $\left(s_{1}^{*},s_{2}^{*}\right)$
is (I) undominated: (II) $\pi_{i}\left(s_{i}^{*},s_{j}^{*}\right)\geq M_{i}^{U}$
for each player $i$, and (III) $s_{1}^{*}\leq\max\left(BR\left(s_{2}^{*}\right)\right)$
and $s_{2}^{*}\geq\min\left(BR\left(s_{1}^{*}\right)\right)$ (i.e.,
player 1 underinvests and player 2 overinvests relative to the best
reply to the opponent).
\item Let $\left(s_{1}^{*},s_{2}^{*}\right)$ be a profile satisfying the
following conditions: (I) $\left(s_{1}^{*},s_{2}^{*}\right)$ is undominated,
(II) $\pi_{i}\left(s_{i}^{*},s_{j}^{*}\right)>\tilde{M}_{i}^{U}$
for each player $i$, and (III) $s_{1}^{*}\leq\max\left(BR\left(s_{2}^{*}\right)\right)$
and $s_{2}^{*}\geq\min\left(BR\left(s_{1}^{*}\right)\right)$. Then,
$\left(s_{1}^{*},s_{2}^{*}\right)$ is a BBE outcome. 
\end{enumerate}
\end{prop}
The proof, which is analogous to the proof of Proposition \textbf{\ref{prop-supermodular-monotone}},
is presented in Appendix \ref{subsec:Proof-of-Proposition-opposites}.\textbf{}\\

The following corollary shows that in games with strategic opposites,
there is a close relation between BBE and pessimism. Specifically,
it shows that any biased belief in any  BBE (with a non-extreme outcome)
of a game with strategic opposites exhibits pessimism. The intuition
is that pessimism causes player 1 to believe that player 2 is playing
a lower action, which induces player 1 to respond with a lower action,
which, in turn, causes player 2 to respond by playing a higher action,
which benefits player 1. Similarly, pessimism causes player 2 to believe
that player 1 is playing a lower action, which induces player 2 to
respond with a higher action, which, in turn, causes player 1 to respond
by playing a higher action, which benefits player 2. 
\begin{cor}
\label{cor-pessimism-withstrategic-opposites}Let $\left(\left(\psi_{1}^{*},\psi_{2}^{*}\right),\left(s_{1}^{*},s_{2}^{*}\right)\right)$
be a BBE of a game with \emph{positive externalities} and strategic
opposites (i.e.,\emph{$\frac{\partial^{2}\pi_{1}\left(s_{1},s_{2}\right)}{\partial s_{1}\partial s_{2}}>0$}
and $\frac{\partial^{2}\pi_{2}\left(s_{1},s_{2}\right)}{\partial s_{1}\partial s_{2}}<0$
for each pair of strategies\emph{ $s_{1},s_{2}$}). If $s_{i}^{*}\notin\left\{ \min\left(S_{i}\right),\max\left(S_{i}\right)\right\} $,
then player $i$ exhibits pessimism (i.e., $\psi_{i}^{*}\left(s_{j}^{*}\right)\leq s_{j}^{*}$). 
\end{cor}
The proof, which is analogous to the proof of Corollary \ref{cor-wishful-tiniking-with-complements},
is presented in Appendix \ref{subsec:Proof-of-Corollary-opposites}.

Next, we present an example of a game with strategic opposites, and
we characterize the set of BBE in this game.
\begin{example}[\emph{Matching pennies with positive externalities}]
\label{exam-matching-pennies-with-positive-externalities}The game
presented in Table \ref{tab:matching-pennies-and-PD}, a variant of
the matching pennies game, is played as follows:
\begin{enumerate}
\item Player 1 (player 2) gains 1 utility point from matching (mismatching)
his opponent.
\item Each player $i$ induces a gain of 3 utility points to his opponent
by choosing heads (action $h_{i}$).
\end{enumerate}
The game admits a unique Nash equilibrium $\left(0.5,0.5\right)$
with a payoff of 1.5 to each player. The (undominated) minmax payoff
of each player is 1 (obtained when the opponent plays $t_{j}$). Observe
that the game has positive externalities, that the strategy of player
2 is a strategic complement for player 1, while the strategy of player
1 is a strategic substitute for player 2.

\begin{table}[h]
\caption{\label{tab:matching-pennies-and-PD}Matching Pennies with Positive
Externalities$\protect\underset{}{}$}

\centering{}%
\begin{tabular}{|c|c|c|}
\hline 
 & $h_{2}$ & $t_{2}$\tabularnewline
\hline 
$h_{1}$ & $\begin{array}{c}
\\
\\
\end{array}4,2\begin{array}{c}
\\
\\
\end{array}$ & $-1,4$\tabularnewline
\hline 
$t_{1}$ & $\begin{array}{c}
\\
\\
\end{array}2,1\begin{array}{c}
\\
\\
\end{array}$ & $1,-1$\tabularnewline
\hline 
\end{tabular}
\end{table}

Applying the analysis of the previous section shows that the game
admits 2 classes of BBE:
\begin{enumerate}
\item A class in which the players mix while giving a larger weight to playing
heads (the action with positive externalities), pessimism, and one-directional
blindness. Specifically, each  BBE in this class $\left(\left(\psi_{1}^{*},\psi_{2}^{*}\right),\left(\beta_{1},\beta_{2}\right)\right)$
satisfies for each player $i$: (I) $\beta_{i}\in\left[0.5,1\right]$
(i.e., both players play heads more frequently than in the unique
Nash equilibrium), (II) pessimism: $\psi_{i}^{*}\left(\beta_{j}\right)=0.5<\beta_{j}$,
and (III) one-sided blindness: $\psi_{1}^{*}\left(\alpha\right)=0.5$
for each $\alpha\geq\beta_{2}$; $\psi_{1}^{*}\left(\alpha\right)<0.5$
for each $\alpha<\beta_{1}$; $\psi_{2}^{*}\left(\alpha\right)=0.5$
for each $\alpha\leq\beta_{2}$; and $\psi_{2}^{*}\left(\alpha\right)<0.5$
for each $\alpha>\beta_{2}$.
\item A class in which player 1 mixes while giving more weight to tails,
while player 2 plays heads. Both players exhibit pessimism. Specifically,
each BBE in this class $\left(\left(\psi_{1}^{*},\psi_{2}^{*}\right),\left(\beta_{1},\beta_{2}\right)\right)$
satisfies for each player $i$: (I) $\beta_{1}\in\left[0,0.5\right]$
and $\beta_{2}=1$ (i.e., player 1 plays tails more frequently than
in the unique Nash equilibrium, while player 2 always plays heads),
(II) pessimism for player 1: $\psi_{1}^{*}\left(\beta_{2}=1\right)=0.5<1$,
and $\psi_{2}^{*}\left(\beta_{1}\right)=0.5$ (player 2 is not pessimistic,
due to the fact that he chooses the extreme action 1), and (III) $\psi_{2}^{*}\left(\alpha\right)>0.5$
for each $\alpha>\beta_{1}$. 
\end{enumerate}
Observe that any profile $\left(\beta_{1},\beta_{2}\right)$, where
$\beta_{2}<0.5$ or ($\beta_{1}<0.5$ and $\beta_{2}<1$), cannot
be a BBE outcome:
\end{example}
\begin{enumerate}
\item If $\beta_{2}<0.5$ and $\beta_{1}=0$, then player 2's payoff is
negative, and less than his undominated minmax payoff of 1.
\item If $\beta_{2}<0.5$ and $\beta_{1}>0$, then player 1 can gain by
deviating to $\psi'_{1}\equiv0$, as the only possible equilibria
of the new biased game are $\left(0_{1},0_{2}\right)$ and $\left(0_{1},\beta_{2}\right)$,
both of which induce a higher payoff to player $1$ relative to $\left(\beta_{1},\beta_{2}\right)$.
\item If $\beta_{1}<0.5$ and $\beta_{2}<1$, then player 2 can gain by
deviating to $\psi'_{2}\equiv0$, as the only possible equilibria
of the new biased game are $\left(0_{1},1_{2}\right)$ and $\left(\beta_{1},1_{2}\right)$,
both of which induce a higher payoff to player $2$ relative to $\left(\beta_{1},\beta_{2}\right)$.
\end{enumerate}

\subsection{Empirical Prediction Regarding Wishful Thinking}

Arguably, the class of games with strategic opposites (which induces
pessimism) is less common in strategic interactions than the classes
of games with strategic complements/substitutes (both of which induce
wishful thinking). This observation suggests the following empirical
predictions of our model: (1) wishful thinking is more common than
pessimism, and (2) there are some (less common) strategic interactions
that induce pessimism. This empirical prediction is consistent with
the experimental evidence that people tend to present wishful thinking,
although, the extent of wishful thinking may substantially differ
across different environments and may disappear in some environments
(see, e.g., \citealp{babad1991wishful,budescu1995relationship,doi:10.1080/13546789508256906,mayraz2013wishful}).

\section{Additional Results\label{sec:Additional-Results}}

\subsection{BBE with Strategic Stubbornness\label{sec:BBE-and-Undominated}}

In this subsection we present an interesting class of BBE that exist
in all games. In this class, one of the players is ``strategically
stubborn'' in the sense that he plays his undominated Stackelberg
strategy (defined below) and has blind beliefs, while his opponent
is ``flexible'' in the sense of having unbiased beliefs.

A strategy is undominated Stackelberg if it maximizes a player's payoff
in a setup in which the player can commit to an undominated strategy,
and his opponent reacts by choosing the best reply that maximizes
player $i$'s payoff. Formally:
\begin{defn}
The strategy $s_{i}$ is an undominated Stackelberg strategy if it
satisfies 
\[
s_{i}=\textrm{argma}\textrm{x}_{s_{i}\in S_{i}^{U}}\left(\textrm{ma}\textrm{x}_{s_{j}\in BR\left(s_{i}\right)}\left(\pi_{i}\left(s_{i},s_{j}\right)\right)\right).
\]
 Let $\pi_{i}^{\textrm{Stac}}=\textrm{ma}\textrm{x}_{s_{i}\in S_{i}^{U}}\left(max_{s_{j}\in BR\left(s_{i}\right)}\left(\pi_{i}\left(s_{i},s_{j}\right)\right)\right)$
be the undominated Stackelberg payoff. Observe that $\pi_{i}^{\textrm{Stac}}\geq\pi_{i}\left(s_{1}^{*},s_{2}^{*}\right)$
for any Nash equilibrium $\left(s_{1}^{*},s_{2}^{*}\right)\in NE\left(G\right)$.

Our next result shows that every game admits a  BBE in which one of
the players: (1) has a blind belief, (2) plays his undominated Stackelberg
strategy, and (3) obtains his undominated Stackelberg payoff. The
opponent has undistorted beliefs. Moreover, this BBE is strong if
the undominated Stackelberg strategy is a unique best reply to some
undominated strategy of the opponent. 

The intuition behind Proposition \ref{Prop-Stackelberg} is as follows.
The ``strategically stubborn'' player $i$ cannot gain from a deviation,
because player $i$ already obtains the highest possible payoff under
the constraint that player $j$ best-replies to player $i$'s strategy.
The ``flexible'' player $j$ cannot gain from a deviation, because
the ``blindness'' of player $i$ implies that player $i$'s behavior
remains the same regardless of player $i$'s deviation, and, thus,
player $i$ cannot do better than best-replying to player $i$'s strategy.
\end{defn}
\begin{prop}
\label{Prop-Stackelberg}Game $G=\left(S,\pi\right)$ admits a BBE
$\left(\left(\psi_{i}^{*},Id\right),\left(s_{i}^{*},s_{j}^{*}\right)\right)$
for each player $i$ with the following properties: (1) $\psi_{i}^{*}$
is blind, (2) $s_{i}^{*}$ is an undominated Stackelberg strategy,
and (3) $s_{j}^{*}=max_{s_{j}\in BR\left(s_{i}^{*}\right)}\left(\pi_{i}\left(s_{i}^{*},s_{j}\right)\right)$.
\\
Moreover, $\left(\left(\psi_{i}^{*},Id\right),\left(s_{i}^{*},s_{j}^{*}\right)\right)$
is a strong BBE if $\left\{ s_{i}^{*}\right\} =BR^{-1}\left(s_{j}^{*}\right)$. 
\end{prop}
\begin{proof}
Let $s_{i}^{*}$ be an undominated Stackelberg strategy of player
$i$. Let 
\[
s_{j}^{*}=\textrm{argma}\textrm{x}_{s_{j}\in BR\left(s_{i}^{*}\right)}\left(\pi_{i}\left(s_{i}^{*},s_{j}\right)\right).
\]
Let $\hat{s}{}_{j}\in BR^{-1}\left(s_{i}^{*}\right)$ ($\left\{ \hat{s}{}_{j}\right\} =BR^{-1}\left(s_{i}^{*}\right)$
with the additional assumption of the ``moreover'' part). We now
show that $\left(\left(\psi_{i}^{*}\equiv\hat{s}{}_{j},Id\right),\left(s_{i}^{*},s_{j}^{*}\right)\right)$
is a (strong) BBE. It is immediate that $\left(s_{i}^{*},s_{j}^{*}\right)\in NE\left(G_{\left(\hat{s}{}_{j},Id\right)}\right)$,
and that both biased beliefs are monotone.

Next, observe that for any biased belief $\psi_{j}'$ there is a plausible
equilibrium (in any equilibrium) of the biased game $G_{\left(\hat{s}{}_{j},\psi'_{j}\right)}$
in which player $i$ plays $s_{i}^{*}$, and player $j$ gains at
most $\pi_{j}\left(s_{i}^{*},s_{j}^{*}\right)$, which implies that
the deviation to $\psi_{j}'$ is not profitable to player $j$ in
this plausible equilibrium (in any equilibrium) of the new biased
game. 

If player $i$ deviates to a biased belief $\psi'_{i}$, then in any
equilibrium of the biased game $G_{\left(\psi'_{i},Id\right)}$ player
$i$ plays some strategy $s'_{i}$ and gains a payoff of at most $\textrm{ma}\textrm{x}_{s'_{j}\in BR\left(s'_{i}\right)}\left(\pi_{i}\left(s'_{i},s'_{j}\right)\right)$,
and this implies that player $i$'s payoff is at most $\pi_{i}^{\textrm{Stac}}$
, and that he cannot gain by deviating. This shows that $\left(\left(\hat{s}{}_{j},Id\right),\left(s_{1}^{*},s_{2}^{*}\right)\right)$
is a (strong)  BBE.
\end{proof}
We demonstrate this class of equilibria in a Cournot competition.
\begin{example}[\emph{Well-behaved BBE that yields the Stackelberg outcome in a Cournot
competition}]
\label{Exam-Cornot-Stackelberg} Consider the symmetric Cournot game
with linear demand in Example \ref{ex-Cournot-Nash-cannot-be-supported-by-identity}:
$G=\left(S,\pi\right)$: $S_{i}=\mathbb{R}^{+}$ and $\pi_{i}\left(s_{i},s_{j}\right)=s_{i}\cdot\left(1-s_{i}-s_{j}\right)$
for each player $i$. Then $\left(\left(0,I_{d}\right),\left(\frac{1}{2},\frac{1}{4}\right)\right)$
is a strong well-behaved BBE that induces the Stackelberg outcome
$\left(\frac{1}{2},\frac{1}{4}\right)$, and yields the Stackelberg-leader
payoff of $\frac{1}{8}$ to player 1 and yields the follower payoff
of $\frac{1}{16}$ to player 2. This is because: (1) $\left(\frac{1}{2},\frac{1}{4}\right)\in NE\left(G_{\left(0,I_{d}\right)}\right)$,
(2) for any biased belief $\psi'_{2}$, player 1 keeps playing $\frac{1}{2}$
and as a result player 2's payoff is at most $\frac{1}{16}$, and
(3) for any biased belief $\psi'_{1}$, player 2 will best-reply to
player's 1 strategy, and thus player 1's payoff will be at most his
Stackelberg payoff of $\frac{1}{8}$.
\end{example}

\subsection{Folk Theorem Results\label{sec:folk-theorem-results}}

In this subsection we present various folk theorem results (i.e.,
general feasibility results) that show that relaxing either of the
two requirements in the definition of a BBE (namely, monotonicity
and ruling out implausible equilibria) yields little predictive power
in various classes of games. Specifically, we show that in those games
a strategy profile is a monotone weak BBE outcome (resp., non-monotone
strong BBE outcome) essentially iff it is (1) undominated, and (2)
induces a payoff above the undominated minmax payoff.

\subsubsection{Preliminary Definitions \label{subsec:prelimanary-definitions}}

We begin by defining the notions of monotone weak BBE, and of non-monotone
strong BBE.
\begin{defn}
A weak BBE $\left(\psi^{*},s^{*}\right)$ is a \emph{monotone weak
BBE} if each biased belief $\psi_{i}^{*}$ is monotone for each player
$i$. \\
A\emph{ weak BBE} $\left(\psi^{*},s^{*}\right)$ is a\emph{ non-monotone
strong} \emph{BBE} if the inequality $\pi_{i}\left(s'_{i},s_{j}'\right)\leq\pi_{i}\left(s_{i}^{*},s_{j}^{*}\right)$
holds for every player $i$, every biased belief $\psi'_{i}$, and
every strategy profile $\left(s'_{i},s_{j}'\right)\in NE\left(G_{\left(\psi'_{i},\psi_{j}^{*}\right)}\right)$.\\
Note that (1) a monotone weak BBE is a weakening of the notion of
a BBE, which relaxes the requirement of ruling out implausible equilibria,
and (2) a non-monotone strong BBE is a weakening of the notion of
strong BBE, which relaxes the requirement of monotonicity. 
\end{defn}

\subsubsection{Folk Theorem Result: Monotone Weak BBE in Finite Games\label{subsec:Folk-Theorem-Result:-Monotone}}

We say that a finite game $G$ admits\emph{ best replies with full
undominated support}, if, for each player $i$, there exists an undominated
strategy $s_{i}\in S_{i}^{U}$ with a support that includes all undominated
actions, i.e., $supp\left(s_{i}\right)=A_{i}\cap S_{i}^{U}$ . Two
classes of games that admit best replies with full undominated support
are:
\begin{enumerate}
\item \emph{All two-action games}. The reason for this is as follows. If
player $i$ has a dominant action, then, trivially, the dominant action
$a_{i}$ is an undominated strategy with a support that includes all
undominated actions. If player $i$ does not have a dominant action,
then there must be a strategy of the opponent for which the player
is indifferent between his two actions, which implies that there exists
an undominated strategy with full support. 
\item Any game with a totally mixed equilibrium (e.g., a rock-paper-scissors
game).
\end{enumerate}
Our next result focuses on finite games that admit\emph{ }best replies
with full undominated support, and shows that in such games a strategy
profile $s^{*}$ is a monotone weak BBE outcome iff (I) $s^{*}$ is
undominated, and (II) the payoff of is $s^{*}$ above the undominated
minmax payoff.

The sketch of the proof is as follows. Each player has a blind belief
that his opponent plays her part of the\emph{ }Nash equilibrium with
full undominated support. This implies that each player is always
indifferent between all undominated actions and, as such, can (1)
play $s_{i}^{*}$ on the equilibrium path, and (2) play a punishing
strategy that guarantees the opponent a payoff of at most her undominated
minmax payoff following any deviation of the opponent. 
\begin{prop}[\emph{Folk Theorem result for monotone weak BBE outcomes}]
\emph{\label{pro-monotone-BBE-outcomes-finite-games}Let} $G$ be
a finite game that admits best replies with full undominated support.
Then the following two statements are equivalent: 
\begin{enumerate}
\item Strategy profile $\left(s_{1}^{*},s_{2}^{*}\right)$ is a monotone
weak BBE outcome. 
\item Strategy profile $\left(s_{1}^{*},s_{2}^{*}\right)$ is (I) undominated
and (II) $\pi_{i}\left(s_{1}^{*},s_{2}^{*}\right)\geq M_{i}^{U}$. 
\end{enumerate}
\end{prop}
\begin{proof}
Proposition \ref{prop-neccesary-conditions} implies that ``1.$\Rightarrow$2.''
We now show that ``2.$\Rightarrow$1.'' Assume that $\left(s_{1}^{*},s_{2}^{*}\right)$
is undominated, and $\pi_{i}\left(s_{1}^{*},s_{2}^{*}\right)\geq M_{i}^{U}$.
For each player $j$, let $s_{j}^{p}$ be an undominated strategy
that guarantees that player $i$ obtains, at most, his minmax payoff
$M_{i}^{U}$, i.e., $s_{j}^{p}=\textrm{argmin}_{s_{j}\in S_{j}^{U}}\left(\max_{s_{i}\in S_{i}}\pi_{i}\left(s_{i},s_{j}\right)\right).$
For each player $j$, let $s_{j}^{e}\in S_{j}^{U}$ be a best-reply
strategy with full undominated support, i.e., $supp\left(s_{j}^{e}\right)=A_{i}\cap S_{i}^{U}$.
For each player $i$, let $s_{i}^{d}\in BR^{-1}\left(s_{j}^{e}\right)$.
The fact that $s_{j}^{e}\in BR\left(s_{i}^{d}\right)$ implies that
$s_{j}^{*},s_{j}^{p}\in\Delta\left(S_{j}^{U}\right)=\Delta\left(supp\left(s_{j}^{e}\right)\right)\subseteq BR\left(s_{i}^{d}\right)$.

We conclude by showing that $\left(\left(s_{1}^{d},s_{2}^{d}\right),\left(s_{1}^{*},s_{2}^{*}\right)\right)$
is a monotone weak BBE (in which both players have blind beliefs).
It is immediate that $\left(s_{1}^{*},s_{2}^{*}\right)\in NE\left(s_{1}^{d},s_{2}^{d}\right)$
. Next, observe that for any deviation of player $i$ to a different
biased belief $\psi_{i}'$, there is a Nash equilibrium of the biased
game $G_{\left(\psi'_{i},s_{j}^{e}\right)}$ in which player $j$
plays $s_{j}^{p}$, and, as a result, player $i$ obtains a payoff
of at most $M_{i}^{U}$, which implies that the deviation is not profitable.
Thus, $\left(s_{1}^{*},s_{2}^{*}\right)$ is a BBE outcome.
\end{proof}
Proposition \ref{pro-monotone-BBE-outcomes-finite-games} suggests
that the notion of monotone weak BBE is too weak. The folk theorem
result relies on the incumbents ``discriminating'' against deviators
who have exactly the same perceived behavior as the rest of the population:
the incumbents of population $j$ ``punish'' deviators by playing
$s_{j}^{p}$ against them, while continuing to play $s_{j}^{*}$ against
the incumbents, even though both the deviators and the incumbents
are perceived to behave the same (i.e., $\psi_{j}^{*}\left(s_{i}^{e}\right)=\psi_{j}^{*}\left(s_{i}^{*}\right)$).

Example \ref{exa-a-a-is-not-monotone-weak-BBE} in Appendix \ref{subsec:The-Folk-Theorem-counter-example-fintie-games}
demonstrates that the folk theorem result does not necessarily hold
for games that do not admit best replies with full undominated support. 

\subsubsection{Folk Theorem Result: Non-Monotone Strong BBE in Interval Games\label{subsec:Folk-Theorem-Result-continous}}

In this section we show a folk theorem result for strong BBE in a
broad family of interval games in which each payoff function $\pi_{i}\left(s_{i},s_{j}\right)$
is (1)  strictly concave in $s_{i}$ and (2) weakly convex in $s_{j}$.
Examples of such games include Cournot competitions, price competitions
with differentiated goods, public good games, and Tullock contests.

The following  result shows that in this class of interval games,
any undominated strategy profile $\left(s_{1}^{*},s_{2}^{*}\right)$
that induces each player a payoff strictly above the player's undominated
minmax payoff can be implemented as an outcome of a strong BBE. Formally:
\begin{prop}
\label{pro-interval-strong-continous-folk}Let $G=\left(S,\pi\right)$
be an interval game. Assume that for each player $i$, $\pi_{i}\left(s_{i},s_{j}\right)$
is strictly concave in $s_{i}$ and weakly convex in $s_{j}$. If
$\left(s_{1}^{*},s_{2}^{*}\right)$ is undominated and $\pi_{i}\left(s_{1}^{*},s_{2}^{*}\right)>M_{i}^{U}$
for each player $i$, then $\left(s_{1}^{*},s_{2}^{*}\right)$ is
a non-monotone strong BBE outcome. 
\end{prop}
The sketch of the proof is as follows (the formal proof is presented
in Appendix \ref{subsec:Proof-of-Proposition-interval}). 

Each player $j$ has a biased belief $\psi_{j}^{*}$ that (I) distorts
$s_{i}^{*}$ into $BR^{-1}\left(s_{j}^{*}\right)$, and (II) distorts
any $s'_{i}$ that is not in a small neighborhood of $s_{i}^{*}$,
to $BR^{-1}\left(s_{j}^{p}\right)$, where $s_{j}^{p}$ is a ``punishing''
strategy that guarantees that player $i$ obtains at most his undominated
minmax payoff. Part (I) implies that $\left(s_{1}^{*},s_{2}^{*}\right)$
is an equilibrium of the biased game. Part (II) implies that following
any deviation of player $i$ to a different biased belief, if player
$i$ plays a strategy that is not in a small neighborhood of $s_{i}^{*}$,
then player $i$ loses from the deviation. Finally, the assumption
that the payoff function $\pi_{i}\left(s_{i},s_{j}\right)$ is convex
in $s_{j}$ implies that we can ``complete'' a continuous description
of $\psi_{j}^{*}$ for $s'_{i}$ that are in a small neighborhood
around $s_{i}^{*}$, such that a player cannot gain from deviating
to playing strategies in this small neighborhood. 

\subsubsection{Discussion of the Folk Theorem Results}

The results of this section show that the notion of weak BBE has little
predictive power in the sense that, essentially, any undominated strategy
profile with a payoff above the undominated minmax payoff is a weak
BBE outcome. Moreover, we show that this multiplicity of BBE outcomes
holds in large classes of games also when applying a refinement of
monotonicity (Prop. \ref{pro-monotone-BBE-outcomes-finite-games}),
or when applying a refinement of strongness (Prop. \ref{pro-interval-strong-continous-folk}).
By contrast, in Section \ref{sec:Main-Results} we show that the combination
of two plausible requirements, namely, monotonicity and ruling out
implausible equilibria, allows us to achieve sharp predictions for
the set of BBE outcomes in various interesting classes of games and
for the set of biased beliefs that support these outcomes.

Our folk theorem results have similar properties to the famous folk
theorem results for repeated games and sufficiently discounted players
(see, e.g., \citealp{fudenberg1986folk}). This is so because it allows
for implicit punishments similar to those used in repeated games in
order to sustain equilibria. This is because our model assumes that
when a player deviates to a different biased belief his opponent can
react to the deviation and deter against it.

Observe that our result has somewhat stronger predictive power than
the folk theorem result for repeated games, in the sense that the
set of monotone weak BBE in one-shot finite games and the set of non-monotone
strong BBE in one-shot interval games are each smaller than the set
of subgame-perfect equilibria of repeated games between patient players.
In particular, the following strategy profiles can be supported as
the subgame-perfect equilibrium outcomes of a repeated game between
patient players, but they cannot be the outcome of a weak BBE outcome
of a one-shot game: (1) strategy profiles in which one of the players
plays a strategy that is strictly undominated in the underlying (one-shot)
game, and (2) strategy profiles in which some of the players obtain
a payoff between the standard minmax payoff and the (higher) undominated
minmax payoff. 

In Appendix \ref{sec:Discontinuous-Biased-Beliefs} we show that if
one relaxes the assumption that the biased beliefs must be continuous,
then one can obtain a folk theorem result in broader classes of games,
namely, (1) in all finite games, and (2) in all interval games with
strictly concave payoffs.

\section{Conclusion\label{sec:Discussion} }

Decision makers\textquoteright{} preferences and beliefs may intermingle.
In strategic environments distorted beliefs can take the form of a
self-serving commitment device. Our paper introduces a formal model
for the persistence  of such beliefs and proposes an equilibrium concept
that supports them. Our analysis characterizes  BBE in a variety of
strategic environments, such as games with strategic complements and
games with strategic substitutes. In particular, we show that agents
present wishful thinking in all BBE in both of these common environments.

Our analysis here deals with simultaneous games of complete information,
but the idea of strategically distorted beliefs may play an important
role also in sequential games and in Bayesian games. In these frameworks,
belief distortion may violate Bayesian updating, and our concept here
can potentially offer a theoretical foundation for some of the cognitive
biases relating to belief updating. It can also potentially identify
the strategic environments in which these biases are likely to occur.
We view this as an important research agenda that we intend to undertake
in the future. 

A different research track that might shed more light on strategic
belief distortion is the experimental one. Laboratory experiments
often conduct belief elicitation with the support of incentives for
truthful revelation. Strong evidence for strategic belief bias in
experimental games can be obtained by showing that players assign
different beliefs to the behavior of their own counterpart in the
game and to a person playing the same role with someone else. In general,
our model predicts that beliefs about a third party\textquoteright s
behavior are more aligned with reality than those involving one\textquoteright s
counterpart in the game. Laboratory experiments can also test whether
specific types of belief distortions (such as wishful thinking) arise
in the strategic environments that are predicted by our model. 

Finally, we point out that strategic beliefs may play an important
role in the design of mechanisms and contracts. Belief distortions
may destroy the desirable equilibrium outcomes that a standard mechanism
aims to achieve. Mechanisms that either induce unbiased beliefs or
adjust the rules of the game to account for possible belief biases
are expected to perform better. 

\bibliographystyle{econometrica}
\bibliography{belief-bias}

\appendix
\newpage{}

\clearpage \pagenumbering{arabic}

\section*{Online Appendices}

\appendix

\section{Additional Examples\label{sec:Additional-Examples}}

\subsection{A Non-Nash Strong BBE Outcome in a Zero-Sum Game\label{subsec:Monotone-Strong-BBE-zero-sum-game-not-NAsh}}

The following example shows that although the weak BBE payoff must
be the Nash equilibrium payoff in a zero-sum game, the strategy profile
sustaining it need not be a Nash equilibrium.
\begin{example}
\label{exa-RPS}Consider the symmetric rock\textendash paper\textendash scissors
zero-sum game described in Table \ref{tab:Symmetric-Rock-Paper-Scissors-Ze}.
\begin{table}[h]
\caption{\label{tab:Symmetric-Rock-Paper-Scissors-Ze}Symmetric Rock-Paper-Scissors
Zero-Sum Game Payoffs}

\centering{}%
\begin{tabular}{|c|c|c|c|}
\hline 
 & R & P & S\tabularnewline
\hline 
\hline 
R & $0,0$ & 0,1 & 1,0\tabularnewline
\hline 
P & 1,0 & $0,0$ & 0,1\tabularnewline
\hline 
S & 0,1 & 1,0 & $0,0$\tabularnewline
\hline 
\end{tabular}
\end{table}
 We show that $\left(\left(I_{d},\left(\frac{1}{3},\frac{1}{3},\frac{1}{3}\right)\right),\left(R,\left(\frac{1}{3},\frac{1}{3},\frac{1}{3}\right)\right)\right)$
is a strong BBE, in which the player 1 (he) has undistorted beliefs
and plays $R$, while player 2 (she) has a blind belief that the opponent
always mixes equally, and she mixes equally. It is immediate that
$\left(R,\left(\frac{1}{3},\frac{1}{3},\frac{1}{3}\right)\right)\in NE\left(G_{\left(I_{d},\left(\frac{1}{3},\frac{1}{3},\frac{1}{3}\right)\right)}\right)$,
and the equilibrium payoff to each player is zero. Next, observe that
after any deviation of player 1 to a biased belief $\psi'_{1}$, there
is an equilibrium of the game $G_{\left(\psi'_{1},\left(\frac{1}{3},\frac{1}{3},\frac{1}{3}\right)\right)}$
in which player 2 mixes equally and player 1 obtains a payoff of zero.
Finally, observe that after any deviation of player 2 to a biased
belief $\psi'_{1}$, player 1 obtains a payoff of at least zero (her
minmax payoff in $G_{\left(I_{d},\psi'_{2}\right)}$) in any Nash
equilibrium in $G_{\left(I_{d},\psi'_{2}\right)}$, which implies
that player 2 obtains a payoff of at most zero, and, as a result,
she does not gain from the deviation.
\end{example}

\subsection{\label{subsec:Prisoner's-Dilemma-with}Prisoner's Dilemma with a
Weakly Dominated Withdrawal Strategy}

Proposition \ref{prop-dominant-action} implies, in particular, that
defection is the unique weak BBE outcome in the prisoner's dilemma
game. The following example demonstrates that a relatively small change
to the prisoner's dilemma game, namely, adding a third weakly dominated
``withdrawal'' strategy that transforms ``cooperation'' into a
weakly dominated strategy, can allow us to sustain cooperation as
a strong BBE outcome. This is done by means of biases under which
a player believes that his opponent is planning to withdraw from the
game whenever he intends to cooperate, which makes cooperation a rational
move. 
\begin{example}
\label{exa-PD-withdrawl-1}Consider the variant of the prisoner's
dilemma game with a third ``withdrawal'' action as described in
Table \ref{tab:Prisoner's-Dilemma-with-1}.
\begin{table}[h]
\centering{}\caption{\label{tab:Prisoner's-Dilemma-with-1}Prisoner's Dilemma Game with
a Withdrawal Action$\protect\underset{}{}$}
\begin{tabular}{|c|c|c|c|}
\hline 
 & \emph{c} & \emph{d} & \emph{w}\tabularnewline
\hline 
\emph{c} & 10,10 & 0,11 & 0,0\tabularnewline
\hline 
\emph{d} & 11,0 & 1,1 & 0,0\tabularnewline
\hline 
\emph{w} & 0,0 & 0,0 & 0,0\tabularnewline
\hline 
\end{tabular}
\end{table}
 In this symmetric game both players get a high payoff of 10 if they
both play action $c$ (interpreted as cooperation). If one player
plays $d$ (\emph{defection}) and his opponent plays $c$, then the
defector gets 11 and the cooperator gets 0. If both players defect,
then each of them gets a payoff of 1. Finally, if either player plays
action $w$ (interpreted as \emph{withdrawal}), then both players
get 0. Observe that defection is a weakly dominant action, and that
the game admits two Nash equilibria: $\left(w,w\right)$ and  $\left(d,d\right)$,
inducing respective symmetric payoffs of zero and one.

We identify a mixed action with a vector $\left(\alpha_{c},\alpha_{d},\alpha_{w}\right)$,
where $\alpha_{c}\geq0$ (resp., $\alpha_{d}\geq0,\,\alpha_{w}\geq0$)
denotes the probability of choosing action $c$ (resp., $d$, $w$).
For each player $i$, let $\psi_{i}$ be the following biased-belief
function: 
\[
\psi_{i}^{*}\left(\alpha_{c},\alpha_{d},\alpha_{w}\right)=\left(0,\alpha_{d},\alpha_{c}+\alpha_{w}\right).
\]
We now show that  $\left(\left(\psi_{1}^{*},\psi_{2}^{*}\right),\left(c,c\right)\right)$
is a non-monotone strong BBE in which both players obtain a high payoff
of 10 (which is strictly better than the best Nash equilibrium payoff,
and strictly better than the Stackelberg payoff of each player). Observe
first that $c\in BR\left(\psi_{i}^{*}\left(c\right)\right)=BR\left(w\right)$,
which implies that $\left(c,c\right)\in NE\left(G_{\left(\psi_{1}^{*},\psi_{2}^{*}\right)}\right)$.
Next, consider a deviation of player $i$ to biased belief $\psi'_{i}$.
Observe that player $i$ can gain a payoff higher than 10 only if
he plays action $d$ with positive probability, but this implies that
the unique best reply of player $j$ to his biased belief about player
$i$'s strategy is defection, which implies that player $i$ obtains
a payoff of at most one. 
\end{example}

\subsection{The Folk Theorem Result Does not Hold for All Finite Games\label{subsec:The-Folk-Theorem-counter-example-fintie-games}}

The following example demonstrates that the folk theorem result (Proposition
\ref{prop-supermodular-monotone}) does not necessarily hold for games
that do not admit best replies with full undominated support. 
\begin{example}
\label{exa-a-a-is-not-monotone-weak-BBE}Consider the three-action
symmetric game described in Table \ref{tab:A-Game-with-no-folk-theorem}.
\begin{table}[h]
\caption{\label{tab:A-Game-with-no-folk-theorem}A Game in which $\left(a,a\right)$
is not a Monotone Weak BBE Outcome}

\centering{}%
\begin{tabular}{|c|c|c|c|}
\hline 
 & \emph{a} & \emph{b} & \emph{c}\tabularnewline
\hline 
\hline 
\emph{a} & 2, 2 & 2, 3 & 1.1, 3\tabularnewline
\hline 
\emph{b} & 3, 2 & 3, 3 & 1, 0\tabularnewline
\hline 
\emph{c} & 3.1, 1 & 0, 1 & 0, 0\tabularnewline
\hline 
\end{tabular}
\end{table}
 Observe that all the actions in the game are undominated, and that
the game does not admit\emph{ }best replies with full undominated
support: there is no strategy of the opponent for which one of the
players has a best reply with full support. This is so because action
$a$ ($c$) is a best reply only to his opponent's strategies that
assign a probability of at least 90\% to action $c$ ($a$), which
implies that actions $a$ and $c$ cannot be best replies simultaneously.
Observe that the undominated minmax payoff of each player is equal
to 1 (because the opponent can play the undominated action $c$, and
by playing this the opponent guarantees that the player gets a payoff
of at most 1).

Consider the undominated action profile $\left(a,a\right)$ (which
induces a payoff strictly above the undominated minmax payoff to each
player). We will show that $\left(a,a\right)$ is not a monotone weak
BBE (which demonstrates that the folk theorem result of Proposition
\ref{pro-monotone-BBE-outcomes-finite-games} does not hold in this
game). Assume to the contrary that $\left(a,a\right)$ is a monotone
weak BBE. Let $\left(\left(\psi_{1}^{*},\psi_{2}^{*}\right),\left(a,a\right)\right)$
be a monotone weak BBE. The fact that $\left(a,a\right)\in NE\left(G_{\left(\psi_{1}^{*},\psi_{2}^{*}\right)}\right)$
implies that $\psi_{1}^{*}\left(a\right)\left(c\right)>90\%$. Consider
a deviation of player 2 to having the blind belief $\psi'_{2}=b$.
Observe that player 2 plays action $b$ in any equilibrium of $G_{\left(\psi_{1}^{*},\psi'_{2}\right)}$.
The monotonicity of $\psi_{1}^{*}$ implies that $\psi_{1}^{*}\left(b\right)\left(a\right)\leq\psi_{1}^{*}\left(a\right)\left(a\right)\leq1-\psi_{1}^{*}\left(a\right)\left(c\right)\leq10\%$,
which implies that the best reply of player 1 to the perceived strategy
of player 2 ($\psi_{1}^{*}\left(b\right)$) does not have action $c$
in its support. This implies that player 1 gains a payoff of at least
3 in any Nash equilibrium of the new biased game $G_{\left(\psi_{1}^{*},\psi'_{2}\right)}$,
which contradicts $\left(\left(\psi_{1}^{*},\psi_{2}^{*}\right),\left(a,a\right)\right)$
being a monotone weak BBE. 
\end{example}

\subsection{Examples of Games with Strategic Complements\label{subsec:Examples-of-Games-complements}}

In this subsection we analyze three examples of games with strategic
complements: input games, stag hunt games, and the traveler's dilemma.

Our first example demonstrates how to implement the undominated Pareto
optimal profile as a strong BBE in an input (or partnership game).
\begin{example}[\emph{Input games}]
\label{exam-partnership-game} Consider the following input game
(closely related games are analyzed in, among others, \citealp{holmstrom1982moral}
and \citealp{heller-sturrock}). Let $S_{i}=S_{j}=\left[0,1\right]$,
and let the payoff function be $\pi_{i}(s_{i},s_{j},\rho)=s_{i}\cdot s_{j}-\frac{s_{i}^{2}}{2\rho}$,
where the parameter $\frac{1}{\rho}$ is interpreted as the cost of
effort. One can show that (1) the best-reply function of each agent
is to exert an effort that is $\rho$<1 times smaller than the opponent's
(i.e., $BR\left(s_{j}\right)=\rho\cdot s_{j}$), (2) in the unique
Nash equilibrium each player exerts no effort $s_{i}=s_{j}=0$, (3)
the highest undominated strategy of each player $i$ is $s_{i}=\rho$,
and (4) the undominated strategy profile $\left(\rho,\rho\right)$
is Nash improving and yields the best payoff to both players out of
all the undominated symmetric strategy profiles. Let $\psi_{i}^{*}$
be the following biased-belief function:
\[
\psi_{i}^{*}\left(s_{j}\right)=\begin{cases}
\frac{s_{j}}{\rho} & s_{j}<\rho\\
1 & s_{j}\geq\rho.
\end{cases}
\]
Observe that $\psi_{i}^{*}$ is monotone and exhibits wishful thinking.
We now show that $\left(\left(\psi_{1}^{*},\psi_{2}^{*}\right),\left(\rho,\rho\right)\right)$
is a strong BBE. Observe that $BR\left(\psi_{i}^{*}\left(s_{j}\right)\right)=BR\left(\frac{s_{j}}{\rho}\right)=s_{j}$
for any $s_{j}\leq\rho$, and that $BR\left(\psi_{i}^{*}\left(s_{j}\right)\right)=BR\left(1\right)=\rho$
for any $s_{j}\geq\rho$. This implies that $\left(\rho,\rho\right)\in NE\left(G_{\left(\psi_{1}^{*},\psi_{2}^{*}\right)}\right)$,
and that for any player $i$, any biased belief $\psi_{i}'$, and
any Nash equilibrium $\left(s'_{1},s'_{2}\right)$ of the biased game
$G_{\left(\psi'_{i},\psi_{j}\right)}$, $s'_{j}=min\left(s_{i}',\rho\right)$.
This implies that $\pi_{i}\left(s'_{1},s'_{2}\right)\leq\pi_{i}\left(\rho,\rho\right)$,
which shows that $\left(\left(\psi_{1}^{*},\psi_{2}^{*}\right),\left(\rho,\rho\right)\right)$
is a strong BBE. Observe that this BBE induces only a small distortion
in the belief of each player, assuming that $\rho$ is sufficiently
close to one: 
\[
\left|\psi_{i}^{*}\left(s_{j}\right)-s_{j}\right|<\left|\frac{s_{j}}{\rho}-s_{j}\right|<\frac{1-\rho}{\rho}.
\]
\end{example}
Our second example characterizes the set of BBE outcomes (and their
supporting beliefs) in stag hunt games.
\begin{example}[\emph{Stag hunt games}]
\label{exa-stag-hunt-analysis} Stag hunt is a two-action game describing
a conflict between safety and social cooperation. Specifically, each
player $i$ has two actions: $s_{i}$ (``stag'') and $h_{i}$ (``hare''),
and his ordinal preferences are $\left(s_{i},s_{j}\right)$ $\succ_{i}$$\left(h_{i},s_{j}\right)$$\succeq_{i}$$\left(h_{i},h_{j}\right)$$\succ_{i}$$\left(s_{i},h_{j}\right)$.
Table \ref{tab:stag-hunt} 
\begin{table}[h]
\caption{\label{tab:stag-hunt}Stag Hunt Game ($g_{1},g_{2}\in\left(0,1\right]$
and $l_{1},l_{2}>0$)$\protect\underset{}{}$}

\centering{}%
\begin{tabular}{|c|c|c|}
\hline 
 & $s_{2}$ & $h_{2}$\tabularnewline
\hline 
$s_{1}$ & $\begin{array}{c}
\\
\\
\end{array}1,1\begin{array}{c}
\\
\\
\end{array}$ & $-l_{1},g_{2}$\tabularnewline
\hline 
$h_{1}$ & $\begin{array}{c}
\\
\\
\end{array}g_{1},-l_{1}\begin{array}{c}
\\
\\
\end{array}$ & $0,0$\tabularnewline
\hline 
\end{tabular}
\end{table}
presents the payoff of a general stag hunt game, where we have normalized,
without loss of generality, the payoff of each player when playing
action profile $\left(s_{i},s_{j}\right)$ ($\left(h_{i},h_{j}\right)$)
to be one (zero), and where each $g_{i}$ is positive and each $l_{i}$
is in the interval $\left(0,1\right)$. A common interpretation of
stag hunt games ($\grave{\textrm{a}}$ la Jean-Jacques Rousseau) is
a situation in which two individuals go hunting. Each can individually
choose to hunt a stag or to hunt a hare. Each player must choose an
action without knowing the choice of the other. If an individual hunts
a stag, he must have the cooperation of his opponent in order to succeed.
An individual can get a hare by himself, but a hare is worth less
than a stag. It is well known that the game admits 3 equilibria: $\left(s_{i},s_{j}\right)$,
$\left(h_{i},h_{j}\right)$, and $\left(\alpha_{1}^{*},\alpha_{2}^{*}\right)$,
with 
\[
\alpha_{i}^{*}=\frac{l_{j}}{l_{j}+\left(1-g_{j}\right)}\in\left(0,1\right),
\]
where each $\alpha_{i}$ represents the probability that player $i$
plays $s_{i}$.

Applying the analysis of the previous section shows that the game
admits 3 classes of BBE:
\begin{itemize}
\item Hunting the hare:$\left(\left(\psi_{1}^{*},\psi_{2}^{*}\right),\left(0,0\right)\right)$,
where each $\psi_{i}^{*}$ is an arbitrary monotone biased belief
that satisfies $\psi_{i}^{*}\left(1\right)\geq\alpha_{i}^{*}$.
\item Hunting the stag. $\left(\left(\psi_{1}^{*},\psi_{2}^{*}\right),\left(1,1\right)\right)$,
where each $\psi_{i}^{*}$ is an arbitrary monotone biased belief
that satisfies $\psi_{i}^{*}\left(1\right)\leq\alpha_{i}^{*}$.
\item Mixing with less weight to hunting the stag, wishful thinking, and
responsiveness to bad news: $\left(\left(\psi_{1}^{*},\psi_{2}^{*}\right),\left(\beta_{1},\beta_{2}\right)\right)$,
where for each player $i$: (1) the payoff is above the minmax payoff:
$\pi_{i}\left(\beta_{i},\beta_{j}\right)\geq0$, (2) the players hunt
the stag less often in the unique Nash equilibrium: $\beta_{i}\in\left(0,\alpha_{i}^{*}\right)$,
(3) wishful thinking: $\psi_{i}^{*}\left(\beta_{j}\right)=\alpha_{j}^{*}>\beta_{j}$,
(4) responsiveness to bad news: $\psi_{i}^{*}\left(\alpha\right)=\alpha_{j}^{*}$
for each $\alpha\geq\beta_{j}$, and $\psi_{i}^{*}\left(\alpha\right)<\alpha_{j}^{*}$
for each $\alpha<\beta_{j}$.
\end{itemize}
Observe that any profile $\left(\beta_{1},\beta_{2}\right)$, where
$\beta_{i}\in\left(\alpha_{i}^{*},1\right)$, cannot be a BBE outcome.
If $\beta_{j}=1$, then player $i$ can gain by deviating to $\psi'_{i}\equiv1$,
as the unique equilibrium of the new biased game is $\left(1,1\right)$,
which induces a higher payoff to player $i$ relative to $\left(\beta_{i},\beta_{j}\right)$.
If $\beta_{j}<1$, then player $j$ can gain by deviating to $\psi'_{j}\equiv1$,
as the only possible equilibria of the new biased game are $\left(1,1\right)$
and $\left(\beta_{i},1\right)$, both of which induce a higher payoff
to player $j$ relative to $\left(\beta_{i},\beta_{j}\right)$. 
\end{example}
Our third example deals with the traveler's dilemma game, in which
each agent has 100 pure ordered actions that have a discrete payoff
structure that resembles strategic complementarity in interval games.
We demonstrate how to implement the undominated Pareto optimal profile
in this game as a strong BBE outcome that presents wishful thinking.
\begin{example}[\emph{Implementing the undominated Pareto optimal profile as a strong
BBE in the traveler's dilemma}]
\label{exam-traveler-dilemma-1}
\end{example}
Consider the following version of the traveler's dilemma game (\citealp{basu1994traveler}).
Each player has 100 actions ($A_{i}=\left\{ 1,...,100\right\} $),
and the payoff function of each player is

\[
\pi_{i}\left(a_{i},a_{j}\right)=\begin{cases}
a_{i}+2 & a_{i}<a_{j}\\
a_{i} & a_{i}=a_{j}\\
a_{j}-2 & a_{i}>a_{j}.
\end{cases}
\]

The interpretation of the game is as follows. Two identical suitcases
have been lost, each owned by one of the players. Each player has
to evaluate the value of his own suitcase. Both players get a payoff
equal to the minimal evaluation (as the suitcases are known to have
identical values), and, in addition, if the evaluations differ, then
the player who gave the lower (higher) evaluation gets a bonus (malus)
of 2 to his payoff. 

It is well known that the unique Nash equilibrium is $\left(1,1\right)$,
which yields a low payoff of one to each player. Observe that the
traveler's dilemma has positive spillovers, in the sense that it is
always weakly better for a player if his opponent chooses a higher
action. The traveler's dilemma has strategic complementarity in the
sense that the best reply of an agent is to stop one stage before
his opponent, and, thus, an agent has an incentive to choose a higher
action if his opponent chooses a higher action.

Observe that action $99$ is the ``highest'' undominated action
of each player (as $99$ is a best reply against $100$, and as action
$100$ is not a best reply against any of the opponent's strategies).
In what follows, we construct a strong BBE exhibiting wishful thinking
that yields a payoff of 99 to each player in the undominated symmetric
Pareto-optimal strategy profile.

We define the biased belief $\psi_{i}^{*}$ as follows: 
\[
\psi_{i}^{*}\left(\alpha_{1},\alpha_{2}...,\alpha_{99},\alpha_{100}\right)=\left(\alpha_{1},\alpha_{2},...,\frac{\alpha_{99}}{2},\frac{\alpha_{99}}{2}+\alpha_{100}\right).
\]
In what follows we show that $\left(\left(\psi_{1}^{*},\psi_{2}^{*}\right),\left(99,99\right)\right)$
is a strong BBE. Observe first that $\psi_{\text{i}}^{*}\left(99\right)=\left(0,...,0,\frac{1}{2},\frac{1}{2}\right)$,
which implies that $99\in BR\left(\psi_{i}^{*}\left(99\right)\right)$,
and, thus, $\left(99,99\right)\in NE\left(G_{\left(\psi_{1}^{*},\psi_{2}^{*}\right)}\right)$.
Let $\psi'_{1}$ be an arbitrary perception bias of player \emph{i}.
Observe that player $i$ never plays action $100$ in a any Nash equilibrium
of any biased game, because action $100$ is not a best reply against
any strategy of player $j$. Next observe that player \emph{i} can
obtain a payoff higher than 99 only if (1) player \emph{j} chooses
action $99$ with a positive probability, and (2) player \emph{i}
chooses action $98$ with a probability strictly higher than his probability
of playing action $100$. However, the biased belief $\psi_{j}^{*}$
of player \emph{j} implies that if player \emph{i} chooses action
$98$ with a probability strictly higher than his probability of playing
$100$, then player \emph{j} never chooses action $99$ in any Nash
equilibrium of the induced biased game because action $99$ yields
a strictly lower payoff to player $j$ than action $98$ against the
perceived strategy of player $i$ (because according to this perceived
strategy, player $i$ plays action $100$ with a probability strictly
less than player $i$'s probability of playing either action $98$
or action $99$).

Note that the BBE equilibrium outcome $\left(99,99\right)$ is consistent
with level-1 behavior in the level-\emph{k} and cognitive hierarchy
literature (see, e.g., \citealp{stahl1994experimental,nagel1995unraveling,costa2001cognition,camerer2004cognitive}),
according to which each agent believes that his opponent is following
a focal non-strategic action (the action 100 in the traveler's dilemma),
and best-replies to this belief. The notion of BBE can help explain
why such level-k behavior induces a strategic advantage in the long
run, and why, therefore, it is likely to emerge in an equilibrium.

\subsection{Hawk-Dove Game\label{subsec:Hawk-dove-game}}

The following example characterizes the set of BBE (and their supporting
beliefs) in a hawk-dove game (which is a game of strategic substitutes).
\begin{example}[\emph{The Hawk-dove game}]
\label{exam-The-hawk-dove-r}The hawk-dove (or ``chicken'') game
is a two-action game in which each player $i$ has two actions: $d_{i}$
(interpreted as a ``dove''-like action of willingness to share a
resource with the opponent) and $h_{i}$ (interpreted as a ``hawk''-like
action of insistence on getting the whole resource, even if this requires
fighting against the opponent), and where the ordinal preferences
of each player $i$ are $\left(h_{i},d_{j}\right)$ (getting the resource)
$\succ$$\left(d_{i},d_{j}\right)$ (sharing the resource) $\succ$$\left(d_{i},h_{j}\right)$
(not getting the the resource) $\succ$$\left(h_{i},h_{j}\right)$
(being involved in a serious fight). Table \ref{tab:Hawk-dove-2-1}
presents the payoff of a general two-action hawk-dove game, where
we have normalized, without loss of generality, the payoff of each
player when playing action profile $\left(d_{i},d_{j}\right)$ ($\left(h_{i},h_{j}\right)$)
to be one (zero), and where each $g_{i}$ positive and each $l_{i}$
is in the interval $\left(0,1\right)$. 
\begin{table}[h]
\caption{\label{tab:Hawk-dove-2-1}Hawk-Dove Game ($g_{1},g_{2}>0$ and $l_{1},l_{2}\in\left(0,1\right)$)$\protect\underset{}{}$}

\centering{}%
\begin{tabular}{|c|c|c|}
\hline 
 & $d_{2}$ & $h_{2}$\tabularnewline
\hline 
$d_{1}$ & $\begin{array}{c}
\\
\\
\end{array}1,1\begin{array}{c}
\\
\\
\end{array}$ & $1-l_{1},1+g_{2}$\tabularnewline
\hline 
$h_{1}$ & $\begin{array}{c}
\\
\\
\end{array}1+g_{1},1-l_{1}\begin{array}{c}
\\
\\
\end{array}$ & $0,0$\tabularnewline
\hline 
\end{tabular}
\end{table}

It is well known that the hawk-dove game admits three equilibria:
two pure equilibria $\left(d_{1},h_{2}\right)$ and $\left(h_{1},d_{2}\right)$,
and one mixed equilibrium $\left(\alpha_{1}^{*},\alpha_{2}^{*}\right)$,
where the probability that player $i$ plays action $\alpha_{i}^{*}$
is 
\[
\alpha_{i}^{*}=\frac{1-l_{j}}{g_{j}+\left(1-l_{j}\right)}\in\left(0,1\right),\,\,\,\,\,\textrm{and}\,\,\,\,\pi\left(\alpha_{i}^{*},\alpha_{j}^{*}\right)=\alpha_{j}^{*}\cdot\left(1+g_{i}\right)=1-\frac{g_{i}}{g_{i}+\left(1-l_{i}\right)}\cdot l_{i}.
\]
The undominated minmax payoff of each player coincides with the minmax
payoff of each player (as there are no dominated actions), and it
is equal to $M_{i}^{U}=1-l_{i},$ which is obtained when the opponent
plays $h_{j}$.

Applying the analysis of the previous section shows that the game
admits 3 classes of BBE:
\begin{itemize}
\item Pure equilibrium hawk-dove:$\left(\left(\psi_{i}^{*},\psi_{j}^{*}\right),\left(0,1\right)\right)$,
where (1) $\psi_{i}^{*}$ is an arbitrary monotone biased belief that
satisfies $\psi_{i}^{*}\left(0\right)\geq\alpha_{i}^{*}$, and (2)
$\psi_{j}^{*}$ is an arbitrary monotone biased belief that satisfies
$\psi_{i}^{*}\left(0\right)\leq\alpha_{i}^{*}$.
\item Mixing (with less weight to playing dove), wishful thinking, and one-directional
blindness:\\
$\left(\left(\psi_{1}^{*},\psi_{2}^{*}\right),\left(\beta_{1},\beta_{2}\right)\right)$,
where for each player $i$: (1) the payoff is above the minmax payoff:
$\pi_{i}\left(\beta_{i},\beta_{j}\right)\geq1-l_{i}$, (2) $\beta_{i}\in\left(0,\alpha_{i}^{*}\right)$
(i.e., agents play dove less often in the unique Nash equilibrium),
(3) wishful thinking: $\psi_{i}^{*}\left(\beta_{j}\right)=\alpha_{j}^{*}>\beta_{j}$,
and (4) responsiveness only to good news: $\psi_{i}^{*}\left(\alpha\right)=\alpha_{j}^{*}$
for each $\alpha\leq\beta_{j}$, and $\psi_{i}^{*}\left(\alpha\right)>\alpha_{j}^{*}$
for each $\alpha>\beta_{j}$.
\end{itemize}
Observe that any profile $\left(\beta_{1},\beta_{2}\right)$ where
$\beta_{i}\in\left(\alpha_{i}^{*},1\right)$ cannot be a BBE outcome.
If $\beta_{j}=1$, then player $i$ can gain by deviating to $\psi'_{i}\equiv1$,
as the unique equilibrium of the new biased game is $\left(0_{i},1_{j}\right)$,
which induces a higher payoff to player $i$ relative to $\left(\beta_{i},\beta_{j}\right)$.
If $\beta_{j}<1$, then player $j$ can gain by deviating into $\psi'_{j}\equiv1$,
as the only possible equilibria of the new biased game are $\left(1_{i},0_{j}\right)$
and $\left(\beta_{i},1_{j}\right)$, both of which induce a higher
payoff to player $j$ relative to $\left(\beta_{i},\beta_{j}\right)$. 
\end{example}

\section{Evolutionary Interpretation of BBE\label{sec:evol-interpetation-of-BBE} }

In this section we present a formal definition of strong BBE that
is exactly analogous to the definition of a stable configuration $\grave{\textrm{a}}$
la \citet{DekelElyEtAl2007Evolution}. This shows that our static
solution concept of strong BBE captures evolutionary stability in
the same way as the solution concepts used in the literature on ``indirect
evolution of preferences.'' Finally, we illustrate a detailed example
of a possible learning dynamic that may result in convergence to strong
BBE.

\subsection{Evolutionary Definition of Strong BBE $\grave{\textrm{a}}$ la \citet{DekelElyEtAl2007Evolution}}

In this subsection we present a definition of a strong BBE that is
completely analogous to the definition of a stable configuration a
la \citet{DekelElyEtAl2007Evolution} (henceforth DEY) for the case
of perfect observability of the opponent's type (i.e., $p=1$ in DEY).

In the adaptation of the notion of stable configuration $\grave{\textrm{a}}$
la \citet{DekelElyEtAl2007Evolution} to our setup we change two aspects
(and only these aspects):
\begin{enumerate}
\item We deal with general two-player games played between two different
populations, rather than DEY's setup that deals with symmetric two-player
games played within a single population.
\item Each agent in DEY's model is endowed with a type that determines the
agent's subjective preferences. By contrast, in our setup each agent
is endowed with a type that determines the agent's monotone biased
belief.
\item We focus on homogeneous configurations. DEY's general definitions
allow one to deal with heterogeneous configurations (in which different
incumbents may have different types). However, their results mainly
deal with homogeneous configurations (in which all incumbents have
the same type). Therefore, to ease notation, we focus on homogeneous
configurations in our adaptation of DEY's definitions.
\end{enumerate}
After adapting DEY's definition of a homogeneous configuration (page
689 in DEY) to the three aspects mentioned above, their definition
is as follows: 
\begin{defn}
A (homogeneous) \emph{configuration} is a pair $\left(\left(\psi_{1}^{*},\psi_{2}^{*}\right),\left(s_{1}^{*},s_{2}^{*}\right)\right)$,
where, for each player $i$, function $\psi_{i}^{*}$ is a monotone
biased belief of player $i$ and $s_{i}^{*}$ is a strategy of player
$i$ satisfying $s_{i}^{*}\in BR\left(\psi_{i}^{*}\left(s_{j}^{*}\right)\right)$.

It is immediate that any monotone weak BBE is a configuration. 

Next, DEY present a notion of a balanced configuration (page 689 in
DEY) that is trivially satisfied by any homogeneous configuration.
\end{defn}
Consider two continuum populations of mass one that follow a configuration
$\left(\left(\psi_{1}^{*},\psi_{2}^{*}\right),\left(s_{1}^{*},s_{2}^{*}\right)\right)$.
Assume that one of these populations (say, population $i$) is invaded
by a small group of $0<\epsilon<<1$ mutants with a different biased
belief $\psi'_{i}\neq\psi_{i}^{*}$. DEY assume that (1) such a mutation
can destabilize a configuration by resulting in the mutants achieving
a higher fitness than the incumbents of the same population\footnote{Under imperfect observability, a mutant can destabilize a configuration
by unraveling the original equilibrium behavior, thereby causing the
incumbents' strategies to substantially diverge following the mutant's
entry into the population. This cannot happen under perfect observability,
as the incumbents can always exhibit the same equilibrium behavior
when being matched against other incumbents (see, page 690 in DEY
for a discussion of focal equilibria).} $i$, and (2) the incumbents continue to play the same behavior among
themselves (what DEY calls ``focal equilibria''). 

Let $\Psi_{i}$ be the set of all biased beliefs of player $i$. Following
DEY (page 690 in DEY) we define $N_{i,\epsilon}\left(\psi_{i}^{*},\psi_{i}'\right)\in\Delta\left(\Psi_{i}\right)$
to be the set of distributions over biased beliefs in population $i$
resulting from entry by no more than $\epsilon$ mutants. Formally,
\[
N_{i,\epsilon}\left(\psi_{i}^{*},\psi_{i}'\right)=\left\{ \mu_{i}'\in\Delta\left(\Psi_{i}\right)|\mu_{i}'=\left(1-\epsilon'\right)\cdot\psi_{i}^{*}+\epsilon'\cdot\psi_{i}',\,\,\epsilon'<\epsilon\right\} .
\]

Given a configuration $\left(\left(\psi_{1}^{*},\psi_{2}^{*}\right),\left(s_{1}^{*},s_{2}^{*}\right)\right)$
and a post-entry distribution of biased beliefs in population $i$
$\tilde{\mu}_{i}\in N_{i,\epsilon}\left(\psi_{i}^{*},\psi_{i}'\right)$,
a \emph{post-entry focal configuration} is a pair $\left(\left(\tilde{\mu}_{i},\psi_{j}^{*}\right),\left(s'_{i},s'_{j}\right)\right),$
where (1) $s'_{i}\in BR\left(\psi_{i}^{'}\left(s'_{j}\right)\right)$
is interpreted as the mutant's strategy, and (2) $s'_{j}\in BR\left(\psi_{j}^{*}\left(s'_{i}\right)\right)$
is interpreted as population j's strategy against the mutants. The
incumbents are assumed to play the same pre-entry strategies $\left(s_{i}^{*},s_{j}^{*}\right)$
when being matched among themselves. Let $B\left(\tilde{\mu}_{i}\right)$
denote the set of all post-entry focal configurations. 

Following DEY (Definition 3 on page 691 in DEY), we define DEY-stability
of a configuration as follows.
\begin{defn}
Configuration $\left(\left(\psi_{1}^{*},\psi_{2}^{*}\right),\left(s_{1}^{*},s_{2}^{*}\right)\right)$
is \emph{DEY-stable} if there exists $\epsilon>0$ such that for every
player $i$, every biased belief $\psi_{i}'$, every post-entry distribution
of biased beliefs $\tilde{\mu}_{i}\in N_{i,\epsilon}\left(\psi_{i}^{*},\psi_{i}'\right)$,
and every post-entry focal configuration $\left(\left(\tilde{\mu}_{i},\psi_{j}^{*}\right),\left(s'_{i},s'_{j}\right)\right),$
the mutants are weakly outperformed relative to the incumbents' payoff
(in their own population), i.e., $\pi_{i}\left(s'_{i},s'_{j}\right)\leq\pi_{i}\left(s_{i}^{*},s{}_{j}^{*}\right).$
\end{defn}

\subsection{Equivalence between the Definitions}

The following result shows that the definition of a stable configuration
coincides with our definition of strong BBE.
\begin{prop}
\label{prop-equivalence-of-strong-BBE-DEY}A configuration $\left(\left(\psi_{1}^{*},\psi_{2}^{*}\right),\left(s_{1}^{*},s_{2}^{*}\right)\right)$
is DEY-stable iff it is a strong BBE.
\end{prop}
\begin{proof}
``If'' part: Let $\left(\left(\psi_{1}^{*},\psi_{2}^{*}\right),\left(s_{1}^{*},s_{2}^{*}\right)\right)$
be a strong BBE. Let $\epsilon>0$, $i\in\left\{ 1,2\right\} $, and
$\psi_{i}'\in\Psi_{i}$. Let $\tilde{\mu}_{i}\in N_{i,\epsilon}\left(\psi_{j}^{*},\psi_{i}'\right)$
be a post-entry distribution of biased beliefs. Let $\left(\left(\tilde{\mu}_{i},\psi_{j}^{*}\right),\left(s'_{i},s'_{j}\right)\right)$
be a post-entry focal configuration. The fact that $\left(\left(\tilde{\mu}_{i},\psi_{j}^{*}\right),\left(s'_{i},s'_{j}\right)\right)$
is a post-entry focal configuration implies that $s'_{i}\in BR\left(\psi_{i}^{'}\left(s'_{j}\right)\right)$
and $s'_{j}\in BR\left(\psi_{j}^{*}\left(s'_{i}\right)\right)$. The
fact that it is a strong BBE implies that $\pi_{i}\left(s'_{i},s'_{j}\right)\leq\pi_{i}\left(s_{i}^{*},s{}_{j}^{*}\right)$,
which shows that $\left(\left(\psi_{1}^{*},\psi_{2}^{*}\right),\left(s_{1}^{*},s_{2}^{*}\right)\right)$
is DEY-stable.

``Only if'' part: Let $\left(\left(\psi_{1}^{*},\psi_{2}^{*}\right),\left(s_{1}^{*},s_{2}^{*}\right)\right)$
be DEY-stable configuration. Let $i\in\left\{ 1,2\right\} $ and $\psi_{i}'\in\Psi_{i}$.
Let $\left(s'_{i},s'_{j}\right)\in NE\left(G_{\left(\psi_{i}',\psi_{j}^{*}\right)}\right)$
be an equilibrium of the new biased game. Let $\epsilon>0$. Let $\tilde{\mu}_{i}\in N_{i,\epsilon}\left(\psi_{i}^{*},\psi_{i}'\right)$
be a post-entry distribution of biased beliefs. For each $\left(s'_{i},s'_{j}\right)\in NE\left(G_{\left(\psi_{i}',\psi_{j}^{*}\right)}\right)$,
let $\left(\left(\tilde{\mu}_{i},\psi_{j}^{*}\right),\left(s'_{i},s'_{j}\right)\right)$
be a post-entry focal configuration. The assumption that $\left(\left(\psi_{1}^{*},\psi_{2}^{*}\right),\left(s_{1}^{*},s_{2}^{*}\right)\right)$
is DEY-stable implies that $\pi_{i}\left(s'_{i},s'_{j}\right)\leq\pi_{i}\left(s_{i}^{*},s{}_{j}^{*}\right).$
This implies that $\left(\left(\psi_{1}^{*},\psi_{2}^{*}\right),\left(s_{1}^{*},s_{2}^{*}\right)\right)$
is a strong BBE.
\end{proof}
\begin{rem*}[Allowing multiple simultaneous invasions of mutants]
 The definition of DEY-stability presented above is unaffected when
various groups of mutants simultaneously invade one of the populations.
By contrast, if one were to require a stable configuration to resist
simultaneous invasions of two groups of mutants, one invasion of each
population, it would require a refinement of the concept of strong
BBE, in the spirit of \citeauthor{smith1973lhe}'s \citeyearpar{smith1973lhe}
notion of evolutionary stability, such that if both $\psi_{1}'$ and
$\psi_{2}'$ are best replies against configuration $\left(\left(\psi_{1}^{*},\psi_{2}^{*}\right),\left(s_{1}^{*},s_{2}^{*}\right)\right)$,
then (1) $\psi_{1}^{*}$ should be a strictly better reply against
$\psi'_{2}$ (relative to $\psi_{1}'$), and (2) $\psi_{2}^{*}$ should
be a strictly better reply against $\psi'_{1}$ (relative to $\psi_{2}'$).
\end{rem*}
Similarly, one can formulate a definition of stability equivalent
to that of monotone BBE by requiring the mutants to be weakly outperformed
in at least one post-entry focal configuration.
\begin{defn}
Configuration $\left(\left(\psi_{1}^{*},\psi_{2}^{*}\right),\left(s_{1}^{*},s_{2}^{*}\right)\right)$
is \emph{weakly stable} if there exists $\epsilon>0$ such that for
every player $i$, every biased belief $\psi_{i}'$, and every post-entry
distribution of biased beliefs $\tilde{\mu}_{i}\in N_{i,\epsilon}\left(\psi_{i}^{*},\psi_{i}'\right)$,
there exists a post-entry focal configuration $\left(\left(\tilde{\mu}_{i},\psi_{j}^{*}\right),\left(s'_{i},s'_{j}\right)\right)$
in which the mutants are weakly outperformed relative to the incumbents'
payoff, i.e., $\pi_{i}\left(s'_{i},s'_{j}\right)\leq\pi_{i}\left(s_{i}^{*},s{}_{j}^{*}\right).$
\end{defn}
The following result shows that the definition of a weakly stable
configuration coincides with our definition of weak BBE. The simple
proof, which is analogous to the proof of \ref{prop-equivalence-of-strong-BBE-DEY},
is omitted for brevity.
\begin{prop}
A configuration $\left(\left(\psi_{1}^{*},\psi_{2}^{*}\right),\left(s_{1}^{*},s_{2}^{*}\right)\right)$
is weakly stable iff it is a monotone weak BBE.
\end{prop}
Finally, one can formulate a definition of stability equivalent to
that of a BBE by requiring the mutants to be weakly outperformed in
at least one plausible post-entry focal configuration.
\begin{defn}
Given configuration $\left(\left(\psi_{1}^{*},\psi_{2}^{*}\right),\left(s_{1}^{*},s_{2}^{*}\right)\right)$,
$\epsilon>0$ , $i\in\left\{ 1,2\right\} $, biased belief $\psi_{i}'$,
and a post-entry distribution of biased beliefs $\tilde{\mu}_{i}\in N_{i,\epsilon}\left(\psi_{i}^{*},\psi_{i}'\right)$,
we say that a post-entry focal configuration $\left(\left(\tilde{\mu}_{i},\psi_{j}^{*}\right),\left(s'_{i},s'_{j}\right)\right)$
is \emph{implausible} if: (1) $\psi_{j}^{*}\left(s'_{i}\right)=\psi_{j}^{*}\left(s_{i}^{*}\right)$,
(2) $s'_{j}\neq s_{j}^{*}$, and (3) $\left(\left(\tilde{\mu}_{i},\psi_{j}^{*}\right),\left(s'_{i},s_{j}^{*}\right)\right)$
is a post-entry focal configuration. A post-entry focal configuration
is\emph{ plausible} if it is not implausible.
\end{defn}
\begin{defn}
Configuration $\left(\left(\psi_{1}^{*},\psi_{2}^{*}\right),\left(s_{1}^{*},s_{2}^{*}\right)\right)$
is \emph{plausibly stable} if there exists $\epsilon>0$ such that
for every player $i$, every biased belief $\psi_{i}'$, and every
post-entry distribution of biased beliefs $\tilde{\mu}_{i}\in N_{i,\epsilon}\left(\psi_{i}^{*},\psi_{i}'\right)$,
there exists a plausible post-entry focal configuration $\left(\left(\tilde{\mu}_{i},\psi_{j}^{*}\right),\left(s'_{i},s'_{j}\right)\right)$
in which the mutants are weakly outperformed relative to the incumbents'
payoff, i.e., $\pi_{i}\left(s'_{i},s'_{j}\right)\leq\pi_{i}\left(s_{i}^{*},s{}_{j}^{*}\right).$
\end{defn}
The following result shows that the definition of a plausibly stable
configuration coincides with our definition of BBE. The simple proof,
which is analogous to the proof of \ref{prop-equivalence-of-strong-BBE-DEY},
is omitted for brevity.
\begin{prop}
A configuration $\left(\left(\psi_{1}^{*},\psi_{2}^{*}\right),\left(s_{1}^{*},s_{2}^{*}\right)\right)$
is \emph{plausibly} stable iff it is a BBE.
\end{prop}

\subsection{Illustration of the Evolutionary Interpretation\label{subsec:Illustration-of-the}}

Similar to DEY, we have presented a reduced-form static notion of
evolutionary stability, without formally modeling a detailed dynamics
according to which the biased beliefs and the strategies co-evolve.
In Section \ref{subsec:Discussion-of-the} we present the essential
features of this evolutionary process, which are analogous to DEY's
essential features (see first paragraph in Section 2.2 in DEY): agents
are endowed by biased beliefs, these biased-beliefs induce equilibrium
behavior in the biased game (presumably by a relatively quick adjustment
of the biased players that leads to equilibrium behavior in the biased
game), behavior determines \textquotedblleft success,\textquotedblright{}
and success (the material payoffs) regulates the evolution of biased
beliefs (presumably by a slow process in which agents occasionally
die and are replaced by new agents who are more likely to mimic the
biased beliefs of more successful incumbents).

In what follows, we illustrate this evolutionary process and its underlying
dynamics in an example. Specifically, we present a strong BBE in an
``input'' game and we illustrate how this strong BBE can persist,
given plausible evolutionary dynamics through which the composition
of the population evolves.
\begin{example}[\emph{Example \ref{exam-partnership-game} revisited}]
\label{exam-partnership-game-2} Consider the following ``input''
game. Let $S_{i}=S_{j}=\left[0,1\right]$, and let the payoff function
be $\pi_{i}(s_{i},s_{j},\rho)=s_{i}\cdot s_{j}-\frac{s_{i}^{2}}{2\rho}$,
where the parameter $\frac{1}{\rho}$ is interpreted as the cost of
effort, and we assume that $\rho\in\left(0.5,1\right)$. One can show
that (1) the best-reply function of each agent is to exert an effort
that is $\rho$ times smaller than the opponent's (i.e., $BR\left(s_{j}\right)=\rho\cdot s_{j}$),
(2) in the unique Nash equilibrium of the unbiased game each player
exerts no effort $s_{i}=s_{j}=0$, and (3) the strategy profile $\left(\rho,\rho\right)$
yields a payoff of $\rho^{2}-\frac{\rho}{2}>0$, which is the highest
symmetric payoff among all strategy profiles in which agents do not
use strictly dominated strategies. Let $\psi_{i}^{*}$ be the following
biased-belief function:
\[
\psi_{i}^{*}\left(s_{j}\right)=\begin{cases}
\frac{s_{j}}{\rho} & s_{j}<\rho\\
1 & s_{j}\geq\rho.
\end{cases}
\]
 In Example \ref{exam-partnership-game} we have shown that $\left(\left(\psi_{1}^{*},\psi_{2}^{*}\right),\left(\rho,\rho\right)\right)$
is a strong BBE. In what follows we illustrate how this strong BBE
can persist. Consider a small group of mutants of population $i$
who have undistorted beliefs. Assume that, initially, the incumbents
of population $j$ use the same strategy against the mutants as they
use against the incumbents of population $i$ (i.e., strategy $\rho$),
and the mutants gradually learn to best reply to the incumbents' behavior
by playing $\rho^{2}$. Recall that we assume that the agents of population
$j$ identify the mutants as a separate group of agents who behave
differently than the rest of population $j$ (without assuming that
the incumbents of population $j$ know anything about the biased beliefs
of the mutants). These incumbents perceive the mutants' play as $\rho$
(due to the incumbents' biased beliefs), and gradually learn to best
reply to this perceived strategy by playing $\rho^{2}$. This, in
turn, induces the mutants to adapt their play to playing $\rho^{3}$,
and, in response, the incumbents of population $j$ adapt their play
against the mutants and play $\rho^{3}$ (the best reply to the mutants'
perceived strategy $\rho^{2}$). This mutual gradual adaptation process
continues until the play in the matches between incumbents of population
$j$ and mutants of population $i$ converges to $\left(0,0\right)$.

Finally, following the convergence of the behavior in the matches
against the mutants to $\left(0,0\right)$, a slow flow of new agents
begins to influence the composition of the population. Each new agent
randomly chooses a mentor among the agents in his own population,
where agents with higher fitness are more likely to be chosen as mentors.
As the mutants get a much lower payoff (0) than the incumbents of
population $i$ ( $\rho^{2}-\frac{\rho}{2}>0$) in the underlying
game, their fitness is expected to be lower, and they are much less
likely to be chosen as mentors. As a result the share of mutants in
the population slowly shrinks until they disappear from the population.
\end{example}

\section{Principal-Agent (Subgame-Perfect) Definition of BBE\label{sec:Principal-Agent}}

In this appendix we present an equivalent definition of BBE as a subgame-perfect
equilibrium of a two-stage game in which in the first round each player
chooses the biased belief of the agent who will play on his behalf
in the second round.

\subsection{The Two-Stage Game $\Gamma_{G}$}

Given an underlying two-player normal-form game $G=\left(S,\pi\right)$
define $\Gamma_{G}$ as the following four-player two-stage extensive-form
game. The four players in the game $\Gamma$ are: principal 1 and
principal 2 (who choose representative agents for the second stage),
agent 1 (who plays on behalf of principal 1 in round 2), and agent
2 (who plays on behalf of principal 2 in round 2). 

The game $\Gamma_{G}$ has 2 stages. In the first stage, the principals
simultaneously choose biased beliefs for their agents. That is, each
principal $i$ chooses a biased belief $\psi_{i}:S_{j}\rightarrow S_{j}$
for agent $i$. In the second stage the agents simultaneously choose
their strategies. That is, each agent $i$ chooses strategy $s_{i}\in S_{i}$.
The payoff of each principal $i$ is $\pi_{i}\left(s_{i},s_{j}\right)$.
The payoff of each agent $i$ is $\pi_{i}\left(\psi_{i}\left(s_{i}\right),s_{j}\right).$
Let $\Psi_{i}$ be the set of all feasible (monotone) biased beliefs
of agent $i$.

A pure strategy profile of $\Gamma_{G}$ (henceforth $\Gamma_{G}$-strategy
profile) is a tuple $\left(\psi_{1},\psi_{2},\sigma_{1},\sigma_{2}\right)$,
where each $\psi_{i}$ is a biased belief, and each $\sigma_{i}:\Psi_{1}\times\Psi_{2}\rightarrow S_{i}$
is a function assigning a strategy to each pair of (monotone) biased
beliefs. Let $SPE\left(\Gamma_{G}\right)$ denote the set of all subgame-perfect
equilibria of $\Gamma$.

\subsection{Subgame-Perfect Definition of Weak BBE}

The following result shows that a weak BBE is equivalent to a subgame-perfect
equilibrium of $\Gamma$. Formally:
\begin{prop}
\label{pro-SPE-equivalence}Let $G$ be a game. Strategy profile $\left(\left(\psi_{1}^{*},\psi_{2}^{*}\right),\left(s_{1}^{*},s_{2}^{*}\right)\right)$
is a weak BBE of $G$ iff there exists a subgame-perfect equilibrium
$\left(\left(\psi_{1}^{*},\psi_{2}^{*}\right),\left(\sigma_{1}^{*},\sigma_{2}^{*}\right)\right)$
of $\Gamma_{G}$ satisfying $\sigma_{i}^{*}\left(\psi_{i}^{*}\right)=s_{i}^{*}$
for each player $i.$
\end{prop}
\begin{proof}
``If side'': Let $\left(\left(\psi_{1}^{*},\psi_{2}^{*}\right),\left(\sigma_{1}^{*},\sigma_{2}^{*}\right)\right)\in SPE\left(\Gamma_{G}\right)$
be a subgame-perfect equilibrium of $\Gamma$ satisfying $\sigma_{i}^{*}\left(\psi_{i}^{*}\right)=s_{i}^{*}$
for each player $i.$ Let $\psi'_{i}$ be a biased belief of player
$i$. Let $s'_{1}=\sigma_{1}^{*}\left(\psi_{i}^{'},\psi_{j}^{*}\right)$
and $s'_{2}=\sigma_{2}^{*}\left(\psi_{i}^{'},\psi_{j}^{*}\right)$.
The fact that $\left(\left(\psi_{1}^{*},\psi_{2}^{*}\right),\left(\sigma_{1}^{*},\sigma_{2}^{*}\right)\right)\in SPE\left(\Gamma_{G}\right)$
implies that (1) $\left(s'_{1},s'_{2}\right)\in NE\left(G_{\left(\psi_{i}^{'},\psi_{j}^{*}\right)}\right)$
and (2) $\pi_{i}\left(s'_{1},s'_{2}\right)\leq\pi_{i}\left(s_{1}^{*},s_{2}^{*}\right)$.
This implies that $\left(\left(\psi_{1}^{*},\psi_{2}^{*}\right),\left(s_{1}^{*},s_{2}^{*}\right)\right)$
is a weak BBE of $G$.

``Only if side'': Let $\left(\left(\psi_{1}^{*},\psi_{2}^{*}\right),\left(s_{1}^{*},s_{2}^{*}\right)\right)$
be a weak BBE of $G$. We define $\left(\sigma_{1}^{*},\sigma_{2}^{*}\right)$
as follows:\footnote{The definition of $\left(\sigma_{1}^{*},\sigma_{2}^{*}\right)$ relies
on the axiom of choice.} (1) $\sigma_{i}^{*}\left(\psi_{1}^{*},\psi_{2}^{*}\right)=s_{i}^{*}$,
(2) for each biased belief $\psi'_{i}\neq\psi_{i}^{*}$, define $\sigma_{i}^{*}\left(\psi'_{i},\psi_{j}^{*}\right)=s'_{i}$
and $\sigma_{j}^{*}\left(\psi'_{i},\psi_{j}^{*}\right)=s'_{j}$ such
that $\left(s'_{i},s'_{j}\right)\in NE\left(G_{\left(\psi'_{i},\psi_{j}^{*}\right)}\right)$
and $\pi_{i}\left(s'_{i},s'_{j}\right)\leq\pi_{i}\left(s_{i}^{*},s{}_{j}^{*}\right)$
(such a pair $\left(s'_{i},s'_{j}\right)$ exists due to $\left(\left(\psi_{1}^{*},\psi_{2}^{*}\right),\left(s_{1}^{*},s_{2}^{*}\right)\right)$
being a weak BBE of G), and (3) for each pair of biased beliefs $\psi'_{i}\neq\psi_{i}^{*}$
and $\psi'_{j}\neq\psi_{j}^{*}$, define $\sigma_{i}^{*}\left(\psi'_{i},\psi'_{j}\right)=s'_{i}$
and $\sigma_{j}^{*}\left(\psi'_{i},\psi'_{j}\right)=s'_{j}$ such
that $\left(s'_{i},s'_{j}\right)\in NE\left(G_{\left(\psi'_{i},\psi'_{j}\right)}\right)$.
The definition of $\left(\sigma_{1}^{*},\sigma_{2}^{*}\right)$ immediately
implies that $\left(\left(\psi_{1}^{*},\psi_{2}^{*}\right),\left(\sigma_{1}^{*},\sigma_{2}^{*}\right)\right)\in SPE\left(\Gamma_{G}\right)$. 
\end{proof}

\subsection{Subgame-Perfect Definition of BBE}

Next, we present an equivalent definition of a BBE as a refinement
of a subgame-perfect equilibrium of $\Gamma_{G}$. Specifically, a
subgame-perfect equilibrium $\left(\psi_{1}^{*},\psi_{2}^{*},\sigma_{1}^{*},\sigma_{2}^{*}\right)$
is required to remain a subgame-perfect equilibrium even after changing
the off-the equilibrium path behavior to a different Nash equilibrium
of the induced subgame in which (I) a single player (say, player $j$)
has deviated to a different biased-belief, (II) the non-deviator perceives
the deviator's strategy in the same way as the original on-the-equilibrium
path opponent's strategy, and (III) the non-deviator changes his behavior
such that after the change it coincides with his on-the-equilibrium
path behavior. Formally,
\begin{defn}
A subgame-perfect equilibrium $\left(\psi_{1}^{*},\psi_{2}^{*},\sigma_{1}^{*},\sigma_{2}^{*}\right)\in SPE\left(\Gamma_{G}\right)$
is a \emph{plausible subgame-perfect equilibrium} if (I) the biased
beliefs $\psi_{1}^{*}$ and $\psi_{2}^{*}$ are monotone, and (II)
$\left(\psi_{1}^{*},\psi_{2}^{*},\sigma'_{1},\sigma'_{2}\right)\in SPE\left(\Gamma\right)$
for each pair of second-stage strategies $\sigma'_{1},\sigma'_{2}$
satisfying: (1) $\left(\psi_{1}^{'},\psi_{2}^{'},\sigma_{1}^{'},\sigma'_{2}\right)\in SPE\left(\Gamma_{G}\right)$
for some pair of first-stage strategies $\left(\psi_{1}^{'},\psi_{2}^{'}\right)$
(i.e., second-stage behavior is consistent with equilibrium behavior
in all subgames) and (2) if $\sigma_{i}^{'}\left(\psi_{1}^{'},\psi_{2}^{'}\right)\neq\sigma_{i}^{*}\left(\psi_{1}^{'},\psi_{2}^{'}\right)$
, then: (I) $\psi_{i}^{'}=\psi_{i}^{*}$ and $\psi_{j}^{'}\neq\psi_{j}^{*}$,
(II) $\psi_{i}^{*}\left(\sigma_{j}^{'}\left(\psi_{1}^{'},\psi_{2}^{'}\right)\right)=\psi_{i}^{*}\left(\sigma_{j}^{*}\left(\psi_{1}^{'},\psi_{2}^{'}\right)\right)$,
and (III) $\sigma_{j}^{'}\left(\psi_{1}^{'},\psi_{2}^{'}\right)=\sigma_{j}^{*}\left(\psi_{1}^{'},\psi_{2}^{'}\right)$.
\end{defn}
\begin{prop}
Let $G$ be a game. Strategy profile $\left(\left(\psi_{1}^{*},\psi_{2}^{*}\right),\left(s_{1}^{*},s_{2}^{*}\right)\right)$
is a BBE of $G$ iff there exists a plausible subgame-perfect equilibrium
$\left(\left(\psi_{1}^{*},\psi_{2}^{*}\right),\left(\sigma_{1}^{*},\sigma_{2}^{*}\right)\right)$
of $\Gamma_{G}$ satisfying $\sigma_{i}^{*}\left(\psi_{i}^{*}\right)=s_{i}^{*}$
for each player $i.$
\end{prop}
The simple proof, which is analogous to the proof of Proposition \ref{pro-SPE-equivalence},
is omitted for brevity.

\subsection{Subgame-Perfect Definition of Strong BBE}

Finally, we present an equivalent definition of a strong BBE as a
refinement of a subgame-perfect equilibrium of $\Gamma$, which remains
an equilibrium even after changing off the equilibrium path in subgames
to other Nash equilibria of the induced subgames. Formally,
\begin{defn}
A subgame-perfect equilibrium $\left(\psi_{1}^{*},\psi_{2}^{*},\sigma_{1}^{*},\sigma_{2}^{*}\right)\in SPE\left(\Gamma_{G}\right)$
is a \emph{strong subgame-perfect equilibrium} if (I) the biased beliefs
$\psi_{1}^{*}$ and $\psi_{2}^{*}$ are monotone, and (II) $\left(\psi_{1}^{*},\psi_{2}^{*},\sigma'_{1},\sigma'_{2}\right)\in SPE\left(\Gamma_{G}\right)$
for each pair of second-stage strategies $\sigma'_{1},\sigma'_{2}$
satisfying: (1) $\left(\psi_{1}^{'},\psi_{2}^{'},\sigma_{1}^{'},\sigma'_{2}\right)\in SPE\left(\Gamma_{G}\right)$
for some pair of first-stage strategies $\psi_{1}^{'},\psi_{2}^{'}$
(i.e., second-stage behavior is consistent with equilibrium behavior
in all subgames) and (2) $\sigma_{i}^{'}\left(\psi_{1}^{*},\psi_{2}^{*}\right)=\sigma_{i}^{*}\left(\psi_{1}^{*},\psi_{2}^{*}\right)$
(i.e., behavior after $\left(\psi_{1}^{*},\psi_{2}^{*}\right)$ is
unchanged).

Our final result shows that a strong BBE is equivalent to a strong
subgame-perfect equilibrium of $\Gamma$. Formally:
\end{defn}
\begin{prop}
Let $G$ be a game. Strategy profile $\left(\left(\psi_{1}^{*},\psi_{2}^{*}\right),\left(s_{1}^{*},s_{2}^{*}\right)\right)$
is a strong BBE of $G$ iff there exists a strong subgame-perfect
equilibrium $\left(\left(\psi_{1}^{*},\psi_{2}^{*}\right),\left(\sigma_{1}^{*},\sigma_{2}^{*}\right)\right)$
of $\Gamma$ satisfying $\sigma_{i}^{*}\left(\psi_{i}^{*}\right)=s_{i}^{*}$
for each player $i.$
\end{prop}
The simple proof, which is analogous to the proof of Proposition \ref{pro-SPE-equivalence},
is omitted for brevity.

\section{Discontinuous Biased Beliefs \label{sec:Discontinuous-Biased-Beliefs}}

In this appendix we present an alternative definition of BBE that
relaxes the assumption that biased beliefs have to be continuous.
We show that all BBE characterized in the main text remain BBE when
deviators are allowed to use discontinuous biased beliefs. 

\subsection{Adapted Definitions: Quasi-equilibria}

We redefine a\emph{ biased belief} $\psi_{i}:S_{j}\rightarrow S_{j}$
to be an arbitrary (rather than continuous) function that assigns
to each strategy of the opponent a (possibly distorted) belief about
the opponent's play. The definition of a configuration $\left(\psi^{*},s^{*}\right)$
is left unchanged (i.e., we require that $\left(s_{i}^{*},s_{j}^{*}\right)\in NE\left(G_{\psi^{*}}\right)$). 

Recall that a configuration is a BBE if each biased belief is a best
reply to the opponent\textquoteright s biased belief, in the sense
that an agent who chooses a different biased belief is weakly outperformed
in the induced equilibrium of the new biased game. Allowing discontinuous
beliefs implies that some biased games $G_{\left(\psi_{1},\psi_{2}\right)}$
in which one (or both) of the biases are discontinuous may not admit
Nash equilibria. This requires us to adapt the definition of a BBE
to deal with behavior in biased games that do not admit Nash equilibria.
We do so by assuming that the resulting behavior in a biased game
that does not admit a Nash equilibrium is a ``$j$-quasi-equilibrium,''
in which the non-deviator (player $j$) best replies to the perceived
behavior of the deviator (player $i$), while the deviator is allowed
to play arbitrarily. Formally:
\begin{defn}
Let $\left(\psi_{i},\psi_{j}\right)$ be a profile of biased beliefs,
and let $j$ be one of the players (interpreted as the non-deviator);
then we define $QE_{j}\left(G_{\left(\psi_{i},\psi_{j}\right)}\right)$
as the set of $j$-quasi-equilibria of the biased game $G_{\left(\psi_{i},\psi_{j}\right)}$
as follows: 
\[
QE_{j}\left(G_{\left(\psi_{i},\psi_{j}\right)}\right)=\begin{cases}
NE\left(G_{\left(\psi_{i},\psi_{j}\right)}\right) & NE\left(G_{\left(\psi_{i},\psi_{j}\right)}\right)\neq\emptyset\\
\left\{ \left(s_{i},s_{j}\right)|s_{j}\in BR\left(\psi_{j}\left(s_{i}\right)\right)\right\}  & NE\left(G_{\left(\psi_{i},\psi_{j}\right)}\right)=\emptyset.
\end{cases}
\]
\end{defn}
Note that any biased game admits a $j$-quasi-equilibrium. 

\subsection{Adapted Definitions: BBE$'$}

We redefine our notions of BBE as follows, and write them as BBE$'$.
In a strong BBE$'$, the deviator (player $i$) is required to be
outperformed in all $j$-quasi-equilibria, and biased beliefs are
required to be monotone. In a weak BBE$'$, the deviator is required
to be outperformed in at least one $j$-quasi-equilibrium. The notion
of a BBE$'$ is in between these two notions. Specifically, in a BBE$'$,
the biased beliefs are required to be monotone, and, in addition,
the deviator (player $i$) is required to be outperformed in at least
one plausible $j$-quasi-equilibrium of the new biased game, where
implausible $j$-quasi-equilibria are defined as follows. We say that
a $j$-quasi-equilibrium of a biased game induced by a deviation of
player $i$ is implausible if (1) player $i$'s strategy is perceived
by the non-deviating player $j$ as coinciding with player $i$'s
original strategy, (2) player $j$ plays differently relative to his
original strategy, and (3) if player $j$ were playing his original
strategy, this would induce a $j$-quasi-equilibrium of the biased
game. That is, implausible $j$-quasi-equilibria are those in which
the non-deviating player $j$ plays differently against a deviator
even though player $j$ has no reason to do so: player $j$ does not
observe any change in player $i$'s behavior, and player $j$'s original
behavior remains an equilibrium of the biased game. Formally:
\begin{defn}
\label{def-implausible-1}Given configuration $\left(\psi^{*},s^{*}\right)$,
deviating player $i$, and biased belief $\psi'_{i}$, we say that
a $j$-quasi-equilibrium of the biased game $\left(s'_{i},s'_{j}\right)\in QE_{j}\left(G_{\left(\psi'_{i},\psi_{j}^{*}\right)}\right)$
is \emph{implausible} if: (1) $\psi_{j}^{*}\left(s'_{i}\right)=\psi_{j}^{*}\left(s^{*}{}_{i}\right)$,
(2) $s_{j}^{*}\neq s'_{j}$, and (3) $\left(s'_{i},s_{j}^{*}\right)\in QE_{j}\left(G_{\left(\psi'_{i},\psi_{j}^{*}\right)}\right)$.
A $j$-quasi-equilibrium is \emph{plausible} if it is not implausible.
Let $PQE_{j}\left(G_{\left(\psi'_{i},\psi_{j}^{*}\right)}\right)$
be the set of all plausible $j$-quasi-equilibria of the biased game
$G_{\left(\psi'_{i},\psi_{j}^{*}\right)}$.
\end{defn}
Note that it is immediate from Definition \ref{def-implausible-1}
and the nonemptiness of $QE_{j}\left(G_{\left(\psi'_{i},\psi_{j}^{*}\right)}\right)$
that $PQE_{j}\left(G_{\left(\psi'_{i},\psi_{j}^{*}\right)}\right)$
is nonempty.
\begin{defn}
\label{Def-A-biased-belief-equilibrium-3}Configuration $\left(\psi^{*},s^{*}\right)$
is:
\end{defn}
\begin{enumerate}
\item \emph{a strong BBE}$'$ if (I) each biased belief $\psi_{i}^{*}$
is monotone, and (II) $\pi_{i}\left(s'_{i},s_{j}'\right)\leq\pi_{i}\left(s_{i}^{*},s_{j}^{*}\right)$
for every player $i$, every biased belief $\psi_{i}'$, and every
$j$-quasi-equilibrium $\left(s'_{i},s_{j}'\right)\in QE_{j}\left(G_{\left(\psi'_{i},\psi_{j}^{*}\right)}\right)$; 
\item \emph{a weak BBE}$'$ if for every player $i$ and every biased belief
$\psi_{i}'$, there exists a $j$-quasi-equilibrium $\left(s'_{i},s_{j}'\right)\in QE_{j}\left(G_{\left(\psi'_{i},\psi_{j}^{*}\right)}\right)$,
such that $\pi_{i}\left(s'_{i},s_{j}'\right)\leq\pi_{i}\left(s_{i}^{*},s_{j}^{*}\right)$;
\item a\emph{ BBE}$'$ if (I) each biased belief $\psi_{i}^{*}$ is monotone,
and (II) for every player $i$ and every biased belief $\psi_{i}'$,
there exists a plausible $j$-quasi-equilibrium $\left(s'_{i},s_{j}'\right)\in PQE_{j}\left(G_{\left(\psi'_{i},\psi_{j}^{*}\right)}\right)$,
such that $\pi_{i}\left(s'_{i},s_{j}'\right)\leq\pi_{i}\left(s_{i}^{*},s_{j}^{*}\right)$. 
\end{enumerate}
It is immediate that any strong BBE$'$ is a BBE$'$, and that any
BBE$'$ is a weak BBE$'$.

(resp., strong, weak) BBE$'$ $\left(\psi^{*},s^{*}\right)$ is continuous
if each biased function $\psi_{i}^{*}$ is continuous. Note, that
deviators are allowed to choose discontinuous biased beliefs.

\subsection{Robustness of BBE to Discontinuous Biased Beliefs}

In what follows we observe that all the BBE that we characterize in
all the results of the paper are also BBE$'$. That is, all of our
BBE are robust to allowing deviators to use discontinuous biased beliefs.
Specifically, any BBE (resp., weak BBE, strong BBE) that is characterized
in any result (or example) in the paper, is a continuous BBE$'$ (resp.,
weak continuous BBE$'$, strong continuous BBE$'$). 

The reason why this observation is true is that in all the arguments
in the proofs of the paper's results for why a configuration $\left(\left(\psi{}_{i}^{*},\psi_{j}^{*}\right),\left(s_{i}^{*},s_{j}^{*}\right)\right)$
is a BBE, when we show that a deviator (player $i$) is outperformed
after deviating to biased belief $\psi'_{i}$ and after the players
play strategy profile $\left(s'_{i},s'_{j}\right)$, we rely only
on the assumption that the non-deviator (player $j$) best replies
to the deviator (i.e., that $s_{j}\in BR\left(\psi_{j}^{*}\left(s'_{i}\right)\right)$,
which is implied by assuming $\left(s'_{i},s'_{j}\right)\in QE_{j}\left(G_{\left(\psi'_{i},\psi_{j}^{*}\right)}\right)$),
and we do not use in any of the arguments the assumption that the
deviator plays a best reply (i.e., we do not rely on $s_{i}\in BR\left(\psi_{i}^{'}\left(s'_{j}\right)\right)$
in any of the proofs).

\section{Partial Observability\label{sec:Partial-Observability} }

Throughout the paper we assume that if an agent deviates to a different
biased belief, then the opponent always observes this deviation. In
this appendix, we relax this assumption, and show that our results
hold also in a setup with partial observability (some results hold
for any level of partial observability, while others hold for a sufficiently
high level of observability).

\subsection{Restricted Biased Games}

Let $p\in\left[0,1\right]$ denote the probability that an agent who
is matched with an opponent who deviates to a different biased belief\emph{
privately} observes the opponent's deviation (henceforth, \emph{observation
probability}). If an agent does not observe the deviation, then he
continues playing his original configuration's strategy.

Our definitions of configuration and biased game remain unchanged.
We now define a restricted biased game $G_{\left(\psi'_{i},\psi_{j}^{*},s_{j}^{*},p\right)}$
as a game with a payoff function according to which (1) each player's
payoff is determined by the opponent's perceived strategy, and (2)
the non-deviator is restricted to playing $s_{j}^{*}$ with probability
$p$ (i.e., when not observing the opponent's deviation). Formally:
\begin{defn}
Given an underlying game $G=\left(S,\pi\right)$, a profile of biased
beliefs $\left(\psi'_{i},\psi_{j}^{*}\right)$, and a strategy $s_{j}^{*}$
of player $j$ (interpreted as the non-deviator), let the \emph{restricted}
\emph{biased game} \emph{$G_{\left(\psi'_{i},\psi_{j}^{*},s_{j}^{*},p\right)}=\left(S,\tilde{\pi}\left(\psi'_{i},\psi_{j}^{*},s_{j}^{*},p\right)\right)$
}be defined as follows:

\[
\tilde{\pi}_{i}\left(\psi'_{i},\psi_{j}^{*},s_{j}^{*},p\right)\left(s_{i},s_{j}\right)=p\cdot\pi_{i}\left(s_{i},\psi'_{i}\left(s_{j}\right)\right)+\left(1-p\right)\cdot\pi_{i}\left(s_{i},\psi'_{i}\left(s_{j}^{*}\right)\right),\,\,\textrm{and}
\]

\[
\tilde{\pi}_{j}\left(\psi'_{i},\psi_{j}^{*},s_{j}^{*},p\right)\left(s_{i},s_{j}\right)=p\cdot\pi_{j}\left(s_{j},\psi_{j}^{*}\left(s_{i}\right)\right)+\left(1-p\right)\pi_{i}\left(s_{j}^{*},\psi_{j}^{*}\left(s_{i}\right)\right).
\]
\end{defn}
A Nash equilibrium of a $p$-restricted biased game is defined in
the standard way. Formally, a pair of strategies $s^{*}=\left(s'_{1},s'_{2}\right)$
is a Nash equilibrium of a restricted biased game $G_{\left(\psi'_{i},\psi_{j}^{*},s_{j}^{*},p\right)}$,
if each $s'_{i}$\emph{ }is a best reply against the perceived strategy
of the opponent, i.e., 
\[
s'_{i}=argmax_{s_{i}\in S_{i}}\left(\tilde{\pi}_{i}\left(\psi'_{i},\psi_{j}^{*},s_{j}^{*},p\right)\left(s_{i},s'_{j}\right)\right).
\]

Let\emph{ $NE\left(G_{\left(\psi'_{i},\psi_{j}^{*},s_{j}^{*}\right)}\right)\subseteq S_{1}\times S_{2}$
}denote the set of all Nash equilibria of the restricted biased game
$G_{\left(\psi'_{i},\psi_{j}^{*},s_{j}^{*},p\right)}$\emph{. }

Observe that the set of strategies of a biased game is convex and
compact, and the payoff function $\tilde{\pi}_{i}\left(\psi'_{i},\psi_{j}^{*},s_{j}^{*},p\right):S_{i}\times S_{j}\rightarrow\mathbb{R}$
is weakly concave in the first parameter and continuous in both parameters.
This implies (due to a standard application of Kakutani's fixed-point
theorem) that each restricted biased game \emph{$G_{\left(\psi'_{i},\psi_{j}^{*},s_{j}^{*},p\right)}$
}admits a Nash equilibrium (i.e., $NE\left(G_{\left(\psi'_{i},\psi_{j}^{*},s_{j}^{*},p\right)}\right)\neq\emptyset$). 

\subsection{$p$-BBE}

We are now ready to define our equilibrium concept. Configuration
$\left(\psi^{*},s^{*}\right)$ is a $p$-BBE if each biased belief
is a best reply to the opponent's biased belief, in the sense that
an agent who chooses a different biased belief is weakly outperformed
in the induced equilibrium of the new restricted biased game. We present
three versions of $p$-BBE that differ with respect to the equilibrium
selection when the new biased game admits multiple equilibria. In
a strong $p$-BBE (I) each biased-belief is monotone, and (II) the
deviator is required to be outperformed in all Nash equilibria of
the new restricted biased game. In a weak BBE, the deviator is required
to be outperformed in at least one equilibrium of the new restricted
biased game. 

The notion of a $p$-BBE is in between these two notions. Specifically,
in at $p$-BBE (I) each biased-belief is monotone, and (II) the deviator
is required to be outperformed in at least one plausible equilibrium
of the new restricted biased game, where implausible equilibria are
defined as follows. We say that a Nash equilibrium of a restricted
biased game induced by a deviation of player $i$ is implausible if
(1) player $i$'s strategy is perceived by the non-deviating player
$j$ as coinciding with player $i$'s original strategy, (2) player
$j$ plays differently relative to his original strategy, and (3)
if player $j$ were playing his original strategy, this would induce
an equilibrium of the biased game. That is, implausible equilibria
are those in which the non-deviating player $j$ plays differently
against a deviator even though player $j$ has no reason to do so:
player $j$ does not observe any change in player $i$'s behavior,
and player $j$'s original behavior remains an equilibrium of the
biased game. Formally:
\begin{defn}
\label{def-implausible-2}Given configuration $\left(\psi^{*},s^{*}\right)$,
deviating player $i$, and biased belief $\psi'_{i}$, we say that
a Nash equilibrium of the restricted biased game $\left(s'_{i},s'_{j}\right)\in NE\left(G_{\left(\psi'_{i},\psi_{j}^{*},s_{j}^{*},p\right)}\right)$
is \emph{implausible} if: (1) $\psi_{j}^{*}\left(s'_{i}\right)=\psi_{j}^{*}\left(s^{*}{}_{i}\right)$,
(2) $s_{j}^{*}\neq s'_{j}$, and (3) $\left(s'_{i},s_{j}^{*}\right)\in NE\left(G_{\left(\psi'_{i},\psi_{j}^{*},s_{j}^{*},p\right)}\right)$.
An equilibrium is \emph{plausible} if it is not implausible. Let $PNE\left(G_{\left(\psi'_{i},\psi_{j}^{*},s_{j}^{*},p\right)}\right)$
be the set of all plausible equilibria of the biased game $G_{\left(\psi'_{i},\psi_{j}^{*},s_{j}^{*},p\right)}$.
\end{defn}
Note that it is immediate from Definition \ref{def-implausible-2}
and the nonemptiness of $NE\left(G_{\left(\psi'_{i},\psi_{j}^{*},s_{j}^{*},p\right)}\right)$
that $PNE\left(G_{\left(\psi'_{i},\psi_{j}^{*},s_{j}^{*},p\right)}\right)$
is nonempty.
\begin{defn}
\label{Def-A-biased-belief-equilibrium-2}Configuration $\left(\psi^{*},s^{*}\right)$
is:
\end{defn}
\begin{enumerate}
\item \emph{a strong $p$-BBE} if (I) each biased belief $\psi_{i}^{*}$
is monotone, and (II) $\pi_{i}\left(s'_{i},s_{j}'\right)\leq\pi_{i}\left(s_{i}^{*},s_{j}^{*}\right)$
for every player $i$, every biased belief $\psi_{i}'$, and every
Nash equilibrium $\left(s'_{i},s_{j}'\right)\in NE\left(G_{\left(\psi'_{i},\psi_{j}^{*},s_{j}^{*},p\right)}\right)$; 
\item \emph{a weak $p$-BBE} if for every player $i$ and every biased belief
$\psi_{i}'$, there exists a Nash equilibrium $\left(s'_{i},s_{j}'\right)\in NE\left(G_{\left(\psi'_{i},\psi_{j}^{*},s_{j}^{*},p\right)}\right)$,
such that $\pi_{i}\left(s'_{i},s_{j}'\right)\leq\pi_{i}\left(s_{i}^{*},s_{j}^{*}\right)$;
\item a\emph{ $p$-BBE} if (I) each biased belief $\psi_{i}^{*}$ is monotone,
and (II) for every player $i$ and every biased belief $\psi_{i}'$,
there exists a plausible Nash equilibrium $\left(s'_{i},s_{j}'\right)\in PNE\left(G_{\left(\psi'_{i},\psi_{j}^{*},s_{j}^{*},p\right)}\right)$,
such that $\pi_{i}\left(s'_{i},s_{j}'\right)\leq\pi_{i}\left(s_{i}^{*},s_{j}^{*}\right)$. 
\end{enumerate}
It is immediate that: (1) any strong $p$-BBE is a $p$-BBE, and that
any $p$-BBE is a weak $p$-BBE, and (2) the definition of $1$-BBE
(resp., weak 1-BBE, strong 1-BBE) coincides with the original definition
of BBE (resp., weak BBE, strong BBE).

\subsection{Extension of Results }

In what follows we sketch how to extend our results to the setup of
partial observability. The adaptations of the proofs are relatively
simple, and, for brevity, we only sketch the differences with respect
to the original proofs.

\subsubsection{Adaptation of Section \ref{sec:Relations-with-Nash} (Nash Equilibria
and BBE Outcomes)}

The example that some Nash equilibria cannot be supported as the outcomes
of weak $P$-BBE with undistorted beliefs can be extended for any
$p>0$.
\begin{example}[\emph{Example \ref{exam-Cournot-first} revisited. Cournot equilibrium
cannot be supported by undistorted beliefs}]
 \label{ex-Cournot-Nash-cannot-be-supported-by-identity-1}Consider
the following symmetric Cournot game with linear demand $G=\left(S,\pi\right)$:
$S_{i}=\left[0,1\right]$ and $\pi_{i}\left(s_{i},s_{j}\right)=s_{i}\cdot\left(1-s_{i}-s_{j}\right)$
for each player $i$. The unique Nash equilibrium of the game is $s_{i}^{*}=s_{j}^{*}=\frac{1}{3}$,
which yields both players a payoff of $\frac{1}{9}.$ Fix observation
probability $p>0$. Assume to the contrary that this outcome can be
supported as a weak $p$-BBE by the undistorted beliefs $\psi_{i}^{*}=\psi_{j}^{*}=I_{d}$.
Fix a sufficiently small $0<\epsilon<<1$. Consider a deviation of
player $1$ to the blind belief $\psi'_{i}\equiv\frac{1}{3}-2\cdot\epsilon$.
Note that this blind belief has a unique best reply: $s'_{i}=\frac{1}{3}+\epsilon$.
The unique equilibrium of the restricted biased game $G_{\left(\psi'_{i},\psi_{j}^{*},s_{j}^{*},p\right)}$
is $s'_{j}=\frac{1}{3}-\frac{\epsilon}{2}$, $s'_{i}=\frac{1}{3}+\epsilon$,
which yields the deviator a payoff of $\frac{1}{9}+\frac{\epsilon}{6}-\frac{\epsilon^{2}}{2}$
with probability $p$ (when his deviation is observed by player 2)
and a payoff of $\frac{1}{9}-\epsilon^{2}$ with probability $1-p$
(when his deviation is not observed by player 2). For a sufficiently
small $\epsilon>0$ the expected payoff of the deviator is strictly
larger than $\frac{1}{9}$.
\end{example}
All the results of Section \ref{sec:Relations-with-Nash} hold for
any observation probability $p\in\left[0,1\right]$ with minor adaptations
to the proofs.
\begin{prop}[Proposition \ref{Prop--Nash-is-BBE} extended]
Let $\left(s_{1}^{*},s_{2}^{*}\right)$ be a (strict) Nash equilibrium
of the game $G=\left(S,\pi\right)$. Let $\psi_{1}^{*}\equiv s_{2}^{*}$
and $\psi_{2}^{*}\equiv s_{1}^{*}$. Then $\left(\left(\psi_{1}^{*},\psi_{2}^{*}\right),\left(s_{1}^{*},s_{2}^{*}\right)\right)$
is a  (strong) $p$-BBE for any $p\in\left[0,1\right]$.
\end{prop}
\begin{claim}[Claim \ref{claim-The-unique-Nash-zero-sum} extended]
\label{claim-zero-sum-p-BBE-payoff}The unique Nash equilibrium payoff
of a zero-sum game is also the unique payoff in any weak $p$-BBE
for any $p\in\left[0,1\right]$.
\end{claim}
\begin{prop}[Proposition \ref{prop-dominant-action} extended]
If a game admits a strictly dominant strategy $s_{i}^{*}$ for player
$i$, then any  weak $p$-BBE outcome is a Nash equilibrium of the
underlying game.
\end{prop}

\subsubsection{Adaptation of Section \ref{sec:Main-Results} (Main Results)}

\paragraph{Adaptation of Subsection \ref{subsec:Preliminary-Result:-Necessary}
(Preliminary Result)}

Minor adaptations of the proof of Proposition \ref{prop-neccesary-conditions}
show that it holds for any $p\in\left[0,1\right]$. Formally:
\begin{prop}
Let $p\in\left[0,1\right]$. If a strategy profile $s^{*}=\left(s_{1}^{*},s_{2}^{*}\right)$
is a weak $p$-BBE outcome, then (1) the profile $s^{*}$ is undominated
and (2) $\pi_{i}\left(s^{*}\right)\geq M_{i}^{U}$.
\end{prop}

\paragraph*{Adaptation of Subsection \ref{sec:Wishful-thinking-and-complements}
(Games with Strategic Complements)}

Minor adaptations to the proofs of the results of Subsection \ref{sec:Wishful-thinking-and-complements}
show that most of these results (namely, part (1) of Proposition \ref{prop-supermodular-monotone}
and Corollaries \ref{cor-no-bad-BBE-outcomes-strategic-complements}
and \ref{cor-wishful-tiniking-with-complements}) hold for any $p\in\left[0,1\right]$,
while part (2) of Proposition \ref{prop-supermodular-monotone} holds
for $p$-s sufficiently close to one. Formally:
\begin{prop}[\emph{Proposition }\ref{prop-supermodular-monotone} extended]
\emph{}Let $G$ be a game with \emph{strategic substitutes and positive
externalities}.\emph{ }
\begin{enumerate}
\item Fix $p\in\left[0,1\right]$. Let $\left(s_{1}^{*},s_{2}^{*}\right)$
be a  $p$-BBE outcome. Then $\left(s_{1}^{*},s_{2}^{*}\right)$ is
(I) undominated, and for each player $i$: (II) $\pi_{i}\left(s_{i}^{*},s_{j}^{*}\right)\geq M_{i}^{U}$,
and (III) $s_{i}^{*}\leq\max\left(BR\left(s_{j}^{*}\right)\right)$
(underinvestment).
\item Let $\left(s_{1}^{*},s_{2}^{*}\right)$ be an undominated profile
satisfying for each player $i$: (II') $\pi_{i}\left(s_{i}^{*},s_{j}^{*}\right)>\tilde{M}_{i}^{U}$,
and (III) $s_{i}^{*}\leq\max\left(BR\left(s_{j}^{*}\right)\right)$.
Then there exists $\bar{p}<1$ such that $\left(s_{1}^{*},s_{2}^{*}\right)$
is a  $p$-BBE outcome for any $p\in\left[\bar{p},1\right]$. \\
Moreover, if $\pi_{i}\left(s_{i},s_{j}\right)$ is strictly concave
then $\left(s_{1}^{*},s_{2}^{*}\right)$ is a strong $p$-BBE outcome
for any $p\in\left[\bar{p},1\right]$.
\end{enumerate}
\end{prop}
\begin{cor}
Fix $p\in\left[0,1\right]$. Let $G$ be a game with strategic complements
and positive externalities with a lowest Nash equilibrium $\left(\underline{s}_{1},\underline{s}_{2}\right)$
satisfying $\underline{s}_{1}<\max\left(S_{i}\right)$ for each player
$i$. Let $\left(s_{1}^{*},s_{2}^{*}\right)$ be a $p$-BBE outcome.
Then $\underline{s}_{i}\leq s_{i}^{*}$ for each player $i$.
\end{cor}
\begin{cor}
Fix $p\in\left[0,1\right]$. Let $G$ be a game with positive externalities
and strategic complements. Let\\
 $\left(\left(\psi_{1}^{*},\psi_{2}^{*}\right),\left(s_{1}^{*},s_{2}^{*}\right)\right)$
be a $p$-BBE. If $s_{i}^{*}\notin\left\{ \min\left(S_{i}\right),\max\left(S_{i}\right)\right\} $,
\emph{then player $i$ exhibits wishful thinking (i.e., $\psi_{i}^{*}\left(s_{j}^{*}\right)\geq s_{j}^{*}$). }
\end{cor}
One can also adapt the examples of Section \ref{sec:Wishful-thinking-and-complements}
(and, similarly, the examples of Sections \ref{sec:Wishful-thinking-and-substitute}
and \ref{sec:Pessimism-in-Games}) to sufficiently high $p$s. 

\paragraph{Adaptation of Section \ref{sec:Wishful-thinking-and-substitute}
(Games With Strategic Substitutes)}

Minor adaptations to the proofs of the results of Subsection \ref{sec:Wishful-thinking-and-substitute}
show that most of these results (namely, part (1) of Proposition \ref{prop-substitues}
and Corollaries \ref{cor-no-efficient-outcome-substitutes} and \ref{cor-wishful-tiniking-with-substitutes})
hold for any $p\in\left[0,1\right]$, while part (2) of Proposition
\ref{prop-substitues} holds for $p$-s sufficiently close to one.
Formally:
\begin{prop}[\emph{Proposition }\ref{prop-substitues} extended]
\emph{}Let $G$ be a game with strategic substitutes and positive
externalities.\emph{ }
\begin{enumerate}
\item Fix $p\in\left[0,1\right]$. Let $\left(s_{1}^{*},s_{2}^{*}\right)$
be a  $p$-BBE outcome. Then $\left(s_{1}^{*},s_{2}^{*}\right)$ is
(I) undominated, and for each player $i$: (II) $\pi_{i}\left(s_{i}^{*},s_{j}^{*}\right)\geq M_{i}^{U}$,
and (III) $s_{i}^{*}\geq\min\left(BR\left(s_{j}^{*}\right)\right)$
(overinvestment).
\item Let $\left(s_{1}^{*},s_{2}^{*}\right)$ be an undominated profile
satisfying for each player $i$: (II') $\pi_{i}\left(s_{i}^{*},s_{j}^{*}\right)>\tilde{M}_{i}^{U}$
, and (III) $s_{i}^{*}\geq\min\left(BR\left(s_{j}^{*}\right)\right)$.
Then there exists $\bar{p}<1$ such that $\left(s_{1}^{*},s_{2}^{*}\right)$
is a  $p$-BBE outcome for any $p\in\left[\bar{p},1\right]$. \\
Moreover, if $\pi_{i}\left(s_{i},s_{j}\right)$ is strictly concave
then $\left(s_{1}^{*},s_{2}^{*}\right)$ is a strong $p$-BBE outcome
for any $p\in\left[\bar{p},1\right]$.
\end{enumerate}
\end{prop}
\begin{cor}
Fix $p\in\left[0,1\right]$. Let $G$ be a game with strategic substitutes
and positive externalities. Let $\left(s_{1}^{*},s_{2}^{*}\right)$
be a BBE outcome. Then, there exists a Nash equilibrium of the underlying
game $\left(s_{1}^{e},s_{2}^{e}\right)$, and a player $i$ such that
$s_{i}^{e}\geq s_{i}^{*}$.
\end{cor}
\begin{cor}
Fix $p\in\left[0,1\right]$. Let $G$ be a game with strategic substitutes
and positive externalities. Let $\left(\left(\psi_{1}^{*},\psi_{2}^{*}\right),\left(s_{1}^{*},s_{2}^{*}\right)\right)$
be a  $p$-BBE. If $s_{i}^{*}\notin\left\{ \min\left(S_{i}\right),\max\left(S_{i}\right)\right\} $,
\emph{then player $i$ exhibits wishful thinking (i.e., $\psi_{i}^{*}\left(s_{j}^{*}\right)\geq s_{j}^{*}$). }
\end{cor}

\paragraph{Adaptation of Section \ref{sec:Wishful-thinking-and-substitute}
(Games With Strategic Opposites)}

Minor adaptations to the proofs of the results of Subsection \ref{sec:Wishful-thinking-and-substitute}
show that most of these results (namely, part (1) of Proposition \ref{prop-strategic-opposites},
and Corollary \ref{cor-pessimism-withstrategic-opposites}) hold for
any $p\in\left[0,1\right]$, while part (2) of Proposition \ref{prop-strategic-opposites}
holds for $p$-s sufficiently close to one. Formally:
\begin{prop}
Let $G$ be a game\emph{ }with positive externalities and strategic
opposites\emph{: $\frac{\partial\pi_{1}\left(s_{1},s_{2}\right)}{\partial s_{1}}>0$
and $\frac{\partial\pi_{2}\left(s_{1},s_{2}\right)}{\partial s_{1}}<0$
for each pair of strategies $s_{1},s_{2}$}.\emph{ }
\begin{enumerate}
\item Fix $p\in\left[0,1\right]$. Let $\left(s_{1}^{*},s_{2}^{*}\right)$
be a  $p$-BBE outcome. Then $\left(s_{1}^{*},s_{2}^{*}\right)$ is
(I) undominated: (II) $\pi_{i}\left(s_{i}^{*},s_{j}^{*}\right)\geq M_{i}^{U}$
for each player $i$, and (III) $s_{1}^{*}\leq\max\left(BR\left(s_{2}^{*}\right)\right)$
and $s_{2}^{*}\geq\min\left(BR\left(s_{1}^{*}\right)\right)$ (underinvestment
of player 1 and overinvestment of player 2).
\item Let $\left(s_{1}^{*},s_{2}^{*}\right)$ be a profile satisfying: (I)
undominated, (II) $\pi_{i}\left(s_{i}^{*},s_{j}^{*}\right)>\tilde{M}_{i}^{U}$
for each player $i$, and (III) $s_{1}^{*}\leq\max\left(BR\left(s_{2}^{*}\right)\right)$
and $s_{2}^{*}\geq\min\left(BR\left(s_{1}^{*}\right)\right)$. Then
there exists $\bar{p}<1$ such that $\left(s_{1}^{*},s_{2}^{*}\right)$
is a $p$-BBE outcome for any $p\in\left[\bar{p},1\right]$. 
\end{enumerate}
\end{prop}
\begin{cor}
Fix $p\in\left[0,1\right]$. Let $\left(\left(\psi_{1}^{*},\psi_{2}^{*}\right),\left(s_{1}^{*},s_{2}^{*}\right)\right)$
be a  $p$-BBE of a game with \emph{positive externalities} and strategic
opposites (i.e.,\emph{ $\frac{\partial\pi_{1}\left(s_{1},s_{2}\right)}{\partial s_{1}}>0$
and $\frac{\partial\pi_{2}\left(s_{1},s_{2}\right)}{\partial s_{1}}<0$
for each pair of strategies $s_{1},s_{2}$}). If $s_{i}^{*}\notin\left\{ \min\left(S_{i}\right),\max\left(S_{i}\right)\right\} $,
\emph{then player $i$ exhibits pessimism (i.e., $\psi_{i}^{*}\left(s_{j}^{*}\right)\leq s_{j}^{*}$). }
\end{cor}

\subsubsection{Adaptation of Section \ref{sec:Additional-Results} (Additional Results)\label{subsec:Adaptation-of-the}}

\paragraph{Adaptation of Subsection \ref{sec:BBE-and-Undominated} (BBE with
Strategic Stubbornness)}

In what follows we show how to extend Example \ref{Exam-Cornot-Stackelberg}
to the setup of partial observability (while we leave the extension
of the general result, Proposition \ref{Prop-Stackelberg}, to future
research). The example focuses on Cournot competition. We show that
for each level of partial observability $p\in\left[0,1\right]$, there
exists a strong BBE in which one of the players (1) has a blind belief
and (2) plays a strategy that is between the Nash equilibrium strategy
and the Stackelberg strategy (and the closer it is to the Stackelberg
strategy, the higher the value of $p$), while the opponent has undistorted
beliefs. The first player's (resp., opponent's) payoff is strictly
increasing (resp., decreasing) in $p$: it converges to the Nash equilibrium
payoff when $p\rightarrow0$, and it converges to the Stackelberg
leader's (resp., follower's) payoff when $p\rightarrow1$.
\begin{example}[\emph{Example \ref{Exam-Cornot-Stackelberg} revisited}]
 Consider the symmetric Cournot game with linear demand: $G=\left(S,\pi\right)$:
$S_{i}=\mathbb{R}^{+}$ and $\pi_{i}\left(s_{i},s_{j}\right)=s_{i}\cdot\left(1-s_{i}-s_{j}\right)$
for each player $i$. Let $p\in\left[0,1\right]$ be the observation
probability. Then 
\[
\left(\left(\frac{1-p}{3-p},I_{d}\right),\left(\frac{1}{3-p},\frac{2-p}{2\cdot\left(3-p\right)}\right)\right)
\]
 is a strong BBE that yields a payoff of $\frac{2-p}{2\cdot\left(3-p\right)}$
to player 1, and yields a payoff of $\left(\frac{2-p}{2\cdot\left(3-p\right)}\right)^{2}$to
player 2. Observe that player 1's payoff is increasing in $p$, and
it converges to the Nash equilibrium (resp., Stackelberg leader's)
payoff of $\frac{1}{9}$ ($\frac{1}{8}$) when $p\rightarrow0$ ($p\rightarrow1$).
Further observe that player 2's payoff is decreasing in $p$, and
it converges to the Nash equilibrium (resp., Stackelberg follower's)
payoff of $\frac{1}{9}$ ($\frac{1}{16}$) when $p\rightarrow0$ ($p\rightarrow1$).
The argument that $\left(\left(\frac{1-p}{3-p},I_{d}\right),\left(\frac{1}{3-p},\frac{2-p}{2\cdot\left(3-p\right)}\right)\right)$
is a strong BBE is sketched as follows: (1) $\left\{ \left(\frac{1}{3-p},\frac{2-p}{2\cdot\left(3-p\right)}\right)\right\} =NE\left(G_{\left(\frac{1-p}{3-p},I_{d}\right)}\right)$
(because $\frac{1}{3-p}$ is the unique best reply against $\frac{1-p}{3-p}$
and $\frac{2-p}{2\cdot\left(3-p\right)}$ is the unique best reply
against $\frac{1}{3-p}$) ; (2) for any biased belief $\psi'_{2}$,
player 1 keeps playing $\frac{1}{3-p}$ due to having a blind belief,
and as a result player 2's payoff is at most $\left(\frac{2-p}{2\cdot\left(3-p\right)}\right)^{2}$;
and (3) for any biased belief $\psi'_{1}$ inducing a deviating player
$1$ to play strategy $x$, player 2 plays $\frac{1-x}{2}$ (the unique
best reply against $x)$ with probability $p$ (when observing the
deviation), and player 2 plays $\frac{2-p}{2\cdot\left(3-p\right)}$
(the original configuration strategy) with the remaining probability
of $1-p$. Thus, the payoff of a deviating player 1 who deviates into
playing strategy $x$ is 
\[
\pi\left(x\right):=p\cdot x\cdot\left(\frac{1-x}{2}\right)+\left(1-p\right)\cdot x\cdot\left(1-x-\frac{2-p}{2\cdot\left(3-p\right)}\right)=\left(1-\frac{p}{2}\right)\cdot x\cdot\left(1-x\right)-\frac{\left(2-p\right)\cdot\left(1-p\right)}{2\cdot\left(3-p\right)}\cdot x,
\]
where this payoff function $\pi\left(x\right)$ is strictly concave
in $x$ with a unique maximum at $x=\frac{1}{3-p}$ (the unique solution
to the FOC $0=\frac{\partial\pi}{\partial x}=\left(1-\frac{p}{2}\right)\cdot\left(1-2\cdot x\right)-\frac{\left(2-p\right)\cdot\left(1-p\right)}{2\cdot\left(3-p\right)}$).
\end{example}

\paragraph{Extending the Folk Theorem Results for Sufficiently High $p$-s}

The main results of Subsection \ref{sec:folk-theorem-results}, show
folk theorem results for: (1) monotone BBE in games that admit best
replies with full undominated support, and (2) strong BBE in interval
games with a payoff function that is strictly concave in the agent's
strategy, and weakly convex in the opponent's strategy. Minor adaptations
of each proof can show that each result can be extended to $p$-s
that are sufficiently close to one. Formally:
\begin{prop}[\emph{Proposition }\ref{pro-monotone-BBE-outcomes-finite-games} extended]
\emph{\label{pro-extneded-folk-finite}Let} $G$ be a finite game
that admits best replies with full undominated support. Let $\left(s_{1}^{*},s_{2}^{*}\right)$
be an undominated strategy profile that induces for each player a
payoff above his minmax payoff (i.e., $\pi_{i}\left(s_{1}^{*},s_{2}^{*}\right)>M_{i}^{U}$
$\forall i\in\left\{ 1,2\right\} $). Then there exists $\bar{p}<1$,
such that $\left(s_{1}^{*},s_{2}^{*}\right)$ is a monotone weak $p$-BBE
outcome for each $p\in\left[\bar{p},1\right]$.
\end{prop}
\begin{prop}[\emph{Proposition }\ref{pro-interval-strong-continous-folk} extended]
\emph{\label{pro-extended-interval-game}}Let $G=\left(S,\pi\right)$
be an interval game. Assume that for each player $i$, $\pi_{i}\left(s_{i},s_{j}\right)$
is strictly concave in $s_{i}$ and weakly convex in $s_{j}$. If
$\left(s_{1}^{*},s_{2}^{*}\right)$ is undominated and $\pi_{i}\left(s_{1}^{*},s_{2}^{*}\right)>M_{i}^{U}$
for each player $i$, then there exists $\bar{p}<1$, such that $\left(s_{1}^{*},s_{2}^{*}\right)$
is a strong $p$-BBE outcome for each $p\in\left[\bar{p},1\right]$.
\end{prop}
\begin{proof}[Sketch of adapting the proofs of Propositions \emph{\ref{pro-extneded-folk-finite}
and \ref{pro-extended-interval-game} }to the setup of partial observability]
 Observe that the gain of an agent who deviates to a different biased
belief, when his deviation is unobserved by the opponent, is bounded
(due to the payoff of the underlying game being bounded). When the
deviation is observed by the opponent, the agent is strictly outperformed,
given the BBE constructed in the proofs of Propositions\emph{ }\ref{pro-monotone-BBE-outcomes-finite-games}\emph{
}and\emph{ }\ref{pro-interval-strong-continous-folk}\emph{. }This
implies that there exists $\bar{p}<1$ sufficiently close to one,
such that the loss of a mutant when being observed by his opponent
outweighs the mutant's gain when being unobserved for any $p\in\left[\bar{p},1\right]$.
\end{proof}

\section{Proofs \label{sec:Proofs}}

\subsection{Proof of Proposition \ref{prop-supermodular-monotone}\label{subsec:Proof-of-Proposition-supermodular}}

\textbf{Part 1: }Proposition \ref{prop-neccesary-conditions} implies
(I) and (II). It remains to show (III, overinvestment). Let $\left(\left(\psi_{i}^{*},\psi_{j}^{*}\right),\right.$
$\left.\left(s_{i}^{*},s_{j}^{*}\right)\right)$ be a BBE. Assume
to the contrary that $s_{i}^{*}<\min\left(BR\left(s_{j}^{*}\right)\right)$.
Consider a deviation of player $i$ to a blind belief that the opponent
always plays strategy $s_{j}^{*}$(i.e., $\psi'_{i}\equiv s_{j}^{*}$).
Let $\left(s'_{i},s'_{j}\right)\in PNE\left(G_{\left(\psi'_{i},\psi_{j}^{*}\right)}\right)$
be a plausible equilibrium of the new biased game. Observe first that
$s'_{i}\in BR\left(\psi'_{i}\left(s'_{j}\right)\right)=BR\left(s_{j}^{*}\right).$
This implies that $s'_{i}>s_{i}^{*}$, and, thus, due to the monotonicity
of $\psi_{j}^{*}$ we have: $\psi_{j}^{*}\left(s'_{i}\right)\geq\psi_{j}^{*}\left(s_{i}^{*}\right)$.
We consider two cases: 
\begin{enumerate}
\item If $\psi_{j}^{*}\left(s'_{i}\right)>\psi_{j}^{*}\left(s_{i}^{*}\right)$,
then the strategic complementarity implies that $s'_{j}\geq\min\left(BR\left(\psi_{j}^{*}\left(s'_{i}\right)\right)\right)\geq\max\left(BR\left(\psi_{j}^{*}\left(s_{i}^{*}\right)\right)\right)\geq s_{j}^{*}$,
and this, in turn, implies that player $i$ strictly gains from his
deviation: $\pi_{i}\left(s'_{i},s'_{j}\right)\geq\pi_{i}\left(s'_{i},s{}_{j}^{*}\right)>\pi_{i}\left(s_{i}^{*},s{}_{j}^{*}\right)$,
a contradiction. 
\item If $\psi_{j}^{*}\left(s'_{i}\right)=\psi_{j}^{*}\left(s_{i}^{*}\right)$,
then $\left(s'_{i},s_{j}^{*}\right)\in PNE\left(G_{\left(\psi'_{i},\psi_{j}^{*}\right)}\right)$
and $\pi_{i}\left(s'_{i},s{}_{j}^{*}\right)>\pi_{i}\left(s_{i}^{*},s{}_{j}^{*}\right)$,
which contradicts that $\left(\left(\psi_{i}^{*},\psi_{j}^{*}\right),\left(s_{i}^{*},s_{j}^{*}\right)\right)$
is a BBE.
\end{enumerate}
\textbf{Part 2:} Assume\textbf{} that strategy profile $\left(s_{1}^{*},s_{2}^{*}\right)$
satisfies I, II, and III. For each player $i$ let $s_{i}^{e}=\min\left(BR^{-1}\left(s_{i}^{*}\right)\right)$.\textbf{
}For every player $i$ and every strategy $s_{i}<s_{i}^{*}$ define
$X\left(s_{i}\right)$ as the set of strategies $s'_{i}$ for which
player $i$ is worse off (relative to $\pi_{i}\left(s_{1}^{*},s_{2}^{*}\right)$)
if he plays strategy $s_{i}$, while player $j$ plays a best-reply
to $s'_{i}$. Formally: 
\[
X_{s^{*}}\left(s_{i}\right)=\left\{ s'_{i}\in S_{i}|\pi_{i}\left(s_{i},s_{j}\right)\leq\pi_{i}\left(s_{i}^{*},s_{j}^{*}\right)\,\,\forall s_{j}\in BR\left(s'_{i}\right)\right\} .
\]
The assumption that $\pi_{i}\left(s_{i}^{*},s_{j}^{*}\right)>\tilde{M}_{i}^{U}$
implies that $X_{s^{*}}\left(s_{i}\right)$ is nonempty for each $s_{i}$.
The assumption of strategic complements implies that $X_{s^{*}}\left(s_{i}\right)$
is an interval starting at $\min\left(S_{i}\right)$. Let $\phi_{s^{*}}\left(s_{i}\right)=\sup\left(X_{s^{*}}\left(s_{i}\right)\right)$.
The assumption that the payoff function is continuously twice differentiable
implies that $\phi_{s^{*}}\left(s_{i}\right)$ is continuous. The
assumption that $s_{j}^{e}=\min\left(BR^{-1}\left(s_{i}^{*}\right)\right)$
implies that $\lim_{s_{i}\nearrow s_{i}^{*}}\left(\phi_{s^{*}}\left(s_{i}\right)\right)=s_{i}^{e}$.
These observations imply that for each player $j$ there exists a
monotone biased belief $\psi_{j}^{*}$ satisfying (1) $\psi_{j}^{*}\left(s_{i}\right)=s_{i}^{e}$
for each $s_{i}\geq s_{i}^{*}$ and (2) $\psi_{j}^{*}\left(s_{i}\right)\leq\phi_{s^{*}}\left(s_{i}\right)$
for each $s_{i}<s_{i}^{*}$ with an equality only if $\phi_{s^{*}}\left(s_{i}\right)=\min\left(S_{i}\right)$. 

We now show that these properties of $\left(\psi_{1}^{*},\psi_{2}^{*}\right)$
imply that $\left(\left(\psi_{1}^{*},\psi_{2}^{*}\right),\left(s_{1}^{*},s_{2}^{*}\right)\right)$
is a BBE (a strong BBE if $\pi_{i}\left(s_{i},s_{j}\right)$ is strictly
concave in $s_{i}$). Consider a deviation of player $i$ into an
arbitrary biased belief $\psi'_{i}$. For each $s'_{i}\geq s_{i}^{*}$,
and each $\left(s'_{i},s_{j}'\right)\in PNE\left(G_{\left(\psi'_{i},\psi_{j}^{*}\right)}\right)$
($\left(s'_{i},s_{j}'\right)\in NE\left(G_{\left(\psi'_{i},\psi_{j}^{*}\right)}\right)$),
the fact that $\psi_{j}^{*}\left(s'_{i}\right)=\psi_{j}^{*}\left(s_{i}^{*}\right)$
implies that $s_{j}'=s_{j}^{*}$, and due to assumption (III) of overinvestment
and the concavity of the payoff function: $\pi_{i}\left(s'_{i},s'_{j}\right)=\pi_{i}\left(s'_{i},s{}_{j}^{*}\right)\leq\pi_{i}\left(s_{i}^{*},s{}_{j}^{*}\right)$
. For each $s'_{i}<s_{i}^{*}$, and each $\left(s'_{i},s_{j}'\right)\in NE\left(G_{\left(\psi'_{i},\psi_{j}^{*}\right)}\right)$,
the fact that $\psi_{j}^{*}\left(s'_{i}\right)\leq\phi_{s^{*}}\left(s'_{i}\right)$
with an equality only if $\phi_{s^{*}}\left(s'_{i}\right)=\min\left(S_{i}\right)$
(and, thus, $\psi_{j}^{*}\left(s'_{i}\right)\in X_{s^{*}}\left(s'_{i}\right))$
implies that $\pi_{i}\left(s'_{i},s'_{j}\right)\leq\pi_{i}\left(s_{1}^{*},s_{2}^{*}\right)$.
This shows that player $i$ cannot gain from his deviation, which
implies that $\left(\left(\psi_{1}^{*},\psi_{2}^{*}\right),\left(s_{1}^{*},s_{2}^{*}\right)\right)$
is a (strong) BBE. 

\subsection{\label{subsec-proof-og-core-no-worse-than-worst}Proof of a Lemma
Required for Corollary \ref{cor-no-bad-BBE-outcomes-strategic-complements}}

\begin{lem}
\label{lemma-no-worse-than-worset-Nash}Let $G$ be a game with strategic
complements and positive externalities with a lowest Nash equilibrium
$\left(\underline{s}_{1},\underline{s}_{2}\right)$ satisfying $\underline{s_{i}}<\max\left(S_{i}\right)_{i}$
for each player $i$. Let $s_{1}^{*}<\underline{s}_{1}$. Then for
each $s_{2}^{*}\in S_{2}$ either (1) $s_{1}^{*}<\min\left(BR\left(s_{2}^{*}\right)\right)$
or (2) $s_{2}^{*}<\min\left(BR\left(s_{1}^{*}\right)\right)$.
\end{lem}
\begin{proof}
Assume first that $s_{2}^{*}>\underline{s}_{2}$. The fact that $\underline{s}_{1}\in BR\left(\underline{s}_{2}\right)$
and $s_{2}^{*}>\underline{s}_{2}$, together with the strategic complements,
imply that $s_{1}^{*}<\underline{s}_{1}<\min\left(BR\left(s_{2}^{*}\right)\right)$.
We are left with the case where $s_{2}^{*}\leq\underline{s}_{2}.$
Consider a restricted game in which the set of strategies of each
player $i$ is restricted to being strategies that are at most $s_{i}^{*}$.
The game is a game of strategic complements, and, thus, it admits
a pure Nash equilibrium $\left(s'_{i},s'_{j}\right)$. The minimality
of $\left(\underline{s}_{1},\underline{s}_{2}\right)$ implies that
$\left(s'_{i},s'_{j}\right)$ cannot be a Nash equilibrium of the
unrestricted game. The strategic complements and the concavity of
the payoff jointly imply that if $\left(s'_{i},s'_{j}\right)$ is
not a Nash equilibrium of the unrestricted game, then there is player
$i$ for which $s_{i}^{*}=s'_{i}<\min\left(BR\left(s'_{j}\right)\right)\leq\min\left(BR\left(s_{j}^{*}\right)\right)$. 
\end{proof}

\subsection{\label{subsec:Proof-of-cor-wishful-complements}Proof of a Lemma
Required for Corollary \ref{cor-wishful-tiniking-with-complements}}

\begin{lem}
\label{lemma-wishful-somplements}Let $G$ be a game with positive
externalities and strategic complementarity of the payoff of player
$i$ (i.e., $\frac{\partial^{2}\pi_{i}\left(s_{i},s_{j}\right)}{\partial s_{i}\cdot\partial s_{j}}>0$
for each $s_{i},s_{j}$). Then $s'_{j}<s_{j}$ implies that $\max\left(BR\left(s'_{j}\right)\right)\leq\min\left(BR\left(s_{j}\right)\right)$
with an equality only if $\max\left(BR\left(s'_{j}\right)\right)=\min\left(BR\left(s_{j}\right)\right)\in\left\{ \min\left(S_{i}\right),\max\left(S_{i}\right)\right\} $.
\end{lem}
\begin{proof}
The inequality $s'_{j}<s_{j}$ and the strategic complementarity of
the payoff of player $i$ implies that $\frac{\partial\pi_{i}\left(s_{i},s_{j}'\right)}{\partial s_{i}}<\frac{\partial\pi_{i}\left(s_{i},s_{j}\right)}{\partial s_{i}}$
for each $s_{i}\in S_{i}$, which implies that whenever $\max\left(BR\left(s'_{j}\right)\right)\notin\left\{ \min\left(S_{i}\right),\max\left(S_{i}\right)\right\} $,
then
\begin{align*}
\max\left(BR\left(s'_{j}\right)\right) & =\max\left\{ s_{i}^{*}\in S_{i}|\left.\frac{\partial\pi_{i}\left(s_{i},s_{j}'\right)}{\partial s_{i}}=0\right|_{s_{i}=s_{i}^{*}}\right\} \\
 & <\min\left\{ s_{i}^{*}\in S_{i}|\left.\frac{\partial\pi_{i}\left(s_{i},s_{j}\right)}{\partial s_{i}}=0\right|_{s_{i}=s_{i}^{*}}\right\} \leq\min\left(BR\left(s_{j}\right)\right).
\end{align*}
This shows that the strict inequality holds whenever $\max\left(BR\left(s'_{j}\right)\right)\notin\left\{ \min\left(S_{i}\right),\max\left(S_{i}\right)\right\} $.
It remains to show that the weak inequality (namely, $\max\left(BR\left(s'_{j}\right)\right)\leq\min\left(BR\left(s_{j}\right)\right)$)
holds when $\max\left(BR\left(s'_{j}\right)\right)\in\left\{ \min\left(S_{i}\right),\max\left(S_{i}\right)\right\} $.
If $\max\left(BR\left(s'_{j}\right)\right)=\min\left(S_{i}\right)$
then this is immediate. Assume that $\max\left(BR\left(s'_{j}\right)\right)=\max\left(S_{i}\right)$.
Then:
\begin{align*}
\max\left(S_{i}\right) & =\max\left\{ s_{i}^{*}\in S_{i}|\left.\frac{\partial\pi_{i}\left(s_{i},s_{j}'\right)}{\partial s_{i}}\geq0\right|_{s_{i}=s_{i}^{*}}\right\} .\\
 & \leq\min\left\{ s_{i}^{*}\in S_{i}|\left.\frac{\partial\pi_{i}\left(s_{i},s_{j}\right)}{\partial s_{i}}\geq0\right|_{s_{i}=s_{i}^{*}}\right\} \leq\min\left(BR\left(s_{j}\right)\right).
\end{align*}
\end{proof}

\subsection{Proof of Proposition \ref{prop-substitues}\label{subsec:Proof-of-Proposition-substitutes}}

The proof is analogous to the proof of Proposition \ref{prop-supermodular-monotone},
and is presented for completeness.

\textbf{Part 1: }Proposition \ref{prop-neccesary-conditions} implies
(I) and (II). It remains to show (III) (underinvestment). Let $\left(\left(\psi_{i}^{*},\psi_{j}^{*}\right),\left(s_{i}^{*},s_{j}^{*}\right)\right)$
be a BBE. Assume to the contrary that $s_{i}^{*}>\max\left(BR\left(s_{j}^{*}\right)\right)$.
Consider a deviation of player $i$ to a blind belief that the opponent
always plays strategy $s_{j}^{*}$(i.e., $\psi'_{i}\equiv s_{j}^{*}$).
Let $\left(s'_{i},s'_{j}\right)\in PNE\left(G_{\left(\psi'_{i},\psi_{j}^{*}\right)}\right)$
be a plausible equilibrium of the new biased game. Observe first that
$s'_{i}\in BR\left(\psi'_{i}\left(s'_{j}\right)\right)=BR\left(s_{j}^{*}\right).$
This implies that $s'_{i}<s_{i}^{*}$, and, thus, due to the monotonicity
of $\psi_{j}^{*}$ we have: $\psi_{j}^{*}\left(s'_{i}\right)\leq\psi_{j}^{*}\left(s_{i}^{*}\right)$.
We consider two cases: 
\begin{enumerate}
\item If $\psi_{j}^{*}\left(s'_{i}\right)<\psi_{j}^{*}\left(s_{i}^{*}\right)$,
then the strategic substitutability implies that $s'_{j}\geq\min\left(BR\left(\psi_{j}^{*}\left(s'_{i}\right)\right)\right)\geq\max\left(BR\left(\psi_{j}^{*}\left(s_{i}^{*}\right)\right)\right)\geq s_{j}^{*}$,
and this, in turn, implies that player $i$ strictly gains from his
deviation: $\pi_{i}\left(s'_{i},s'_{j}\right)\geq\pi_{i}\left(s'_{i},s{}_{j}^{*}\right)>\pi_{i}\left(s_{i}^{*},s{}_{j}^{*}\right)$,
a contradiction. 
\item If $\psi_{j}^{*}\left(s'_{i}\right)=\psi_{j}^{*}\left(s_{i}^{*}\right)$,
then $\left(s'_{i},s_{j}^{*}\right)\in PNE\left(G_{\left(\psi'_{i},\psi_{j}^{*}\right)}\right)$
and $\pi_{i}\left(s'_{i},s{}_{j}^{*}\right)>\pi_{i}\left(s_{i}^{*},s{}_{j}^{*}\right)$,
which contradicts that $\left(\left(\psi_{i}^{*},\psi_{j}^{*}\right),\left(s_{i}^{*},s_{j}^{*}\right)\right)$
is a BBE.
\end{enumerate}
\textbf{Part 2:} Assume that strategy profile $\left(s_{1}^{*},s_{2}^{*}\right)$
satisfies I, II, and III. For each player $i$ let $s_{i}^{e}=\max\left(BR^{-1}\left(s_{i}^{*}\right)\right)$.\textbf{
}For each player $i$ and each strategy $s_{i}>s_{i}^{*}$ define
$X\left(s_{i}\right)$ as the set of strategies $s'_{i}$ for which
player $i$ is worse off (relative to $\pi_{i}\left(s_{1}^{*},s_{2}^{*}\right)$)
if he plays strategy $s_{i}$, while player $j$ plays a best-reply
to $s'_{i}$. Formally: 
\[
X_{s^{*}}\left(s_{i}\right)=\left\{ s'_{i}\in S_{i}|\pi_{i}\left(s_{i},s_{j}\right)\leq\pi_{i}\left(s_{i}^{*},s_{j}^{*}\right)\,\,\forall s_{j}\in BR\left(s'_{i}\right)\right\} .
\]
The assumption that $\pi_{i}\left(s_{i}^{*},s_{j}^{*}\right)>\tilde{M}_{i}^{U}$
implies that $X_{s^{*}}\left(s_{i}\right)$ is nonempty for each $s_{i}$.
The assumption of strategic substitutes implies that $X_{s^{*}}\left(s_{i}\right)$
is an interval ending at $\max\left(S_{i}\right)$. Let $\phi_{s^{*}}\left(s_{i}\right)=\textrm{inf}\left(X_{s^{*}}\left(s_{i}\right)\right)$.
The assumption that the payoff function is continuously twice differentiable
implies that $\phi_{s^{*}}\left(s_{i}\right)$ is continuous. The
assumption that $s_{j}^{e}=\max\left(BR^{-1}\left(s_{i}^{*}\right)\right)$
implies that $\lim_{s_{i}\searrow s_{i}^{*}}\left(\phi_{s^{*}}\left(s_{i}\right)\right)=s_{i}^{e}$.
These observations imply that for each player $j$ there exists a
monotone biased belief $\psi_{j}^{*}$ satisfying (1) $\psi_{j}^{*}\left(s_{i}\right)=s_{i}^{e}$
for each $s_{i}\leq s_{i}^{*}$ and (2) $\psi_{j}^{*}\left(s_{i}\right)\geq\phi_{s^{*}}\left(s_{i}\right)$
for each $s_{i}>s_{i}^{*}$ with an equality only if $\phi_{s^{*}}\left(s_{i}\right)=\max\left(S_{i}\right)$. 

We now show that these properties of $\left(\psi_{1}^{*},\psi_{2}^{*}\right)$
imply that $\left(\left(\psi_{1}^{*},\psi_{2}^{*}\right),\left(s_{1}^{*},s_{2}^{*}\right)\right)$
is a BBE (a strong BBE if $\pi_{i}\left(s_{i},s_{j}\right)$ is strictly
concave in $s_{i}$). Consider a deviation of player $i$ to an arbitrary
biased belief $\psi'_{i}$. For each $s'_{i}\leq s_{i}^{*}$, and
each $\left(s'_{i},s_{j}'\right)\in PNE\left(G_{\left(\psi'_{i},\psi_{j}^{*}\right)}\right)$
($\left(s'_{i},s_{j}'\right)\in NE\left(G_{\left(\psi'_{i},\psi_{j}^{*}\right)}\right)$),
the fact that $\psi_{j}^{*}\left(s'_{i}\right)=\psi_{j}^{*}\left(s_{i}^{*}\right)$
implies that $s_{j}'=s_{j}^{*}$ and, due to assumption (III) of underinvestment
and the concavity of the payoff function, it follows that $\pi_{i}\left(s'_{i},s'_{j}\right)=\pi_{i}\left(s'_{i},s{}_{j}^{*}\right)\leq\pi_{i}\left(s_{i}^{*},s{}_{j}^{*}\right)$
. For each $s'_{i}>s_{i}^{*}$, and each $\left(s'_{i},s_{j}'\right)\in NE\left(G_{\left(\psi'_{i},\psi_{j}^{*}\right)}\right)$,
the fact that $\psi_{j}^{*}\left(s'_{i}\right)\geq\phi_{s^{*}}\left(s'_{i}\right)$
with an equality only if $\phi_{s^{*}}\left(s_{i}\right)=\max\left(S_{i}\right)$
(and, thus, $\psi_{j}^{*}\left(s'_{i}\right)\in X_{s^{*}}\left(s'_{i}\right))$
implies that $\pi_{i}\left(s'_{i},s'_{j}\right)\leq\pi_{i}\left(s_{1}^{*},s_{2}^{*}\right)$.
This shows that player $i$ cannot gain from his deviation, which
implies that $\left(\left(\psi_{1}^{*},\psi_{2}^{*}\right),\left(s_{1}^{*},s_{2}^{*}\right)\right)$
is a (strong) BBE.

\subsection{\label{subsec:Proof-of-Lemma-3}Proof of a Lemma Required for Corollary
\ref{cor-no-efficient-outcome-substitutes}}
\begin{lem}
\label{lemma-no-better-than-all-Nash-both-players}Let $G$ be a game
with strategic substitutes and positive externalities. Let $\left(s_{1}^{*},s_{2}^{*}\right)$
be a strategy profile satisfying $s_{i}^{*}>s_{i}^{e}$ for each player
$i$ and each Nash equilibrium $\left(s_{1}^{e},s_{2}^{e}\right)\in NE\left(G\right)$.Then,
either (1) $s_{1}^{*}>\max\left(BR\left(s_{2}^{*}\right)\right)$
or (2) $s_{2}^{*}>\max\left(BR\left(s_{1}^{*}\right)\right)$.
\end{lem}
\begin{proof}
Consider a restricted game in which the set of strategies of each
player $i$ is restricted to being strategies that are at least $s_{i}^{*}$.
The restricted game is a game with strategic substitutes, and, thus,
it admits a pure Nash equilibrium $\left(s'_{1},s'_{2}\right)$ (recall,
that after relabeling the set of strategies of one of the players,
the game becomes supermodular, and because of this the game admits
a pure Nash equilibrium due to \citealp{milgrom1990rationalizability}).
The assumption that $s_{i}^{*}>s_{i}^{e}$ for each player $i$ and
each Nash equilibrium $\left(s_{1}^{e},s_{2}^{e}\right)\in NE\left(G\right)$
implies that $\left(s'_{1},s'_{2}\right)$ cannot be a Nash equilibrium
of the unrestricted game. The concavity of the payoff and the strategic
substitutes jointly imply that if $\left(s'_{i},s'_{j}\right)$ is
not a Nash equilibrium of the unrestricted game, then there is a player
$i$ for which $s_{i}^{*}=s'_{i}>\max\left(BR\left(s'_{j}\right)\right)\geq\max\left(BR\left(s_{j}^{*}\right)\right)$. 
\end{proof}

\subsection{\label{subsec:Proof-of-Corollary-wishful-substitutes}Proof of Corollary
\ref{cor-wishful-tiniking-with-substitutes}}

The proof is analogous to Corollary \ref{cor-wishful-tiniking-with-complements},
and is presented for completeness. Assume to the contrary that $\psi_{i}^{*}\left(s_{j}^{*}\right)<s_{j}^{*}$.
Lemma \ref{lemma-wishful-substitute} (below) implies that $\min\left(BR\left(\psi_{i}^{*}\left(s_{j}^{*}\right)\right)\right)\geq\max\left(BR\left(s_{j}^{*}\right)\right)$
with an equality only if 
\[
\min\left(BR\left(\psi_{i}^{*}\left(s_{j}^{*}\right)\right)\right)\in\left\{ \min\left(S_{i}\right),\max\left(S_{i}\right)\right\} .
\]
Part 1 of Proposition \ref{prop-substitues} and the definition of
a monotone BBE imply that 
\[
\min\left(BR\left(\psi_{i}^{*}\left(s_{j}^{*}\right)\right)\right)\leq s_{i}^{*}\leq\max\left(BR\left(s_{j}^{*}\right)\right).
\]
 The previous inequalities jointly imply that 
\[
\min\left(BR\left(\psi_{i}^{*}\left(s_{j}^{*}\right)\right)\right)=s_{i}^{*}=\max\left(BR\left(s_{j}^{*}\right)\right)\in\left\{ \min\left(S_{i}\right),\max\left(S_{i}\right)\right\} ,
\]
 which contradicts the assumption that $s_{i}^{*}\notin\left\{ \min\left(S_{i}\right),\max\left(S_{i}\right)\right\} $.
\begin{lem}
\label{lemma-wishful-substitute}Let $G$ be a game with positive
externalities and strategic substitutability of the payoff of player
$i$ (i.e., $\frac{\partial^{2}\pi_{i}\left(s_{i},s_{j}\right)}{\partial s_{i}\cdot\partial s_{j}}>0$
for each $s_{i},s_{j}$). Then $s'_{j}<s_{j}$ implies that $\min\left(BR\left(s'_{j}\right)\right)\geq\max\left(BR\left(s_{j}\right)\right)$
with an equality only if $\min\left(BR\left(s'_{j}\right)\right)=\min\left(BR\left(s_{j}\right)\right)\in\left\{ \min\left(S_{i}\right),\max\left(S_{i}\right)\right\} $.
\end{lem}
\begin{proof}
The proof is analogous to the proof of Lemma \ref{lemma-no-worse-than-worset-Nash},
and is presented for completeness. The inequality $s'_{j}<s_{j}$
and the strategic substitutability of the payoff of player $i$ implies
that $\frac{\partial\pi_{i}\left(s_{i},s_{j}'\right)}{\partial s_{i}}>\frac{\partial\pi_{i}\left(s_{i},s_{j}\right)}{\partial s_{i}}$
for each $s_{i}\in S_{i}$, which implies that whenever $\min\left(BR\left(s'_{j}\right)\right)\notin\left\{ \min\left(S_{i}\right),\max\left(S_{i}\right)\right\} $,
then
\begin{align*}
\min\left(BR\left(s'_{j}\right)\right) & =\min\left\{ s_{i}^{*}\in S_{i}|\left.\frac{\partial\pi_{i}\left(s_{i},s_{j}'\right)}{\partial s_{i}}=0\right|_{s_{i}=s_{i}^{*}}\right\} \\
 & >\max\left\{ s_{i}^{*}\in S_{i}|\left.\frac{\partial\pi_{i}\left(s_{i},s_{j}\right)}{\partial s_{i}}=0\right|_{s_{i}=s_{i}^{*}}\right\} =\max\left(BR\left(s_{j}\right)\right).
\end{align*}
This shows that the strict inequality holds whenever $\min\left(BR\left(s'_{j}\right)\right)\notin\left\{ \min\left(S_{i}\right),\max\left(S_{i}\right)\right\} $.
It remains to show that the weak inequality (namely, $\min\left(BR\left(s'_{j}\right)\right)\geq\max\left(BR\left(s_{j}\right)\right)$)
holds when $\min\left(BR\left(s'_{j}\right)\right)\in\left\{ \min\left(S_{i}\right),\max\left(S_{i}\right)\right\} $.
If $\min\left(BR\left(s'_{j}\right)\right)=\max\left(S_{i}\right)$
then this is immediate. Assume that $\min\left(BR\left(s'_{j}\right)\right)=\min\left(S_{i}\right)$.
Then:
\begin{align*}
\min\left(S_{i}\right) & =\min\left\{ s_{i}^{*}\in S_{i}|\left.\frac{\partial\pi_{i}\left(s_{i},s_{j}'\right)}{\partial s_{i}}\geq0\right|_{s_{i}=s_{i}^{*}}\right\} \\
 & \geq\max\left\{ s_{i}^{*}\in S_{i}|\left.\frac{\partial\pi_{i}\left(s_{i},s_{j}\right)}{\partial s_{i}}\geq0\right|_{s_{i}=s_{i}^{*}}\right\} \geq\max\left(BR\left(s_{j}\right)\right).
\end{align*}
\end{proof}

\subsection{Proof of Proposition \ref{prop-strategic-opposites}\label{subsec:Proof-of-Proposition-opposites}}

The proof is analogous to the proof of Proposition \ref{prop-supermodular-monotone},
and is presented for completeness.

\textbf{Part 1: }Proposition \ref{prop-neccesary-conditions} implies
(I) and (II). It remains to show (III). Let $\left(\left(\psi_{1}^{*},\psi_{2}^{*}\right),\left(s_{1}^{*},s_{2}^{*}\right)\right)$
be a BBE. We begin by showing overinvestment of player 2. Assume to
the contrary that $s_{2}^{*}<\min\left(BR\left(s_{j}^{*}\right)\right)$.
Consider a deviation of player $2$ to a blind belief that the opponent
always plays strategy $s_{1}^{*}$ (i.e., $\psi'_{2}\equiv s_{1}^{*}$).
Let $\left(s'_{1},s'_{2}\right)\in PNE\left(G_{\left(\psi_{1}^{*},\psi'_{2}\right)}\right)$
be a plausible equilibrium of the new biased game. Observe first that
$s'_{2}\in BR\left(\psi'_{2}\left(s'_{1}\right)\right)=BR\left(s_{1}^{*}\right).$
This implies that $s'_{2}>s_{2}^{*}$, and, thus, due to the monotonicity
of $\psi_{1}^{*}$, we have : $\psi_{1}^{*}\left(s'_{2}\right)\geq\psi_{1}^{*}\left(s_{2}^{*}\right)$.
We consider two cases: 
\begin{enumerate}
\item If $\psi_{1}^{*}\left(s'_{2}\right)>\psi_{1}^{*}\left(s_{2}^{*}\right)$,
then the strategic complementarity of player 1's payoff implies that
$s'_{1}\geq\min\left(BR\left(\psi_{1}^{*}\left(s'_{2}\right)\right)\right)\geq\max\left(BR\left(\psi_{1}^{*}\left(s_{2}^{*}\right)\right)\right)\geq s_{1}^{*}$,
and, this, in turn, implies that player $2$ strictly gains from his
deviation: $\pi_{2}\left(s'_{1},s'_{2}\right)\geq\pi_{2}\left(s'_{1},s{}_{2}^{*}\right)>\pi_{2}\left(s_{1}^{*},s{}_{2}^{*}\right)$,
a contradiction. 
\item If $\psi_{1}^{*}\left(s'_{2}\right)=\psi_{1}^{*}\left(s_{2}^{*}\right)$,
then $\left(s_{1}^{*},s'_{2}\right)\in PNE\left(G_{\left(\psi_{1}^{*},\psi'_{2}\right)}\right)$
and $\pi_{2}\left(s_{1}^{*},s'_{2}\right)>\pi_{2}\left(s_{1}^{*},s{}_{2}^{*}\right)$,
which contradicts that $\left(\left(\psi_{1}^{*},\psi_{2}^{*}\right),\left(s_{1}^{*},s_{2}^{*}\right)\right)$
is a BBE.
\end{enumerate}
Next we show underinvestment of player 1. Assume to the contrary that
$s_{1}^{*}>\max\left(BR\left(s_{2}^{*}\right)\right)$. Consider a
deviation of player $1$ to a blind belief that the opponent always
plays strategy $s_{2}^{*}$ (i.e., $\psi'_{1}\equiv s_{2}^{*}$).
Let $\left(s'_{1},s'_{2}\right)\in PNE\left(G_{\left(\psi'_{1},\psi_{2}^{*}\right)}\right)$
be a plausible equilibrium of the new biased game. Observe first that
$s'_{1}\in BR\left(\psi'_{1}\left(s'_{2}\right)\right)=BR\left(s_{2}^{*}\right).$
This implies that $s'_{1}<s_{2}^{*}$ and, thus, due to the monotonicity
of $\psi_{2}^{*}$, we have: $\psi_{2}^{*}\left(s'_{1}\right)\leq\psi_{2}^{*}\left(s_{1}^{*}\right)$.
We consider two cases: 
\begin{enumerate}
\item If $\psi_{2}^{*}\left(s'_{1}\right)<\psi_{2}^{*}\left(s_{1}^{*}\right)$,
then the strategic substitutability of player 2's payoff implies that
$s'_{2}\geq\min\left(BR\left(\psi_{2}^{*}\left(s'_{1}\right)\right)\right)\geq\max\left(BR\left(\psi_{2}^{*}\left(s_{1}^{*}\right)\right)\right)\geq s_{2}^{*}$,
and this, in turn, implies that player $1$ strictly gains from his
deviation: $\pi_{1}\left(s'_{1},s'_{2}\right)\geq\pi_{1}\left(s'_{1},s{}_{2}^{*}\right)>\pi_{1}\left(s_{1}^{*},s{}_{2}^{*}\right)$,
a contradiction. 
\item If $\psi_{2}^{*}\left(s'_{1}\right)=\psi_{2}^{*}\left(s_{1}^{*}\right)$,
then $\left(s'_{1},s_{2}^{*}\right)\in PNE\left(G_{\left(\psi'_{1},\psi_{2}^{*}\right)}\right)$
and $\pi_{1}\left(s'_{1},s{}_{2}^{*}\right)>\pi_{1}\left(s_{1}^{*},s{}_{2}^{*}\right)$,
which contradicts that $\left(\left(\psi_{1}^{*},\psi_{2}^{*}\right),\left(s_{1}^{*},s_{2}^{*}\right)\right)$
is a BBE.
\end{enumerate}
\textbf{Part 2:} Assume that strategy profile $\left(s_{1}^{*},s_{2}^{*}\right)$
satisfies I, II, and III. Let $s_{1}^{e}=\min\left(BR^{-1}\left(s_{2}^{*}\right)\right)$.\textbf{
}For each strategy $s_{2}<s_{2}^{*}$ define $X\left(s_{2}\right)$
as the set of strategies $s'_{2}$ for which player $2$ is worse
off (relative to $\pi_{2}\left(s_{1}^{*},s_{2}^{*}\right)$) if he
plays strategy $s_{2}$, while player $1$ plays a best-reply to $s'_{2}$.
Formally: 
\[
X_{s^{*}}\left(s_{2}\right)=\left\{ s'_{2}\in S_{2}|\pi_{2}\left(s_{1},s_{2}\right)\leq\pi_{2}\left(s_{i}^{*},s_{j}^{*}\right)\,\,\forall s_{1}\in BR\left(s'_{2}\right)\right\} .
\]
The assumption that $\pi_{2}\left(s_{1}^{*},s_{2}^{*}\right)>\tilde{M}_{2}^{U}$
implies that $X_{s^{*}}\left(s_{2}\right)$ is nonempty for each $s_{2}\in S_{2}$.
The assumption of strategic complements of player 1's payoff implies
that $X_{s^{*}}\left(s_{2}\right)$ is an interval starting at $\min\left(S_{2}\right)$.
Let $\phi_{s^{*}}\left(s_{2}\right)=\sup\left(X_{s^{*}}\left(s_{2}\right)\right)$.
The assumption that the payoff function is continuously twice differentiable
implies that $\phi_{s^{*}}\left(s_{2}\right)$ is continuous. The
assumption that $s_{1}^{e}=\min\left(BR^{-1}\left(s_{2}^{*}\right)\right)$
implies that $\lim_{s_{2}\nearrow s_{2}^{*}}\left(\phi_{s^{*}}\left(s_{2}\right)\right)=s_{2}^{e}$.
These observations imply that there exists a monotone biased belief
$\psi_{1}^{*}$ satisfying (1) $\psi_{1}^{*}\left(s_{2}\right)=s_{1}^{e}$
and (2) $\psi_{1}^{*}\left(s_{2}\right)\leq\phi_{s^{*}}\left(s_{2}\right)$
for each $s_{2}<s_{2}^{*}$ with an equality only if $\phi_{s^{*}}\left(s_{2}\right)=\min\left(S_{2}\right)$.

Let $s_{2}^{e}=\max\left(BR^{-1}\left(s_{1}^{*}\right)\right)$.\textbf{
}For each strategy $s_{1}>s_{1}^{*}$ define $X\left(s_{1}\right)$
as the set of strategies $s'_{1}\in S_{1}$ for which player $1$
is worse off (relative to $\pi_{2}\left(s_{1}^{*},s_{2}^{*}\right)$)
if he plays strategy $s_{1}$, while player $2$ plays a best-reply
to $s'_{1}$. Formally: 
\[
X_{s^{*}}\left(s_{2}\right)=\left\{ s'_{1}\in S_{1}|\pi_{2}\left(s_{1},s_{2}\right)\leq\pi_{2}\left(s_{1}^{*},s_{2}^{*}\right)\,\,\forall s_{1}\in BR\left(s'_{2}\right)\right\} .
\]
The assumption that $\pi_{1}\left(s_{1}^{*},s_{2}^{*}\right)>\tilde{M}_{1}^{U}$
implies that $X_{s^{*}}\left(s_{1}\right)$ is nonempty for each $s_{1}\in S_{1}$.
The assumption of strategic substitutes of player 2's payoff implies
that $X_{s^{*}}\left(s_{1}\right)$ is an interval ending at $\max\left(S_{1}\right)$.
Let $\phi_{s^{*}}\left(s_{1}\right)=\textrm{inf}\left(X_{s^{*}}\left(s_{1}\right)\right)$.
The assumption that the payoff function is continuously twice differentiable
implies that $\phi_{s^{*}}\left(s_{1}\right)$ is continuous. The
assumption that $s_{2}^{e}=\max\left(BR^{-1}\left(s_{1}^{*}\right)\right)$
implies that $\lim_{s_{1}\searrow s_{1}^{*}}\left(\phi_{s^{*}}\left(s_{1}\right)\right)=s_{1}^{e}$.
These observations imply that there exists a monotone biased belief
$\psi_{2}^{*}$ satisfying (1) $\psi_{2}^{*}\left(s_{1}\right)=s_{1}^{e}$
and (2) $\psi_{2}^{*}\left(s_{1}\right)\geq\phi_{s^{*}}\left(s_{1}\right)$
for each $s_{1}>s_{1}^{*}$ with an equality only if $\phi_{s^{*}}\left(s_{1}\right)=\max\left(S_{1}\right)$. 

We now show that these properties of $\left(\psi_{1}^{*},\psi_{2}^{*}\right)$
imply that $\left(\left(\psi_{1}^{*},\psi_{2}^{*}\right),\left(s_{1}^{*},s_{2}^{*}\right)\right)$
is a BBE. Consider a deviation of player $2$ into an arbitrary biased
belief $\psi'_{2}$. For each $s'_{2}\geq s_{2}^{*}$, and each $\left(s'_{1},s_{2}'\right)\in PNE\left(G_{\left(\psi_{1}^{*},\psi'_{2}\right)}\right)$,
the fact that $\psi_{1}^{*}\left(s'_{2}\right)=\psi_{1}^{*}\left(s_{2}^{*}\right)$
implies that $s_{1}'=s_{1}^{*}$, and due to assumption (III) of the
overinvestment of player 2 and the concavity of the payoff function,
we have $\pi_{2}\left(s'_{1},s'_{2}\right)=\pi_{2}\left(s'_{1},s{}_{2}^{*}\right)\leq\pi_{2}\left(s_{1}^{*},s{}_{2}^{*}\right)$.
For each $s'_{2}<s_{2}^{*}$, and each $\left(s'_{1},s_{2}'\right)\in NE\left(G_{\left(\psi_{1}^{*},\psi'_{2}\right)}\right)$,
the fact that $\psi_{1}^{*}\left(s'_{2}\right)\leq\phi_{s^{*}}\left(s'_{2}\right)$
with an equality only if $\phi_{s^{*}}\left(s'_{2}\right)=\min\left(S_{2}\right)$
implies that $\pi_{2}\left(s'_{1},s'_{2}\right)\leq\pi_{2}\left(s_{1}^{*},s_{2}^{*}\right)$.
This shows that player $2$ cannot gain from his deviation.

Finally, consider a deviation of player $1$ to an arbitrary biased
belief $\psi'_{1}$. For each $s'_{1}\leq s_{1}^{*}$, and each $\left(s'_{1},s_{2}'\right)\in PNE\left(G_{\left(\psi'_{1},\psi_{2}^{*}\right)}\right)$,
the fact that $\psi_{2}^{*}\left(s'_{1}\right)=\psi_{2}^{*}\left(s_{1}^{*}\right)$
implies that $s_{2}'=s_{2}^{*}$, and due to assumption (III) of the
underinvestment of player 1 and the concavity of the payoff function,
we have $\pi_{1}\left(s'_{1},s'_{1}\right)=\pi_{1}\left(s'_{1},s{}_{3}^{*}\right)\leq\pi_{1}\left(s_{1}^{*},s{}_{3}^{*}\right)$.
For each $s'_{1}>s_{1}^{*}$, and each $\left(s'_{1},s_{3}'\right)\in NE\left(G_{\left(\psi'_{1},\psi_{2}^{*}\right)}\right)$,
the fact that $\psi_{2}^{*}\left(s'_{1}\right)\geq\phi_{s^{*}}\left(s'_{1}\right)$
with an equality only if $\phi_{s^{*}}\left(s_{1}\right)=\max\left(S_{1}\right)$
implies that $\pi_{1}\left(s'_{1},s'_{2}\right)\leq\pi_{1}\left(s_{1}^{*},s_{2}^{*}\right)$.
This shows that player $1$ cannot gain from his deviation, which
implies that $\left(\left(\psi_{1}^{*},\psi_{2}^{*}\right),\left(s_{1}^{*},s_{2}^{*}\right)\right)$
is a BBE. 

\subsection{\label{subsec:Proof-of-Corollary-opposites}Proof of Corollary \ref{cor-pessimism-withstrategic-opposites}
(Pessimism in Games with Strategic Opposites)}

The proof is analogous to the proof of Corollary \ref{cor-wishful-tiniking-with-complements},
and is presented for completeness.

Assume to the contrary that $\psi_{i}^{*}\left(s_{j}^{*}\right)>s_{j}^{*}$
for some player $i$. Assume first that $\psi_{2}^{*}\left(s_{1}^{*}\right)>s_{1}^{*}$;
then Lemma \ref{lemma-wishful-substitute} implies that $\max\left(BR\left(\psi_{2}^{*}\left(s_{1}^{*}\right)\right)\right)\leq\min\left(BR\left(s_{1}^{*}\right)\right)$
with an equality only if 
\[
\max\left(BR\left(\psi_{2}^{*}\left(s_{1}^{*}\right)\right)\right)\in\left\{ \min\left(S_{2}\right),\max\left(S_{2}\right)\right\} .
\]
Part 1 of Proposition \ref{prop-strategic-opposites} and the definition
of a monotone BBE imply that 
\[
\max\left(BR\left(\psi_{2}^{*}\left(s_{1}^{*}\right)\right)\right)\geq s_{2}^{*}\geq\min\left(BR\left(s_{1}^{*}\right)\right).
\]
 The previous inequalities jointly imply that 
\[
\max\left(BR\left(\psi_{2}^{*}\left(s_{1}^{*}\right)\right)\right)=s_{2}^{*}=\min\left(BR\left(s_{1}^{*}\right)\right)\in\left\{ \min\left(S_{2}\right),\max\left(S_{2}\right)\right\} ,
\]
 which contradicts the assumption that $s_{2}^{*}\notin\left\{ \min\left(S_{2}\right),\max\left(S_{2}\right)\right\} $.

We are left with the case of $\psi_{1}^{*}\left(s_{2}^{*}\right)>s_{2}^{*}$;
then Lemma \ref{lemma-wishful-somplements} implies that $\min\left(BR\left(\psi_{1}^{*}\left(s_{2}^{*}\right)\right)\right)\geq\max\left(BR\left(s_{2}^{*}\right)\right)$
with an equality only if 
\[
\min\left(BR\left(\psi_{1}^{*}\left(s_{2}^{*}\right)\right)\right)\in\left\{ \min\left(S_{1}\right),\max\left(S_{1}\right)\right\} .
\]
Part 1 of Proposition \ref{prop-strategic-opposites} and the definition
of a monotone BBE imply that 
\[
\min\left(BR\left(\psi_{1}^{*}\left(s_{2}^{*}\right)\right)\right)\leq s_{1}^{*}\leq\max\left(BR\left(s_{2}^{*}\right)\right).
\]
 The previous inequalities jointly imply that 
\[
\min\left(BR\left(\psi_{1}^{*}\left(s_{2}^{*}\right)\right)\right)=s_{1}^{*}=\max\left(BR\left(s_{2}^{*}\right)\right).\in\left\{ \min\left(S_{1}\right),\max\left(S_{1}\right)\right\} ,
\]
 which contradicts the assumption that $s_{1}^{*}\notin\left\{ \min\left(S_{1}\right),\max\left(S_{1}\right)\right\} $.

\subsection{\label{subsec:Proof-of-Proposition-interval}Proof of Proposition
\ref{pro-interval-strong-continous-folk}}

Recall that we assume the payoff function $\pi_{i}$ to be continuously
twice differentiable. This implies that $\pi_{i}$ is Lipschitz continuous.
Let $K_{i}>0$ be the Lipschitz constant of the payoff function $\pi_{i}$
with respect to its first parameter, i.e., $K_{i}$ satisfies 
\[
\left\Vert \pi_{i}\left(s_{1},s_{2}\right)-\pi_{i}\left(s'_{1},s_{2}\right)\right\Vert \leq K_{i}\cdot\left\Vert s_{1}-s'_{1}\right\Vert .
\]
Assume that $\left(s_{1}^{*},s_{2}^{*}\right)$ is undominated and
$\pi_{i}\left(s_{1}^{*},s_{2}^{*}\right)>M_{i}^{U}$ for each player
$i$. Let $0<D_{i}=\pi_{i}\left(s_{1}^{*},s_{2}^{*}\right)-M_{i}^{U}$.
For each player $j$, let $s_{j}^{p}$ be an undominated strategy
that guarantees that player $i$ obtains, at most, his minmax payoff
$M_{i}^{U}$, i.e., $s_{j}^{p}=\textrm{argmin}_{s_{j}\in S_{j}^{U}}\left(\max_{s_{i}\in S_{i}}\pi_{i}\left(s_{i},s_{j}\right)\right).$
The strict concavity of $\pi_{i}\left(s_{i},s_{j}\right)$ with respect
to $s_{i}$ implies that the best-reply correspondence is a continuous
one-to-one function. Thus, $BR^{-1}\left(s_{i}\right)$ is a singleton
for each $s_{i}$, and we identify $BR^{-1}\left(s_{i}\right)$ with
the unique element in this singleton set. 

Let $\epsilon>0$ be a sufficiently small number satisfying $\epsilon<\min\left(\frac{D_{i}}{K_{i}},\frac{D_{j}}{K_{j}}\right)$.
For each $\delta\in\left[0,1\right]$ define for each player $i$:
\[
s_{i}^{\delta}=\frac{\epsilon-\delta}{\epsilon}\cdot s_{i}^{*}+\frac{\delta}{\epsilon}\cdot s_{i}^{p}.
\]

Let $\psi_{i}^{\epsilon}$ be defined as follows: 
\[
\psi_{i}^{\epsilon}\left(s_{j}'\right)=\begin{cases}
BR^{-1}\left(s_{i}^{\left|s'_{j}-s_{j}\right|}\right) & \left|s'_{j}-s_{j}\right|<\epsilon\\
BR^{-1}\left(s_{i}^{p}\right) & \left|s'_{j}-s_{j}\right|\geq\epsilon.
\end{cases}
\]
Note that $\psi_{i}^{\epsilon}$ is continuous. We now show that $\left(\left(\psi_{1}^{\epsilon},\psi_{2}^{\epsilon}\right),\left(s_{1}^{*},s_{2}^{*}\right)\right)$
is a strong BBE. Observe first that the definition of $\left(\psi_{1}^{\epsilon},\psi_{2}^{\epsilon}\right)$
immediately implies that $\left(s_{1}^{*},s_{2}^{*}\right)\in NE\left(G_{\left(\psi_{1}^{\epsilon},\psi_{2}^{\epsilon}\right)}\right)$.
Next, consider a deviation of player \emph{i} to an arbitrary biased
belief $\psi'_{i}$. Consider any equilibrium of the new biased game
$\left(s'_{i},s'_{j}\right)\in NE\left(G_{\left(\psi'_{i},\psi_{j}^{\epsilon}\right)}\right)$.
If $\left|s'_{i}-s_{i}\right|\geq\epsilon$, then the definition of
$\psi_{j}^{\epsilon}\left(s_{i}'\right)$ implies that $s_{j}^{p}=s'_{j}$,
and that player $i$ achieves a payoff of at most $M_{i}^{U}<\pi_{i}\left(s_{1}^{*},s_{2}^{*}\right)$.
If $s'_{i}=s_{i}^{*}$, then it is immediate that $s'_{j}=s_{j}^{*}$
and that player $i$ does not gain from his deviation. If $0<\left|s'_{i}-s_{i}\right|<\epsilon$,
then the definition of $\psi_{j}^{\epsilon}\left(s_{i}'\right)$ implies
that 
\[
\pi_{i}\left(s'_{i},s'_{j}\right)=\pi_{i}\left(s'_{i},s_{j}^{\left|s'_{i}-s_{i}\right|}\right)=\pi_{i}\left(s'_{i},\frac{\epsilon-\left|s'_{i}-s_{i}\right|}{\epsilon}\cdot s_{j}^{*}+\frac{\left|s'_{i}-s_{i}\right|}{\epsilon}\cdot s_{j}^{p}\right)\leq
\]
\[
\frac{\epsilon-\left|s'_{i}-s_{i}\right|}{\epsilon}\cdot\pi_{i}\left(s'_{i},s_{j}^{*}\right)+\frac{\left|s'_{i}-s_{i}\right|}{\epsilon}\cdot\pi_{i}\left(s'_{i},s_{j}^{p}\right)\leq\frac{\epsilon-\left|s'_{i}-s_{i}\right|}{\epsilon}\cdot\pi_{i}\left(s'_{i},s_{j}^{*}\right)+\frac{\left|s'_{i}-s_{i}\right|}{\epsilon}\cdot M_{i}^{U}\leq
\]

\[
\frac{\epsilon-\left|s'_{i}-s_{i}\right|}{\epsilon}\cdot\pi_{i}\left(s^{*}{}_{i},s_{j}^{*}\right)+K_{i}\cdot\left|s'_{i}-s_{i}\right|+\frac{\left|s'_{i}-s_{i}\right|}{\epsilon}\cdot M_{i}^{U}=
\]
\[
\frac{\epsilon-\left|s'_{i}-s_{i}\right|}{\epsilon}\cdot\pi_{i}\left(s^{*}{}_{i},s_{j}^{*}\right)+K_{i}\cdot\left|s'_{i}-s_{i}\right|+\frac{\left|s'_{i}-s_{i}\right|}{\epsilon}\cdot\left(\pi_{i}\left(s^{*}{}_{i},s_{j}^{*}\right)-D_{i}\right)=
\]
\[
\pi_{i}\left(s^{*}{}_{i},s_{j}^{*}\right)+\frac{\epsilon-\left|s'_{i}-s_{i}\right|}{\epsilon}\cdot K_{i}\cdot\left|s'_{i}-s_{i}\right|-\frac{\left|s'_{i}-s_{i}\right|}{\epsilon}\cdot D_{i}\leq\pi_{i}\left(s^{*}{}_{i},s_{j}^{*}\right)+K_{i}\cdot\left|s'_{i}-s_{i}\right|-\frac{\left|s'_{i}-s_{i}\right|}{\epsilon}\cdot D_{i}=
\]
\[
\pi_{i}\left(s^{*}{}_{i},s_{j}^{*}\right)+\left|s'_{i}-s_{i}\right|\cdot\left(K_{i}-\frac{D_{i}}{\epsilon}\right)<\pi_{i}\left(s^{*}{}_{i},s_{j}^{*}\right),
\]
where the first inequality is due to the convexity of $\pi_{i}\left(s_{i},s_{j}\right)$
with respect to $s_{j}$, the second inequality is due to $\pi_{i}\left(s'_{i},s_{j}^{p}\right)\leq M_{i}^{U}$,
the third inequality is due to the Lipschitz continuity, the penultimate
inequality is implied by $\frac{\epsilon-\left|s'_{i}-s_{i}\right|}{\epsilon}<1$,
and the last inequality is due to defining $\epsilon$ to satisfy
$\epsilon<\min\left(\frac{D_{i}}{K_{i}},\frac{D_{j}}{K_{j}}\right)$.
This proves that player $i$ cannot gain from his deviation, and that
$\left(\left(\psi_{1}^{\epsilon},\psi_{2}^{\epsilon}\right),\left(s_{1},s_{2}\right)\right)$
is a strong BBE.
\end{document}